\documentclass[runningheads]{llncs}
\usepackage[T1]{fontenc}
\usepackage{graphicx}
\usepackage[numbers]{natbib}

\usepackage{hyperref}
\usepackage{amsmath}
\usepackage{amssymb}
\usepackage{mathtools}
\usepackage{xcolor}
\usepackage{xspace}
\usepackage{stmaryrd}
\usepackage[altpo]{backnaur}
\usepackage{mathpartir}
\usepackage{bbm}
\usepackage{mathtools}
\usepackage{cleveref}
\usepackage{caption}
\usepackage{subcaption}

 \renewcommand{\cref}[1]{\lcnamecref{#1}~\labelcref{#1}}
  \makeatletter
      \def\lcfirstnamecrefs#1,#2\@nil{\lcnamecrefs{#1}}
      \newcommand{\lcfirstnamecref}[1]{\lcfirstnamecrefs #1,\@nil}
  \makeatother

\usepackage{thmtools}
\usepackage{thm-restate}
\usepackage[shortlabels]{enumitem}
\usepackage[left=43.9mm,right=43.9mm,top=52.1mm,bottom=52.1mm]{geometry}
\usepackage{booktabs} % For formal tables
\usepackage{siunitx} % Provides the "S" column type

\hypersetup{
    colorlinks=true,
    linkcolor=blue,
}

\newcommand{\noz}[1]{{#1}^{\setminus 0}}
\newcommand{\abs}[1]{\left| {#1} \right|}
\newcommand{\exps}{\mathcal{E}}
\newcommand{\vals}{\mathcal{V}}
\newcommand{\sort}{\mathtt{S}}

\newcommand{\con}[3]{#1 \xrightarrow{#2} #3}
\newcommand{\psatl}[1]{\,\widetilde{\vDash}^{#1}\,}
\newcommand{\psat}{\psatl{ }}

\newcommand{\ut}{\hat{\tau}}
\newcommand{\lint}{\tau}
\newcommand{\rd}[1]{\overline{#1}}
\newcommand{\unrd}[1]{\lvert #1 \rvert}

\newcommand{\maxdens}[2]{\max {#1}/{#2}}
\newcommand{\id}{\text{Id}}

\newcommand{\fsub}[2]{\{#2/#1\}}
\newcommand{\ssub}[3]{{#1}|_{#2}^{#3}}
\newcommand{\sub}[2]{\ssub{S}{#1}{#2}}
\newcommand{\cn}{c}
\newcommand{\ilist}[1]{}
\newcommand{\tleft}{\mathtt{l}}
\newcommand{\tright}{\mathtt{r}}
\newcommand{\caketerm}{\mathtt{cake}}

\newcommand{\pieces}{\mathbb{P}}

\newcommand{\agents}{\mathbb{A}}

\newcommand{\eval}{\Downarrow}

\newcommand{\pnt}{\#\mathsf{Pt}}

\newcommand{\kw}[1]{\mathsf{#1}}

\newcommand{\vset}{\overline{V}}
\newcommand{\uset}{\overline{U}}

    \newcommand{\rname}[1]{[\textsc{#1}]}

 \DeclareMathOperator{\dom}{\text{dom}}
 \DeclareMathOperator{\interp}{\mathcal{A}}
  \DeclareMathOperator{\assn}{\mu}
  \newcommand{\assnd}[1]{\{ {\#s} \mapsto D\}}
\newcommand{\sem}[1]{\llbracket {#1} \rrbracket}
\DeclareMathOperator{\val}{\mathcal{V}}

\newcommand{\simpl}{\mathtt{R}}

\renewcommand{\val}{\mathtt{V}}
\renewcommand{\v}{v}
\newcommand{\lv}{\mathtt{v}}

    \newcommand{\SYSTEM}{\textsc{Slice}\xspace}
    \newcommand{\kwtrue}{\kw{true}}
    \newcommand{\kwfalse}{\kw{false}}
    \newcommand{\kwlet}{\kw{let}}
    \newcommand{\kwin}{\kw{in}}
    \newcommand{\kwif}{\kw{if}}
    \newcommand{\kwthen}{\kw{then}}
    \newcommand{\kwelse}{\kw{else}}

    \newcommand{\kwsplit}{\kw{split}}
    \newcommand{\kwdiv}{\kw{divide}}
    \newcommand{\kwmark}{\kw{mark}}
    
    \newcommand{\kweval}{\kw{eval}}
    
    \newcommand{\kwcake}{\kw{cake}}
    \newcommand{\kwinit}{\kwcake}

    \newcommand{\kwpiece}{\kw{piece}}

    \newcommand{\kwassert}{\kw{assert}}

    \newcommand{\fv}{\mathtt{FV}}

    \newcommand{\var}{\mathcal{X}}
    \newcommand{\lvar}{\mathcal{W}}
    \newcommand{\yctx}{\Upsilon}
    \newcommand{\pctx}{\mathcal{E}}

    \newcommand{\pt}[2]{\kwassert\ {#1}\ \kwin\ {#2}}
    \newcommand{\pf}[2]{\kwassert\ \neg{#1}\ \kwin\ {#2}}

    \newcommand{\ediv}[2]{\kwdiv({#1}, {#2})}
    \newcommand{\emark}[3]{\kwmark_{{#1}}({#2}, {#3})}
    
    \newcommand{\eeval}[2]{\kweval_{{#1}}({#2})}
    \newcommand{\epiece}[1]{\kwpiece({#1})}

\newcommand{\conj}{\psi}

\newif\iflong
\longtrue

\usepackage{color}

\usepackage{graphicx} % necessary for inserting .pdfs
\usepackage{hyperref} % "available" badge is a link to the actual DOI
\usepackage[firstpage]{draftwatermark} % free badge placement

\SetWatermarkAngle{0}

\SetWatermarkText{\raisebox{15.5cm}{%
\hspace{0.0cm}%
\href{https://doi.org/10.5281/zenodo.10947124}{\includegraphics{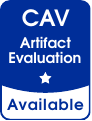}}%
\hspace{9cm}%
\includegraphics{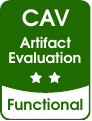}%
}}

\newcommand{\append}{\iflong
the appendix\else
the full paper~\citep{fullversion}\fi
}
\begin{document}
\title{Verifying Cake-Cutting, Faster}
\author{Noah Bertram, Tean Lai, Justin Hsu }%\inst{1}\orcidID{0000-1111-2222-3333} \and
%Second Author\inst{2,3}\orcidID{1111-2222-3333-4444} \and
%Third Author\inst{3}\orcidID{2222--3333-4444-5555}}
\institute{Cornell University}
\maketitle              % typeset the header of the contribution
\begin{abstract}
    Envy-free cake-cutting protocols procedurally divide an infinitely divisible
    good among a set of agents so that no agent prefers another's allocation to
    their own. These protocols are highly complex and difficult to prove
    correct.  Recently, Bertram, Levinson, and Hsu introduced a
    language called Slice for describing and verifying cake-cutting protocols.
    Slice programs can be translated to formulas encoding envy-freeness, which
    are solved by SMT. While Slice works well on smaller protocols, it has
    difficulty scaling to more complex cake-cutting protocols.

    We improve Slice in two ways. First, we show any protocol execution
    in Slice can be replicated using piecewise uniform valuations. We then
    reduce Slice's constraint formulas to formulas within the theory of linear
    real arithmetic, showing that verifying envy-freeness is efficiently decidable. Second, we design and implement a linear type system
    which enforces that no two agents receive the same part of the
    good. We implement our methods and verify a range of challenging examples,
    including the first nontrivial four-agent protocol.

\end{abstract}
\section{Introduction}
\label{sec:intro}
How would you divide a piece of cake between two children? Classic wisdom would say to have one child cut the piece into two, and have the other take their preferred slice. Procedures that divide an infinitely divisible good amongst a set of agents are called \emph{cake-cutting protocols}. If the protocol ensures no agent prefers what another received, it is called \emph{envy-free}. While classic wisdom gives an envy-free protocol for two agents, a three-agent envy-free protocol was not discovered until 1960, and a four-agent envy-free protocol that does not dispose any cake was only proposed in 2015 by \citet{aziz4agents}. Modern cake-cutting protocols are highly complex, and proving envy-freeness requires checking an enourmous number of cases.

\paragraph{Verifying envy-freeness.}
To make cake-cutting protocols easier to verify, \citet{slice} introduced a
language called Slice that can describe cake-cutting protocols, and encode
envy-freeness as a logical formula that can be dispatched to an SMT solver.
While Slice can verify envy-freeness fully automatically, it has some drawbacks.
First, it is not able to verify that agents receive non-overlapping pieces. This
basic property, known as \emph{disjointness}, is crucial for correctness.
% Conceivably, the constraint approach could be applied to verify disjointness, though it is quite a heavy-handed solution to verify such a basic property.

Another drawback of Slice is the SMT instances encoding envy-freeness for complicated protocols are difficult to solve, and only scale to some three-agent protocols---non-trivial algorithms, but relatively simple compared to modern cake-cutting protocols.
One reason the instances are difficult is they are \emph{higher order}: they quantify over \emph{valuations}, which are functions that describe the agents' preferences.
%These valuations appear throughout Slice's constraint formulas, and
%It is not clear whether the verification problem is decidable.

\paragraph{Our work: verifying disjointness and envy-freeness, faster.}
We address these weaknesses in Slice. To verify disjointness, we develop an affine type system for Slice which restricts usage of the cake and then prove that well-typed programs are disjoint.
%Our type soundness theorem ensures that well-typed programs are disjoint.
Typechecking is straightforward and syntax-directed, requiring no use of SMT.

%To verify envy-freeness faster, we first show that any Slice protocol execution can be replicated using \emph{piecewise uniform valuations}, a result inspired by \citet{Kurokawa_Lai_Procaccia_2013}. For these valuations, we can represent the value of a piece as a term in linear real arithmetic. This enables us to reduce Slice constraint formulas to linear real arithmetic formulas and avoid quantifying over valuations when verifying envy-freeness. As a side benefit, our work shows that verifying envy-freeness of Slice protocols is decidable.

To verify envy-freeness more efficiently, we reduce Slice constraints into
linear real arithmetic formulas, removing the need to quantify over valuations.
This reduction leverages a key observation:
%inspired by \citet{Kurokawa_Lai_Procaccia_2013},
the behavior of a protocol on any valuation can be replicated by a
piecewise uniform valuation, which enables envy-freeness to be encoded as a
first-order formula in linear real arithmetic. As a side benefit, our work shows that verifying envy-freeness of Slice protocols is decidable.

%We develop a reduction of Slice constraint formulas into the theory of linear real arithmetic for piecewise uniform valuations and prove that it preserves validity. Since any protocol execution can be replicated with piecewise uniform valuations, the reduction enables us to verify envy-freeness without quantifying over valuations. As a side benefit, our work shows that verifying envy-freeness of Slice protocols is decidable.

Finally, we implement both our affine type system and formula reduction
procedure on top of the Slice implementation and transcribe two significantly
more complicated protocols into Slice, including the first nontrivial four agent
cake-cutting protocol~\citep{wmh}. For all Slice protocols, our type system
establishes disjointness and our constraints encoding envy-freeness can be
verified in substantially less time than in the previous version of Slice.

\paragraph{Outline.}
After
describing the cake-cutting model (\Cref{sec:model}), we present the Slice language and our new linear type system for verifying disjointness (\Cref{sec:lang}). We then review Slice's constraints (\Cref{sec:constraints}) and describe our
%piecewise uniform
new
constraint translation (\Cref{sec:pu}). We discuss our implementation and evaluation (\Cref{sec:eval}), and then conclude with related work and future directions (\Cref{sec:rw}).

\section{Cake-Cutting Preliminaries}
\label{sec:model}

In this section, we introduce the basics of cake-cutting protocols; the reader
can consult a standard text for more background~\citep{Procaccia_2016}.

We begin by fixing a finite set of agents $\agents$.
The cake or good is modeled by the unit interval $[0,1]$. A \emph{piece} $P$ is
a finite union of intervals from the cake:  $P = [r_1, r'_1]\cup \cdots \cup
[r_{n},r'_{n}]$ where $r_{1} \leq r'_{1} < r_2 \leq r'_{2} < \cdots < r_{n} \leq
r'_{n}$; the points $r_i$ and $r_i'$ are \emph{boundary points} of $P$, and we
write $\partial P$ for the set of all boundary points.  Two pieces
$P_1$ and $P_2$ are \emph{disjoint} if $(P_1\setminus \partial P_1) \cap (P_2
\setminus \partial P_2) = \emptyset$, that is, $P_1$ and $P_2$ only share possibly their boundary points.

Cake-cutting protocols produce an \emph{allocation} of pieces to agents, i.e.,
an $\agents$-tuple of pieces $(P_{a}\mid P_{a} \in \mathbb{P},a\in \agents)$.
%This allocation will depend on what agents prefer, and it is standard to model an agent's preferences by a \emph{valuation} function $V : \pieces \to [0, 1]$.
Protocols produce allocations based on agent preferences, which are typically modelled by functions $V : \pieces \to [0,1]$  called \emph{valuations}.
We assume that valuations satisfy five standard assumptions: (1) Additivity: $V(P\cup P') =
V(P) + V(P')$ provided $P$ and $P'$ are disjoint; (2) Non-negativity: $V(P) \geq
0$; (3) Continuity: $V([r,r'])$ is continuous in both $r$ and $r'$; (4)
Monotonicity: $V(P)\geq V(P')$ if $P' \subseteq P$; and (5) Normalization:
$V([0,1]) = 1$. We will often write $V[r,r']$ for $V([r,r'])$.
% This set of properties is not minimal. It can be shown that the first three properties imply monotonicity.
A \emph{valuation set} $\vset$ is an $\agents$-tuple of valuations $(V_{a} \mid a \in \agents)$.
We write $\vset_{a}$ for agent $a$'s valuation.

Cake-cutting protocols aim to produce fair allocations where no agent prefers
another agent's piece. More precisely, if $A$ is an allocation and $\vset$ is a valuation set, we say $A$ is
\emph{envy-free} (with respect to $\vset$) if $\vset_{a}(A_{a}) \geq
\vset_{a}(A_{a'})$ for all $a, a'\in \agents$.

% Stromquist \cite{strom1} shows in general, that envy-free (connected) allocations always exist under assumptions on the valuations more mild than ours.
Protocols are assumed to have indirect access to agent valuations through specific kinds of \emph{agent queries}.
Slice implements the Robertson-Webb (RW) query model \citep{robweb}, which is
the typical query model in the cake-cutting literature and captures most protocols.
In the RW model, there are two kinds of queries.
An \emph{eval query} takes as input an agent and a piece
and reports the agent's value of that piece:
\[
    \kweval_{a}(P)\ \text{reports}\ \vset_{a}(P).
\]
A \emph{mark query}, when supplied an interval and a value, reports how much of the interval is needed to attain that value:
\[\kwmark_{a}([\ell,r],v)\ \text{reports}\ r'\ \text{where}\ \vset_{a}[\ell,r']
= v,\ \text{provided that}\ v \leq \vset _{a}[\ell,r].\] This query enables us
to find intervals within the cake which have a specified value for a certain
agent.  For example, $\emark{a}{[0,1]}{1/2}$ will output a point $r'$ such that
$\vset_{a}[0,r'] = 1/2 = \vset_{a}[r', 1]$.  The assumption $v \leq
\vset_{a}[\ell,r]$ is required since $\vset_{a}$ is monotone: if $v >
\vset_{a}[\ell,r]$, no such point exists. Note that if multiple points $r'$ are
a valid answer to a mark query, then mark can report any of them.

\section{Language and Type System}
\label{sec:lang}

We review the language~\citep{slice} before describing our novel affine type system. Full details for this section can be found in \append.

\subsection{Syntax of Base Slice}
The set of all basic Slice expressions $\exps$ is given by the grammar shown in \Cref{fig:grammar}.
%We walk through select syntax.
  The expression $\v$ is a value and $\var$
  is an infinite set of variables. We can form tuples and, through the
  $\kwsplit$ expression, extract their components. We have standard if-then-else
  expression, and a set $\mathcal{O}$ consisting of primitive operations like
  $+$,   $\geq$, etc.

  The remaining expressions are cake-cutting specific. The expression $\kwcake$ represents the whole cake, $\kwdiv$ takes an interval and a point, splitting the interval into two at the point, and $\kwpiece$ takes in a list of intervals and forms a piece out of them.
  The expression $\kweval_{a}$ implements the eval query by taking in an
  interval or piece, and producing its value according to agent $a$. The
  expression $\kwmark_{a}$ implements the mark query by taking in an interval
  and the target value, returning any point satisfying the query.

\begin{figure}
    \begin{subfigure}{\textwidth}

  \begin{bnf*}
    \bnftd{e} \bnfpo
    \bnfmore{%
        \v
      \bnfor x \in \var
      \bnfor (e_1, \ldots , e_{n})
      \bnfor \kwlet\ x_1,\ldots, x_{n} =\ \kwsplit\ e_1\ \kwin\ e_2
    }\\
    \bnfmore{%
      \bnfor \bnfts{$\kwif$}\ e_1\ \bnfts{$\kwthen$}\ e_2\ \bnfts{$\kwelse$}\ e_3
      \bnfor o(e_1, \dots, e_n) \qquad (o \in \mathcal{O})
    }\\
    \bnfmore{%
      \bnfor \bnfts{$\kwinit$}
      \bnfor \bnfts{$\ediv{e_1}{e_2}$}
      \bnfor \bnfts{$\epiece{e_1,\ldots ,e_{n}}$}
    }\\
    \bnfmore{%
      \bnfor \bnfts{$\emark{a}{e_1}{e_2}$}
      \bnfor \bnfts{$\eeval{a}{e}$} \qquad (a \in \agents)
    }
  \end{bnf*}
    \end{subfigure}
    \begin{subfigure}{\textwidth}
\begin{bnf*}
    \v \bnfpo
    \bnfmore{%
        \kwtrue
        \bnfor \kwfalse
        \bnfor r\pnt
        \bnfor \left[ r, r' \right]_{} \quad (r \leq r')
    } \\
    \bnfmore{
        \bnfor (\v_1, \ldots ,\v_{n})
        \bnfor P \left[ r_1, r_1' \right]_{} ,\ldots , \left[ r_n, r_n' \right]_{}
        \quad (r_i \leq r_{i}')
    } \\
    \bnfmore{
     \bnfor r_{1} \cdot V_{a_{1}}(P_{1}) + \cdots + r_{n} \cdot V_{a_{n}}(P_{n})
    }
\end{bnf*}
    \end{subfigure}
  \caption{The grammar for Slice expressions (top) and values (bottom).}
  \label{fig:grammar}
\end{figure}
  The set of all values is denoted by $\vals$\hypertarget{values}{}.
We have boolean constants, \emph{points} $r\pnt$, and \emph{intervals} $[r, r']$.
Points represent positions within the cake. Intervals describe contiguous pieces
of the cake. Tuple values enable us to describe allocations.

Values of the form $P [r_1, r'_{1}] ,\ldots ,[r_{n}, r'_{n}]$ and $r_{1} \cdot
V_{a_{1}}(P_{1}) + \cdots + r_{n} \cdot V_{a_{n}}(P_{n})$ are referred to as
\emph{pieces} and \emph{valuations}, respectively. Note that for piece values, we do not
assume that $[r_{i},r'_{i}]$ is disjoint from $[r_{j}, r'_{j}]$ if $i \neq j$.
We sometimes write piece values as $P_{i = 1}^{n} [r_{i},r'_{i}]$, or
$P_{i}[r_{i}, r'_{i}]$, where $i$ ranges over a finite set.  Within the
valuation value, $P_1,\ldots ,P_{n}$ are interval or piece values, $a_1,\ldots
,a_{n}$ are agents, and $r_1,\ldots ,r_{n}$ are real numbers.  We sometimes
write $\sum_{i = 1}^{n}r_{i} \cdot V_{a_{i}}(P_{i})$, or $\sum_{i}r_{i}
\cdot V_{a_{i}}(P_{i})$ for short.

\Cref{proto:cc} shows the two agent protocol described in \Cref{sec:intro}
implemented in Slice; for now, we can ignore the bars over variables.
This protocol uses the eval and mark queries to divide the cake into two pieces equally preferred by agent 1, and then uses eval queries for agent 2's comparison.

\begin{figure}
  \[
    \begin{array}{l}
        \kwlet\ p = \kwsplit\ \kwcake\ \kwin\\
        \kwlet\ p_1, p_2 = \kwsplit\ \ediv{p}{\emark{1}{\rd{p}}{1/2\cdot \eeval{1}{\rd{p}}}}\ \kwin\\
        \kwif\ \eeval{2}{\rd{p_1}} \geq \eeval{2}{\rd{p_2}}\ \kwthen \\
      \quad (\kwpiece (p_2), \kwpiece (p_1))    \\
      \kwelse           \\
      \quad (\kwpiece (p_1),\kwpiece (p_2))
    \end{array}
  \]
  \caption{Cut-choose in \SYSTEM.}
  \label{proto:cc}
\end{figure}

\subsection{A Linear Type System for Slice}
\label{sec:types}
In this section, we develop a new, affine type system for
Slice. At a high level, our type system ensures that no two
agents receive overlapping pieces in the allocation.
In order to accomplish this, it suffices to ensure that any duplicated variable bound to an interval or piece cannot be used either make further cuts or form more pieces.
After all, you can't have your cake
and eat it too!

\paragraph*{Types.}
\hypertarget{types}{}
Slice types include affine types $\lint$ and non-affine types $\ut$:
\begin{bnf*}
    \bnftd{$\ut$} \bnfpo
    \bnfmore{%
        \mathsf{Bool}
        \bnfor \mathsf{Point}
        \bnfor \mathsf{Vltn}
        \bnfor \rd{\mathsf{Intvl}}
        \bnfor \rd{\mathsf{Piece}}
    }\\
    \bnftd{$\lint$} \bnfpo
    \bnfmore{%
        \mathsf{Intvl}
        \bnfor \mathsf{Piece}
        \bnfor \ut_1 \times \cdots \times \ut_{n} \times \lint_1 \times \cdots \times \lint_{n}
    }
\end{bnf*}
Any non-linear type can be viewed as a linear type (i.e., as a unary product).

We treat $\mathsf{Intvl}$ and $\mathsf{Piece}$ as affine types to prevent their values from being duplicated. However, restricting interval and piece
types poses a problem: protocols often query an agent before using the same
interval or piece for division or allocation. For example, in the second line in \Cref{proto:cc}, $p$ needs to be used to mark itself appropriately before being divided. To address this issue, we include two
new base non-affine types, $\rd{\mathsf{Intvl}}$ and $\rd{\mathsf{Piece}}$, called
``read only'' types. Since these types are non-affine, variables of these types can be
freely used in queries. However, dividing or forming pieces from read-only types is not allowed. This restriction ensures we can only create disjoint pieces.

\paragraph*{Values and expressions.}
We extend Slice with values of read-only type:
\begin{bnf*}
    \v \bnfpo
    \bnfmore{%
        \cdots
        \bnfor \rd{[r, r']}
        \bnfor \rd{P \left[ r_1, r_1' \right]_{} ,\ldots , \left[ r_n, r_n' \right]}
    }
\end{bnf*}
The overline syntax is also extended to notation on other values and types, e.g. $\rd{r} = r$, $\rd{(v_1,v_2)} = (\rd{v_1},\rd{v_2})$, $\rd{\mathsf{Vltn}} = \mathsf{Vltn}$, and $\rd{\tau_1\times \tau_2} = \rd{\tau_1} \times \rd{\tau_2}$.
Next, we extend Slice expressions with two new classes of variables. Affine
variables are drawn from $\lvar$, while read-only variables are drawn from $\rd{\lvar}$.
Finally, we extend the syntax of the split expression to bind these variables:
\[
    \kwlet\ x_1,\ldots ,x_{n}, w_1,\ldots ,w_{n'} = \kwsplit\ e_1\ \kwin\ e_2
\]
This expression implicitly binds read-only variables $\rd{w_1} ,\ldots
,\rd{w_{n'}}$ corresponding to the affine variables $w_1 ,\ldots ,w_{n}$. For example, in
\Cref{proto:cc}, $\rd{p}$ is bound in the first line and both $\rd{p_1}$ and
$\rd{p_2}$ are bound in the second line.

\paragraph*{Affine typing rules.}
Our typing judgements are of the form $\Gamma; \Delta \vdash e : \tau$, for
$\Gamma$ a partial map from $\var\cup\rd{\lvar}$ to non-affine types, and
$\Delta$ a partial map from $\lvar$ to linear types. We present a selection of
rules in \Cref{fig:trules}.
\begin{figure}
\begin{mathpar} \scriptstyle
    \mprset {sep = 1em}
    \inferrule*[right=T-Var]
    { }
    {
        \Gamma, x : \ut; \Delta \vdash x : \ut
    }
    \and
\inferrule*[right=T-AffVar]
    { }
    {
        \Gamma; w : \lint \vdash w : \lint
    }
    \and
    \inferrule*[right=T-Piece]
    {%
        \Gamma;\Delta_1 \vdash e_1 : \mathsf{Intvl} \\
    \cdots \\
    \Gamma;\Delta_{n} \vdash e_n : \mathsf{Intvl}
    }
    {
        \Gamma ; \Delta_1,\ldots ,\Delta_{n}\vdash \kwpiece (e_1, \ldots ,e_n) : \mathsf{Piece}
    }
    \and
 \inferrule*[right=T-Div]
    {%
      \Gamma;\Delta_1 \vdash e_1 : \mathsf{Intvl} \\
      \Gamma;\Delta_2 \vdash e_2 : \mathsf{Point} \\
    }
    {\Gamma;\Delta_1,\Delta_2 \vdash \ediv{e_1}{e_2} : \mathsf{Intvl} \times \mathsf{Intvl}}
    \and
    \inferrule*[right=T-Mark]
    {%
        \Gamma; \Delta_1 \vdash e_1 : \rd{\mathsf{Intvl}} \\
        \Gamma; \Delta_2 \vdash e_2 : \mathsf{Vltn} \\
    }
    {%
    \Gamma; \Delta_1, \Delta_2 \vdash \emark{a}{e_1}{e_2} : \mathsf{Point}
    }
    \and
    \inferrule*[right=T-EvalPc]
    {%
        \Gamma; \Delta \vdash e : \rd{\mathsf{Piece}}
    }
    {\Gamma; \Delta \vdash \eeval{a}{e} : \mathsf{Vltn}}
    \and
 \inferrule*[right=T-Split]
    {%
      \Gamma; \Delta_1  \vdash e_1 : \ut_1\times \cdots \times \ut_{n} \times \lint_{1} \times \cdots \times \lint_{n'} \\
      \Gamma, x_1 : \ut_1,\ldots ,x_{n} : \ut_{n}, \rd{w_1} : \rd{\lint_1} ,\ldots ,\rd{w_{n'}} : \rd{\lint_{n'}};\Delta_2,w_1 : \lint_1 ,\ldots ,w_{n'} : \lint_{n'} \vdash e_2 : \tau
    }
    {\Gamma;\Delta_1 , \Delta_2 \vdash \kwlet\ x_1,\ldots ,x_{n}, w_1,\ldots ,w_{n'} = \kwsplit\ e_1\ \kwin\ e_2 : \tau}
\end{mathpar}
\caption{Select typing rules for Slice expressions.}
\label{fig:trules}
\end{figure}
For affine type contexts $\Delta_1$ through $\Delta_{n}$, the concatenation
$\Delta_1,\ldots ,\Delta_{n}$ denotes the union of \emph{disjoint} contexts: $\dom(\Delta_{i})\cap \dom(\Delta_{j}) = \emptyset$ if $i \neq j$.

The variable rules \rname{T-Var} and \rname{T-AffVar} type the given variable based on its context.
The rule \rname{T-Piece} shows the role of affine variables. The premise has
expressions $e_{i}$ under linear type contexts $\Delta_{i}$, while the
conclusion has the combined affine type context $\Delta_1,\ldots ,\Delta_{n}$.
Since the $\Delta_i$ must have disjoint domain, the expressions $e_i$ cannot
share affine variables. In particular, it is not possible for the same interval variable to
appear more than once in a piece.

The rules \rname{T-Piece}, \rname{T-Div}, \rname{T-Mark}, and \rname{T-EvalPc}
highlight the difference between affine types and their read-only variants.
\rname{T-Piece} requires its subexpressions have type $\mathsf{Intvl}$ as it is
forming a piece. \rname{T-Div} requires the first argument have type
$\mathsf{Intvl}$ since it is forming new pieces. In contrast, \rname{T-Mark} requires
the first argument to have type $\rd{\mathsf{Intvl}}$ as it is querying a
valuation, not forming a piece, and similarly for \rname{T-EvalPc}.

The most complicated rule is \rname{T-Split}, which binds multiple variables at once by pattern matching on tuples.
\subsection{Semantics}
\label{sec:sem}
We present a big-step style semantics, defined by a relation ${\eval}_{\vset}
\subseteq \exps \times \vals$ indexed by a valuation set $\vset$, so
our judgments are of the form $e \eval_{\vset} \v$. We omit $\vset$ when clear from context. Our big-step rules are straightforward.
We present a few rules in
\Cref{fig:srules} and discuss them here.

\begin{figure}
\begin{mathpar} \scriptstyle
    \inferrule*[right=E-Tup]
    {%
    e_1 \eval \v_1 \\
    \cdots \\
    e_n \eval \v_n
    }
    {(e_1, \ldots ,e_{n}) \eval (\v_1,\ldots ,\v_{n})}

    \inferrule*[right=E-Mark]
    {%
        e_1 \eval \rd{[r_1,r'_1]} \\
        e_2 \eval {\textstyle\sum_{i}r_{i}}V_{a_{i}}(P_{i}) \\
        V_a^{}([r_1,r]) = {\textstyle \sum_{i}r_{i}\cdot V_{a_{i}}(P_{i})}
    }
    {%
        \emark{a}{e_1}{e_2} \eval r
    }
    \and
    \inferrule*[right=E-EvalPc]
    {%
        e \eval \rd{P\ [r_1, r_1'] , \ldots, [r_n, r_n']}
    }
    {\eeval{a}{e} \eval V_a(P [r_1, r_{1}'],\ldots ,[r_n, r_{n}'])}
    \and
     \inferrule*[right=E-Div]
    {%
      e_1 \eval [r_1, r_1'] \\
      e_2 \eval r_2 \\
      r_1 \leq r_2 \leq r_1'\\
    }
    {\ediv{e_1}{e_2} \eval ([r_1, r_2], [r_2, r_1'])}
    \and
    \inferrule*[right=E-Split]
    {%
        e_1 \eval (\v_1,\ldots ,\v_{n + n'}) \\
        e_2\{x_{i} \mapsto \v_i \mid 1\leq i \leq n\}\{\rd{w_{i}} \mapsto \rd{\v_{i + n}}, w_{i} \mapsto \v_{i + n}\mid 1\leq i \leq n'\} \eval \v
    }
    {\kwlet\ x_1,\ldots ,x_{n}, w_{1},\ldots ,w_{n'} = \kwsplit\ e_1\ \kwin\ e_2 \eval \v}
\end{mathpar}
    \caption{Select evaluation rules for Slice expressions.}
    \label{fig:srules}
\end{figure}
The rule \rname{E-Tup} forms a tuple out of a collection of values.
\rname{E-Mark} implements the mark query; since the
equation $\vset_{a}[r_1,r] = \sum_{i}r_{i}\vset_{a_{i}}(P_{i})$ can be
satisfied by more than one $r$, the big-step semantics is non-deterministic.
\rname{E-EvalPc} implements the eval query. \rname{E-Div} splits an interval
into two, requiring the interval to contain the split point. Both this condition
and the condition for \rname{E-Mark} mean that some well-typed protocols may
become stuck: they may not
evaluate to values. Lastly,
\rname{E-Split} binds the variables $x_1,\ldots ,x_{n}, w_{1},\ldots ,w_{n'}$ to
the values $\v_1,\ldots ,\v_{n},\ldots ,\v_{n + n'}$, and
binds $\rd{w_1},\ldots ,\rd{w_{n}}$ to read-only versions
$\rd{\v_{n + 1}},\ldots ,\rd{\v_{n + n'}}$ of the affine values.
It is straightforward to show that evaluation preserves types:

% \iffalse
\begin{restatable}[Type soundness]{proposition}{typesoundness}
\label{prop:typesoundness}
If $\cdot \vdash e : \tau$ and $e \eval \v$, then $\cdot \vdash \v : \tau$.
\end{restatable}

% With that said, we now use our type system to verify disjointness.
% \fi

\subsection{Disjointness}

Our affine type system is designed to ensure that well-typed programs
produce only disjoint allocations, i.e., tuples of pieces that do not overlap. To
prove this claim, we generalize and define disjointness for values and general
expressions, and then show that a well-typed disjoint program can only evaluate
to a disjoint value.

Informally, an expression is disjoint if all interval values within it, excluding read-only versions, are disjoint from each other. Disjointness ignores read-only values since well-typed programs are allowed to duplicate them; this does not affect disjointness verification since we are only concerned with programs that return allocations, i.e., values of type $\mathsf{Piece}^{\agents}$.

Since our type system prevents multiple uses of variables with type
$\mathsf{Intvl}$ and $\mathsf{Piece}$, they cannot be duplicated so programs
cannot construct pieces and intervals with overlapping components. This
invariant enables us to show that disjoint expressions only evaluate to disjoint values.
\begin{restatable}{proposition}{disjointness}
    If $\cdot \vdash e : \tau$ and $e$ is disjoint, then $e \eval \v$ implies $\v$ is disjoint.
    \label{prop:disjointness}
\end{restatable}
Checking that a well-typed protocol is disjoint is easily done syntactically, and in the protocols we are concerned with, amounts to ensuring $\kwcake$ is only used once.
% As an illustration, the following protocol evidently would not evaluate to a disjoint value. This happens because $p$ is duplicated, which means this protocol is not well-typed:
%   \[
%     \begin{array}{l}
%         \kwlet\ p = \kwsplit\ \kwcake\ \kwin \\
%         (\kwpiece (p), \kwpiece (p))
%     \end{array}
%   \]

\begin{example}
    We illustrate our type system with the two-agent \emph{Surplus} protocol. In
    brief, both agents are asked to mark the cake at half the value of the whole
    cake. The agent that marked furthest to the left is given all the cake to
    the left of their own mark. Symmetrically, the other agent is given
    everything to the right of their own mark, leaving the cake lying between
    the marks un-allocated. The Slice programs shown in \Cref{fig:surplus} both
    correctly implement the Surplus protocol, however, the left program is not
    well-typed, while the right one is. The left program does not type check
    because it divides the whole cake twice (highlighted in {\color{red}{red}}), leaving either $p_1$ and $p_4$ or $p_2$ and $p_3$ to overlap. Disjointness cannot be verified in this instance since there are intermediate expressions that will not be in its evaluation. The right program avoids this issue by only dividing the cake once it is known where the marks lie in relation to each other.
    \qed
\end{example}
\begin{figure}
    \centering
    \begin{subfigure}[t]{.45\textwidth}
        \centering
  \[
    \begin{array}{l}
        \kwlet\ p = \kwsplit\ \kwcake\ \kwin\\
        \kwlet\ m_1 = \emark{1}{\rd{p}}{1/2\cdot \eeval{1}{\rd{p}}}\ \kwin \\
        \kwlet\ m_2 = \emark{2}{\rd{p}}{1/2\cdot \eeval{2}{\rd{p}}}\ \kwin \\
        \kwlet\ p_1,p_2 = \kwsplit\ \ediv{\color{red}p}{m_1}\ \kwin \\
        \kwlet\ p_3, p_4 = \kwsplit\ \ediv{\color{red}p}{m_2}\ \kwin \\
        \kwif\ m_2 \geq m_1\ \kwthen \\
        \quad (\kwpiece(p_1), \kwpiece(p_4)) \\
        \kwelse\\
        \quad (\kwpiece(p_2), \kwpiece(p_3)) \\ \\ \\
    \end{array}
  \]
  \caption{Not well-typed.}
    \end{subfigure}
    \begin{subfigure}[t]{.45\textwidth}
        \centering
  \[
    \begin{array}{l}
        \kwlet\ p = \kwsplit\ \kwcake\ \kwin\\
        \kwlet\ m_1 = \emark{1}{\rd{p}}{1/2\cdot \eeval{1}{\rd{p}}}\ \kwin \\
        \kwlet\ m_2 = \emark{2}{\rd{p}}{1/2\cdot \eeval{2}{\rd{p}}}\ \kwin \\
        \kwif\ m_2 \geq m_1\ \kwthen \\
        \quad \kwlet\ p_1, p_2 = \kwsplit\ \ediv{p}{m_1}\ \kwin\\
        \quad \kwlet\ p_2, p_3 = \kwsplit\ \ediv{p_2}{m_2} \kwin \\
        \quad (\kwpiece(p_1), \kwpiece(p_3)) \\
        \kwelse\\
        \quad \kwlet\ p_1, p_2 = \kwsplit\ \ediv{p}{m_2}\ \kwin\\
        \quad \kwlet\ p_2, p_3 = \kwsplit\ \ediv{p_2}{m_1} \kwin \\
        \quad (\kwpiece(p_3), \kwpiece(p_1))
    \end{array}
  \]
  \caption{Well-typed.}
    \end{subfigure}
    \caption{The Surplus protocol written in two ways.}
    \label{fig:surplus}
\end{figure}
\section{Constraints}
\label{sec:constraints}
Now that we've seen the Slice language, we review the original Slice constraint translation. For full details see \append.
\paragraph{Paths.}\hypertarget{paths}{}
\label{sec:paths}
As is standard, we consider each path through a program separately. Paths $b$
are Slice expressions (\Cref{sec:lang}) with an assert expression $\kwassert\
b_1\ \kwin\ b_2$ in place of if-then-else.

\paragraph{Logical syntax} \hypertarget{lsyntax}{}
Protocol paths are translated into a multi-sorted first order logic. The logic is standard, so most details areomitted, though we make note of select function symbols:
\begin{itemize}
    \item $[\_, \_]$, $\cup$ for forming intervals and pieces respectively
    \item $\tleft$, $\tright$ for obtaining the left and right endpoints of an interval respectively
    \item $\val_{a}$ for each $a \in \agents$ for representing agent valuations
    \item $\pi_{i}$ for the $i$th component of a tuple
    \item $\mathtt{O}$ which contains logical counterparts to the primitive operations $O$ (e.g. $+$, $\geq$)
\end{itemize}
Through the function symbols and constants, any program value $\v$ can be encoded as a logical term $\lv$. Throughout, the typewriter font designates logical counterparts to program objects.  We also include a special set of variables $\mathcal{Y}$, disjoint from $\mathcal{X}$, which will only be used to represent points in our formulas.

With our logic, we can express envy-freeness, where $x$ represents allocation:
\begin{equation}
    E(x) \triangleq \bigwedge_{a,a' \in \agents} \val_{a}(\pi_{a} x) \geq \val_{a}(\pi_{a'}x).
    \label{eq:envyfree}
\end{equation}

\paragraph*{Logical semantics.}
Formula semantics are given by an interpretation $\interp$ and variable assignment $\assn$. An interpretation associates sorts with sets and
function symbols with functions on these sets. A variable assignment is a map
from variables to elements of these sets. For our purposes, we fix a base
interpretation that interprets everything but the symbols $\val_{a}$ for all
$a$, and all full interpretations agree with the base. The base interprets
objects as one would expect, e.g. $\sem{\tleft} ([r,r']) = r$. For full
interpretations, the symbols $\val_{a}$ are interpreted over all possible
valuations. Thus, full interpretations are uniquely determined by the choice of
valuation set, and we write $\interp_{\vset}$ for the interpretation such that $\sem{\val_{a}}_{\interp_{\vset}} = \vset_{a}$.

For a logical term $t$, we let $\sem{t}_{\interp}^{\assn}$ denote the interpretation
of $t$ according to $\interp$,  with variable values determined by $\assn$,
defined in the usual way (e.g., $\sem{\val_{a}([y,y'])}_{\interp_{\vset}}^{\assn}
=\vset_{a}[\assn(y),\assn(y')]$). Likewise, for a formula $\varphi$, we write $\interp,\assn \vDash \varphi$ if $\varphi$ is true when interpreted through $\interp$ with variable values determined by $\assn$, also defined in the usual way. We write $\interp \vDash \varphi$ if for all assignments $\assn$ we have $\interp, \assn \vDash \varphi$. If $t$ is a term containing no $\val_{a}$ symbols, then for a fixed assignment $\assn$, the
term $t$ is always interpreted the same way and we write just $\sem{t}^{\assn}$.

If $\lv$ is an allocation, $\interp_{\vset}, \assn \vDash E(\lv)$ states that $\sem{\lv}_{\interp_{\vset}}^{\assn}$ is an envy-free allocation. If $e$ is a expression, we say that $e$ \emph{satisfies} $E(x)$ and write $e \vDash E(x)$ if for all valuation sets $\vset$, $e \eval_{\vset} \v$ implies $\interp_{\vset}\vDash E(\lv)$. Thus $e\vDash E(x)$ means $e$ is envy-free.

In order to verify envy-freeness, Slice translates programs $e$ to logical
formulas ensuring $e \vDash E(x)$. We review this constraint translation next,
before describing our improved translation.

\paragraph*{Constraints.}
\hypertarget{constraints}{}

To translate protocols to formulas, we translate each path in a protocol to a
formula consisting of a logical term $\rho(b)$ and a formula $\cn(b)$.
Intuitively, $\rho(b)$ is the logical term representation of the value that $b$
evaluates to assuming that the formula $\cn(b)$ holds. We give some cases of the definition in \Cref{fig:selectrho}.
\begin{figure}
      \begin{mathpar} \displaystyle
          \rho(\rd{\v}) \triangleq \lv
          \and
          \rho_{}(\v)  \triangleq{\lv}
          \and
          \rho_{}(x)  \triangleq{x}
          \and
          \rho_{}(w) \triangleq{w}
          \and
          \rho_{}(\rd{w})  \triangleq \rd{w}
          \and
\rho_{ }(\ediv{b_1}{b_{2}})_{\ilist{i}} \triangleq
    \left(\left[\tleft(\rho_{ }(b_1)_{\ilist{1:i}}), \rho_{ }(b_2)_{\ilist{2:i}}  \right]_{},
    \left[\rho_{ }(b_2)_{\ilist{2:i}}, \tright(\rho_{ }(b_1)_{\ilist{1:i}})  \right]_{}\right)
    \and
    {\rho_{ }}(\emark{a}{b_1}{b_2})_{\ilist{i}}  \triangleq y \in \mathcal{Y}
    \and
    {\rho_{ }}(\eeval{a}{b})_{\ilist{i}} \triangleq \val_{a}(\rho_{ }(b)_{\ilist{1:i}})
    \and
      \rho_{ } (\pt{b_1}{b_2})_{\ilist{i}}  \triangleq  \rho(b_2)_{\ilist{2:i}}
      \end{mathpar}
\begin{mathpar} \displaystyle
\cn(\ediv{b_1}{b_2})_{\ilist{i}}
    \triangleq \cn(b_1)_{\ilist{1:i}}\wedge \cn(b_2)_{\ilist{2:i}} \wedge \tleft(\rho(b_1)_{\ilist{1:i}})\leq \rho(b_2)_{\ilist{ 2 : i}}\leq \tright(\rho(b_1)_{\ilist{ 1 : i}})
    \and
    \cn(\eeval{a}{b})_{\ilist{i}} \triangleq \cn(b)
    \and
    \cn(\emark{a}{b_1}{b_2}) \triangleq \cn(b_1)\wedge \cn(b_2)\wedge (\val_{a}([\tleft(\rho(b_1)), \rho(\emark{a}{b_1}{b_2})])= \rho(b_2))
    \and
    \cn(\pt{b_1}{b_2}) \triangleq (\rho(b_1) = \mathsf{true})\wedge \cn(b_1)\wedge \cn(b_2)
\end{mathpar}
\caption{$\rho(b)$ and $\cn(b)$ for select path expressions $b$.}
    \label{fig:selectrho}
\end{figure}

      It is informative to compare these definitions to the big-step semantics shown in \Cref{sec:lang}. For instance, $\rho(\ediv{b_1}{b_2})$ is a logical encoding of the original interval being split into two, $\rho(\emark{a}{b_1}{b_2})$ is a variable that represents the mark, $\rho(\eeval{a}{b})$ is the value of the interval or piece provided.

The formula $\cn(b)$ is referred to as the \emph{constraint of $b$}.
Roughly, the formula corresponds to the side conditions shown in the big-step semantics.
The most interesting cases of the definition are in \Cref{fig:selectrho}. All constraints conjoin the conditions from their subexpressions.
 The constraint for $\kwdiv$ encodes that the point dividing the interval must be within.
The constraint of $\kweval$ has no additional conditions to satisfy, so it is just the constraint of its subexpression.
The constraint of $\kwmark$ ensures that the new point has the required property and the constraint of $\kwassert$ asserts that the guard must hold.

 \begin{example}\hypertarget{ex:ccpath}{}
     \label{ex:cnstr}
     Consider the following path, denoted $b$, from Cut-Choose (\Cref{proto:cc}):
  \[
    \begin{array}{l}
        \kwlet\ p = \kwsplit\ \kwcake\ \kwin\\
        \kwlet\ p_1, p_2 = \kwsplit\ \ediv{p}{\emark{1}{\rd{p}}{1/2\cdot \eeval{1}{\rd{p}}}}\ \kwin\\
        \pt{\eeval{2}{\rd{p_1}} \geq \eeval{2}{\rd{p_2}}}{(\kwpiece (p_2), \kwpiece (p_1))}
    \end{array}
  \]
     The path $b$ gives the following (simplified) constraint:
\begin{equation*}
  \begin{array}{rl}
      \cn(b) &= (\val_{1}([0,y]) = 1/2 \cdot \val_1([0,1])) \wedge ( \val_{2}([0,y]) \geq \val_{2}([y,1])) \\
      \rho(b) &= (\cup [y,1], \cup [0,y])
  \end{array}
\end{equation*}
The first conjunct in $\cn(b)$ is from the expression $\emark{1}{\overline{p}}{1/2\cdot \eeval{1}{\overline{p}}}$, while the second is from $\eeval{2}{\overline{p_1}}\geq \eeval{2}{\overline{p_2}}$. The term $\rho(b)$ is a logical encoding of $b$'s evaluation.
\qed
\end{example}

The following result, akin to Corollary 4.8 for Slice~\cite{slice}, characterizes paths in terms of their constraints.
\begin{restatable}[]{theorem}{ccharacterization}
    Suppose $\cdot \vdash b : \tau$. Then $b \eval_{\vset} \v$  if and only if there is a variable assingment $\assn$ such that $\interp_{\vset},\assn \vDash \cn(b)$ and $\sem{\rho(b)}_{\interp_{\vset}}^{\assn} = \unrd{\v}$.
\end{restatable}

With our constraint translation being sound and complete, we look to use constraints to verify envy-freeness only by checking the validity of certain formulas involving the constraint. For the following, let $\mathcal{Y}_{b}$ be the set of free variables contained
in $\cn(b)$, and let $B(e)$ be the set of paths within $e$. The following theorem forms the basis for automated verification in Slice.
\begin{restatable}[]{theorem}{solvingthm}
    \label{prop:solvingenvy}
    Suppose that $e$ is a well-formed expression and $\cdot \vdash e : \mathsf{Piece}^{\agents} $.
    Then
    \begin{equation}
        \interp_{\vset} \vDash \bigwedge_{b \in B(e)} \forall \mathcal{Y}_{b}.(\cn(b)\Rightarrow E(\rho(b)))
        \label{eq:solvingenvy}
    \end{equation}
    for all $\vset$ if and only if $e \vDash E(x)$.
\end{restatable}
The formal definition for a well-formed expression can be found within \append, though the imposed conditions are mild; any typical cake-cutting protocol is well-formed. This theorem can be generalized from $E(x)$ to general formulas $F(x)$ satisfying mild conditions.

We stress that in order to apply this theorem to conclude $e \vDash E(x)$, Formula (\ref{eq:solvingenvy}) needs to be valid \emph{for all} valuation sets. Our logic is not rich enough to quantify over valuations and their axioms, so for verification, these formulas must be embedded in a richer theory (e.g., from a modern SMT solver).

\section{Piecewise uniform reduction}
\label{sec:pu}

Now that we have seen how the existing constraint translation works in Slice, we
show how to produce a result similar to \Cref{prop:solvingenvy}, but instead with
a formula in the theory of linear real arithmetic. Formula (\ref{eq:solvingenvy}) contains
terms like $\tleft([t_1,t_2])$, $\pi_{k}(t_1,\ldots ,t_{n})$, and $\val_{a}(t)$,
which all need to be reduced to linear sums of real variables. Most terms can be
reduced via syntactic simplifications, but reducing valuation terms
$\val_{a}(t)$ is much more challenging.

The broad approach is to show a protocol execution on any valuation set can be
replicated with a \emph{piecewise uniform valuation} set, then replace terms
$\val_{a}(t)$ with sums of differences of real variables that represent
$\val_{a}(t)$. We discuss conditions under which protocol executions can be
replicated, then show there are always piecewise uniform valuations meeting
these conditions. Then, we describe how to construct the formula reduction,
prove that it preserves validity, and then apply it to obtain an analog to
\Cref{prop:solvingenvy}. Our approach is inspired by Theorem 1 from Kurokawa,
Lai, and Procaccia~\cite{Kurokawa_Lai_Procaccia_2013}.

\textbf{For this section only, we will assume that the operations $\mathcal{O}$
consist only of boolean operators, comparisons, constant multiplication and
addition}. These operations are sufficient for describing cake-cutting
protocols. A more detailed description of $\mathcal{O}$ is shown in \append.

\subsection{Replicating protocol executions}

In this subsection, we give a condition when the same evaluation judgement holds for two possibly different valuation sets. For this, we define the following relationship between valuation sets.
\begin{definition}
    \hypertarget{def:repl}{}
Let $M\supseteq \{0,1\}$ a finite set of points. We say that valuation sets $\uset$ and $\vset$ \emph{agree on $M$} if
       for any piece $P$ with boundary points in $M$, $V_{a}(P) = U_{a}(P)$ for all $a \in \agents$.
       \label{def:repl}
\end{definition}
The following theorem says that we can identically derive an evaluation judgement with a different valuation set, as long as the valuation set agrees with the original on all pieces formed from points in the derivation.
\begin{restatable}[]{theorem}{pured}
    Let $\uset$ and $\vset$ be valuation sets, and suppose $e\eval_{\vset} \v$. If $\uset$ and $\vset$ agree on all points considered in the derivation of $e \eval_{\vset}\v$, then $e \eval_{\uset} \v$.
    \label{thm:red}
\end{restatable}
There is an analog for formulas.
\begin{restatable}[]{theorem}{pufequiv}
If valuation sets $\uset$ and $\vset$ agree on the set of points considered in a formula $\varphi$ under variable assignment $\assn$, then
    $\interp_{\vset},\assn \vDash \varphi \iff \interp_{\uset},\assn\vDash \varphi$.
    \label{thm:pufequiv}
\end{restatable}
\subsection{Piecewise uniform valuations}
Now that we have seen what is required for replication, we show that there is always a special piecewise uniform valuation set that meets the requirements.

We first formally define piecewise uniform valuations.
It is easiest to define these valuations in terms of their \emph{density}.
For our purposes, a \emph{density} is a function $w : [0,1] \to \mathbb{R}_{\geq 0}$, and a valuation \emph{$W$ has density $w$} if
$W({P}) = \int_{P}w\quad$ for all $P\in \pieces$.
\begin{definition}
    \hypertarget{def:pu}{}
    We say that a valuation $U$ is \emph{piecewise uniform} if $U$ has density $u$ for which there exists a piece $P \in \pieces$ and a constant $c$ such that
    \[u(x) =
    \begin{cases}
        c &\text{if}\ x\in P \\
        0 &\text{if}\ x\not\in P.
    \end{cases}\]
We let $P(U)$ denote $P$ and $c(U)$ denote $c$.
\label{def:pu}
\end{definition}
Because valuations are normalized, the constant associated with a piecewise uniform valuation $U$ is the reciprocal length of $P(U)$. Therefore, any piece $P$ uniquely determines a piecewise uniform valuation $U_{P}$, where $P(U_{P}) = P$.

Much of the advantage of these valuations lies in how we can represent their values on specific pieces. For intervals built from right endpoints of $P(U)$, the valuation reduces to a simple sum of differences between real numbers. If we write out $P(U)= [l_1, r_1]\cup \cdots \cup [l_{n}, r_{n}]$ where $l_1 \leq r_1 < \cdots < l_{n} \leq r_{n}$, then
\begin{equation}
U[r_{i}, r_{i'}] = c(U)\cdot \sum_{i' \geq j > i} (r_{j} - l_{j}).
\label{eq:pusum}
\end{equation}
 This formula is key for our reduction, as it enables us to convert valuations applied to intervals (left) to sums of differences of real numbers (right).

We call a valuation set a \emph{piecewise uniform valuation} set if all valuations within it are piecewise uniform. Our formula reduction will benefit from the following key conditions on piecewise uniform valuation sets.
\begin{definition}
    \label{def:er}
    \hypertarget{def:er}{}
    Let $\uset$ be a piecewise uniform valuation set and let $M\supseteq \{0,1\}$ be a finite set of points. We say that $\uset$ is \emph{easily replaceable on $M$} if
\begin{enumerate}
    \item For each $a \in \agents$, and for each $m \in \noz{M}\triangleq M \setminus \{0\}$, there exists $l_{a}(m)$ such that if $l_{a}(m) < m'$ for $m' \in \noz{M}$, then $m \leq m'$ and
            \[
                P(\uset_{a}) =  \bigcup_{m \in \noz{M}}[l_{a}(m), m].
        \]
    \item For each $a,a'\in \agents$, $c(\uset_{a}) = c(\uset_{a'})$.
    \end{enumerate}
\end{definition}
The first part is valuable for the reduction as it removes the need to keep track of distinct right endpoints for each agent. The second part means that the coefficient in \Cref{eq:pusum} can be ignored when comparing these valuations with each other, which will be important later for the formulas to be in real linear arithmetic.
\begin{restatable}[]{theorem}{puexists}
    For any valuation set $\vset$ and any finite set of points $M\supseteq \{0,1\}$, there exists a piecewise uniform valuation set that both agrees with $\vset$ on $M$ and is easily replaceable on $M$.
    \label{prop:puexists}
\end{restatable}
The proof constructs a specific piecewise uniform valuation set that satisfies these properties. The construction is a slightly more general version of the construction shown in Theorem 1 by Kurokawa, Lai, and Procaccia~\cite{Kurokawa_Lai_Procaccia_2013}.

A consequence of the theorem is that any protocol execution can be replicated by the specific piecewise uniform valuation set, and any formulas that hold for the original valuation set that only consider points from the execution will also hold for this specific valuation set.
\begin{example}
    \label{ex:exists}
    We illustrate the construction for $\agents = \{1,2\}$ with valuation $\vset_{1}$ being the uniform valuation over the cake, and $\vset_2$ having density $x \mapsto 2x$, and the set of points $M = \{0, 1/2, 1\}$. Set $\uset_1 = U_{P_1(d)}$ and $\uset_2 = U_{P_2(d)}$ for pieces
    \[
        \begin{array}{l}
            P_1(d) = [1/2 - 1/2\cdot 1/d  , 1/2]\cup [1- 1/2\cdot 1/d, 1]\\
            P_2(d) = [1/2 - 1/4\cdot 1/d  , 1/2]\cup [1- 3/4\cdot 1/d, 1]
        \end{array}
    \]
    for $d \geq 3/2$. Clearly both pieces have interval right endpoints of $\{1/2,1\} = \noz{M}$. Also, it is easily to calculate that $c(U_{P_1}) = c(U_{P_2}) = d$. Thus, $\uset$ is easily replaceable on $M$. We additionally have
    \iffalse
    \begin{align*}
     &\uset_1[0,1/2] = d\cdot (1/2 - (1/2 - 1/2\cdot 1/d)) &= 1/2 &= \vset_1[0,1/2],\\
     &\uset_1[1/2,1] = d\cdot (1 - (1 - 1/2\cdot 1/d)) &= 1/2 &= \vset_1[1/2,1],\\
     &\uset_2[0,1/2] = d\cdot (1/2 - (1/2 - 1/4\cdot 1/d)) &= 1/4 &= \vset_2[0,1/2],\\
     &\uset_2[1/2,1] = d\cdot (1 - (1 - 3/4\cdot 1/d)) &= 3/4 &= \vset_2[1/2,1],
    \end{align*}
\fi
    \[
        \begin{array}{l l l}
     \uset_1[0,1/2] = & d\cdot (1/2 - (1/2 - 1/2\cdot 1/d)) &= 1/2 = \vset_1[0,1/2],\\
     \uset_1[1/2,1] = & d\cdot (1 - (1 - 1/2\cdot 1/d)) &= 1/2 = \vset_1[1/2,1],\\
     \uset_2[0,1/2] = & d\cdot (1/2 - (1/2 - 1/4\cdot 1/d)) &= 1/4 = \vset_2[0,1/2],\\
     \uset_2[1/2,1] = & d\cdot (1 - (1 - 3/4\cdot 1/d)) &= 3/4 = \vset_2[1/2,1],
     \end{array}
     \]
    so $\uset$ replicates $\vset$ on $M$. \qed
\end{example}

\subsection{Piecewise uniform replacment}
We leverage the above results to produce a reduction on protocol constraints. At a high level, for any formula having only variables for points, we can simplify it to be a disjunction of inequalities in terms of the form $\sum_{i}r_{i}\cdot \val_{a_{i}}(P_{i})$ for real $r_{i}$ and logical pieces and intervals $P_{i}$. We then replace terms of the form $\val_{a}(P)$ with sums of differences of real variables. Using \Cref{thm:pufequiv} and \Cref{prop:puexists} we can show that the original formula holds if and only if the replaced formula holds for a piecewise uniform valuation set.

 \emph{Simplified terms} are the terms within the following sets, indexed by sort:
\begin{align*}
    R_{\mathtt{Point}} &\triangleq \mathcal{Y} \cup \{r\pnt \mid r \in \mathbb{R}\}
                       &
    R_{\mathtt{Piece}} &\triangleq  \{\cup(t_1,\ldots ,t_{i})\mid t_{i} \in R_{\mathtt{Intvl}}\}
    \\
    R_{\mathtt{Intvl}} &\triangleq  \{[t,t']\mid t,t' \in R_{\mathtt{Point}}\}
                       &
    R_{\mathtt{Vltn}} &\triangleq \{ {\textstyle\sum_{i}}r_{i}\cdot \val_{a_{i}}(P_{i})\mid P_{i} \in R_{\mathtt{Intvl}}\cup R_{\mathtt{Piece}}\}
\end{align*}
For any well-sorted term $t$ containing only variables in $\mathcal{Y}$, we can produce an equivalent simplified version of it, which we denote by $\simpl(t)$. The notion of simplified and the simplification operation $\simpl$ easily extends to whole formulas as well. For further details of this step, see \append.
For further use, if $t$ is a simplified term or formula, we let $\pnt(t)$ denote the subset of $R_{\mathtt{Point}}$ contained as subterms of $t$.

Proceeding with the reduction, we introduce a new set of logical variables, $\mathcal{Z}$ and we assume for each $y\in \mathcal{Y}\cup \{1\}$, and $a\in \agents$, there is a unique $z_{a,y} \in \mathcal{Z}$. To understand the purpose of $\mathcal{Z}$, consider a piecewise unifrom valuation set $\uset$ that is easily replaceable on $M$. Then $P(\uset_{a})$, the piece corresponding to agent $a$'s valuation, is the union of intervals of the form $[l_{a}(m),m]$ for $m \in \noz{M}$. If the variable $y$ represents the variable $m$, then the variable $z_{a,y}$ then represents the left endpoint $l_{a}(m)$.

\begin{definition}
    A \emph{piecewise uniform replacement} is a finite totally ordered subset $(S, >_{S})$ of $\mathcal{Y}\cup \{0,1\}$ such that $\{0,1\}\subseteq S$ and $0 \leq_{S} y \leq_{S}1 $ for all $y\in S$.
    For $y,y'\in S$, we define $\sub{y}{y'} \triangleq \{y '' \in S \mid y' \geq_{S} y'' >_{S} y\}$. We let
    $S|_{t}\triangleq \sub{y}{y'}$ if $t = [y, y']$ and $S|_{t} \triangleq \sub{y_1}{y'_{1}}\cup \cdots \cup \sub{y_n}{y'_{n}}$ if
    $t = \cup [y_1 ,y'_1],\ldots ,[y_{n}, y'_{n}]$.
\end{definition}
The piecewise uniform replacement packages neatly all the data needed to replace valuation symbols in formulas. The order on $S$ represents the ordering of real numbers, since variables from $\mathcal{Y}$ represent points. A replacement is applied to terms:
\begin{definition}
    \label{def:rep}
    The application of a piecewise uniform substitution $S$ on a term $t\in R_{\mathtt{S}}$ for which $\pnt(t)\subseteq S$ is as follows:
    \begin{align*}
        S(\val_{a}(t)) &\triangleq  {\sum_{y \in S|_{t}}} (y - z_{a,y})
                       &
        S({ \sum_{i}}r_{i}\cdot t_{i}) &\triangleq {\sum_{i}}r_{i}\cdot S(t_{i})
                                       &
        S(t) &\triangleq t\ \text{otherwise.}
    \end{align*}
    $S$ can be applied to formulas by passing itself down to its terms. When $S|_{t}$ is empty, we replace the term with $0$.
\end{definition}
The piecewise uniform replacement syntactically applies \Cref{eq:pusum} (ignoring the constant) to valuation terms for valuations of the form shown in \Cref{def:er}.
The following example illustrates this concretely.
\begin{example}
    \label{ex:lincnstr}
    Returning to the path $b$ shown in \Cref{ex:cnstr}, consider the piecewise uniform replacement $S = \{0, y ,1\}$ where $0 <_{S} y <_{S} < 1$. Then $\val_{a}([0,y])$ and $\val_{a}([y,1])$ are replaced with $y - z_{a,y}$ and $1 - z_{a,1}$ respectively. The formula $\cn(b)$ simplifies to
    \begin{equation*}
        \begin{array}{l}
            S(c(b)) = (y - z_{1,y} = 1/2 \cdot (y - z_{1,y} + 1 - z_{1,1})) \wedge (y - z_{2,y} \geq 1 - z_{2,1}), \\
        \end{array}
    \end{equation*}
    and the encoding of envy-freeness becomes
    \begin{equation*}
        \begin{array}{l}
            S(E(\rho(b))) = (y - z_{2,y} \geq 1 - z_{2,1}) \wedge (y - z_{1,y} \geq 1 - z_{1,1}).
        \end{array}
    \end{equation*}
    Both are clearly linear inequalities in real variables. \qed
\end{example}

Piecewise uniform replacements are used to reduce simplified formulas to linear real inequalities:
\begin{proposition}
    Let $f$ be a simplified formula. Let $S$ be a piecewise uniform replacement such that $\pnt(f)\subseteq S$. Then $S(f)$ consists only of conjunctions and disjunctions of linear inequalities of real variables.
    \label{cor:linear}
\end{proposition}

To apply piecewise uniform replacements in a sound way, the variable assignment must properly line it up with the valuation set; the precise conditions for this are given in the following definition.
\begin{definition}
    \label{def:con}
    Let $\uset$ be a piecewise uniform valuation set.
    Let $S$ is a piecewise uniform replacement and $\assn$ a variable assignment. We write $\con{S}{\assn}{\uset}$ if
    \begin{enumerate}
        \item $\uset$ is easily replaceable on $\assn(S)$ (by convention $\assn(0) = 0$ and $\assn(1) = 1$)
        \item $\assn(z_{a,y}) = l_{a}(\assn(y))$ if $y = \min\{ y'\in S \mid \assn(y') = \assn(y) \}$
        \item $\assn(z_{a,y}) = \assn (y)$ if $y \neq \min\{ y'\in S \mid \assn(y') = \assn(y) \}$
        \item If $\assn(y) < \assn(y')$ then $y <_{S} y'$.
    \end{enumerate}
\end{definition}
This definition formalizes how we think of variables in a piecewise uniform
replacement. Condition (1) says that $\assn(S)$ captures the right endpoints of
$P(\uset_{a})$ correctly, (2) and (3) together ensure that we don't repeat
values in our sums, and (4) ensures that the variable ordering is compatible with the real ordering given by $\assn$.

\begin{example}
    \label{ex:con}
    Let $S = \{0, y, 1\}$, and let $\uset$ be the piecewise uniform valuation set described in \Cref{ex:exists}.
    Set $\assn(y) =1/2$, and
    \begin{align*}
        \begin{array}{l l}
            \assn(z_{1,y}) = {1/2 - 1/2\cdot 1/d} \quad & \quad \assn(z_{1,1}) = {1 - 1/2\cdot 1/d} \\
            \assn(z_{2,y}) = {1/2 - 1/4\cdot 1/d} \quad & \quad \assn(z_{2,1}) = {1 - 3/4\cdot 1/d}.
        \end{array}
    \end{align*}
    Then $\assn(S) = \{0, 1/2, 1\}$ and \Cref{ex:exists} illustrates that $\uset$ is easily replaceable on $\assn(S)$.
    Also, $z_{a,y}$ is the left endpoint, $l_{a}(1/2)$, of the left interval for $P_{a}(d)$, and $z_{a,1}$ is the left endpoint, $l_{a}(1)$, of the right interval for $P_{a}(d)$, hence condition (2) is satisfied. Condition (3) is vacuous here. Clearly, $0 < 1/2 < 1$ and $0 <_{S} y <_{S} 1$ so condition (4) is satisfied. Thus we have that $\con{S}{\assn}{\uset}$. \qed
\end{example}

Piecewise uniform replacements preserve validity when the conditions in \Cref{def:con} are met.
\begin{restatable}[]{theorem}{purepsound}
    \label{prop:purepsound}
    Let $f$ be a simplified formula and let $S$ be a piecewise uniform replacement such that $\pnt(f) \subseteq S$. Let $\assn$ be an assignment and $\uset$ a piecewise uniform valuation set. If $\con{S}{\assn}{\uset}$ then $\interp_{\uset},\assn \vDash f \iff \interp_{\uset},\assn\vDash S(f)$.
\end{restatable}
\begin{example}
    We illustrate the theorem by applying it with $S = \{0, y, 1\}$ for the formula $\cn(b)$ \Cref{ex:cnstr} reproduced here:
\begin{equation*}
  \begin{array}{rl}
      \cn(b) &= (\val_{1}([0,y]) = 1/2 \cdot \val_1([0,1])) \wedge ( \val_{2}([0,y]) \geq \val_{2}([y,1])).
  \end{array}
\end{equation*}
Supposing $\con{S}{\assn}{\uset}$, this formula is equivalent to its reduced version from \Cref{ex:lincnstr}:
    \begin{equation*}
        \begin{array}{l}
            S(c(b)) = (y - z_{1,y} = 1/2 \cdot (y - z_{1,y} + 1 - z_{1,1})) \wedge (y - z_{2,y} \geq 1 - z_{2,1}).
        \end{array}
    \end{equation*}
    One can verify this equivalence for the example $\uset$ and $\assn$ shown in \Cref {ex:con}.
\end{example}

Associated with each piecewise uniform replacement $S$, is a formula $\conj(S)$.
\begin{definition}
    Let $S$ be a piecewise uniform replacement, written $S = \{y_1 ,\ldots ,y_{n}\}$ so that $y_1 <_{S} \cdots < _{S} y_{n}$.
    We let $\conj(S)$ denote the conjunction of the following formulas for all agents $a,a'\in\agents$:
    \begin{align*}
    &0 \leq z_{a, y_{1}} \leq y_{1} \leq \cdots \leq z_{a,y_{n}} \leq y_{n} \leq 1,& \sum_{y \in S} y - z_{a, y} &= \sum_{y \in S\setminus\{0\}} y - z_{a', y}.
    \end{align*}
\end{definition}
Whenever the above formula holds for some variable assignment $\assn$, a piecewise uniform valuation set $\uset$ that is easily replaceable on $\assn(S)$ can be constructed:
\[
    P(\uset_{a}) = \bigcup_{y \in S\setminus\{0\}}[\assn(z_{a,y}), \assn(y)], \quad \quad c (\uset_{a}) = \sum_{y \in S\setminus \{0\}} \assn(y) - \assn(z_{a,y}).
\]
This assists us in showing that our constraint reduction procedure is complete.

We now state our main theorem. For a path $b$, let $S_{b}$ be the set of
piecewise uniform replacements on $\mathcal{Y}_{b}\cup \{0,1\}$---note that this
set is \emph{finite}.
\begin{restatable}[]{theorem}{simplsolvingenvy}
    Suppose $e$ is well-formed and $\cdot \vdash e : \mathsf{Piece}^{\agents}$. Then $e \vDash E(x)$ if and only if
    \begin{equation}
        \label{eq:pusolvingenvy}
        \vDash \bigwedge_{b \in B(e)}  \bigwedge_{S \in S_{b}} \forall \mathcal{Y}_{b}.S(\simpl(\cn(b)\wedge \conj(S) \Rightarrow E(\rho(b)))).
    \end{equation}
    \label{prop:simplsolvingenvy}
\end{restatable}

In contrast to \Cref{prop:solvingenvy}, we no longer need to quantify over valuations---a valuation set is baked into the formula through $\conj(S)$. This also gives a valuation set witness whenever Formula (\ref{eq:pusolvingenvy}) does not hold.

Similar to \Cref{prop:solvingenvy}, this theorem can be extended to more general
formulas $F(x)$.

\begin{proof}[sketch]
For the forward direction, we assume that $e \eval_{\vset} \v$ and apply \Cref{prop:puexists,thm:red} to obtain a piecewise uniform valuation set $\uset$ for which $e \eval_{\uset} \v$. Then it is a matter of applying \Cref{prop:purepsound} and \Cref{thm:pufequiv} to obtain that $E(\lv)$ is satisfied. For the backward direction, we suppose that Formula (\ref{eq:pusolvingenvy}) doesn't hold for some $b$ and $S_{b}$, and use $\conj(S)$ to construct a piecewise uniform valuation set that evaluates to $\v$ yet $E(\lv)$ is not satisfied.
\end{proof}

According to \Cref{cor:linear}, Formula (\ref{eq:pusolvingenvy}) consists
entirely of linear inequalities of real variables. Thus, we have the following corollary.
\begin{restatable}[]{corollary}{decidableenvy}
    Let $e$ be a well-formed and well-typed Slice protocol. Checking if $e$ is
    envy-free is decidable.
    \label{cor:decidable}
\end{restatable}

\section{Implementation \& Evaluation}
\label{sec:eval}

We implemented our type system and formula reduction on top of the Slice
implementation. Protocols are first type-checked following our linear typing
rules, and then compiled to linear real arithmetic constraints encoding envy-freeness, which are
dispatched to Z3~\cite{10.1007/978-3-540-78800-3_24}.

\paragraph*{Benchmark protocols.}
In our benchmarks, we include all original protocols implemented in
Slice~\citep{slice}; we briefly describe them here.  \emph{Cut-choose} is the
classic 2 agent protocol where one agent cuts and the other picks.
\emph{Surplus} is a two agent protocol which leaves a
``surplus'' piece of the cake in the center. \emph{Selfridge-Conway-Full} is the
classic three agent protocol \cite{brams2006better}. \emph{Selfridge-Conway-Surplus} is a variant of
Selfridge-Conway-Full that disposes the trimming, and
\emph{Waste-Makes-Haste-3}~\cite{wmh} effectively is a minor variant on
Selfridge-Conway-Surplus.

We also implement two new, more complicated protocols. The first,
\mbox{\emph{Aziz-Mackenzie-3}}, is the three-agent variant of the
first bounded envy-free four agent protocol with no free disposal \cite{aziz4agents}. Briefly, this protocol obtains an envy-free allocation by first obtaining a partial allocation where one agent does not care how the rest is allocated amongst the others. Cut-Choose is then applied.
The second,
\emph{Waste-Makes-Haste-4}, is the four-agent connected variant of the Waste-Makes-Haste
free disposal protocol~\citep[Section 6]{wmh}.
This protocol relies on \emph{equalize} queries: Equalize$_{a}(n)$ has agent $a$ divide the cake to produce $n$ equally most preferred (according to $a$) pieces of the cake. It can be shown using Hall's marriage theorem that an envy-free allocation can be made from a set of pieces following some sequence of equalize queries among the agents (for 4 agents, $n \leq 5$), although the allocation must be found through exhaustive search. This protocol exhaustively tries certain sequences of equalize queries until an envy-free allocation is obtained. Notably, this is the first four agent envy-free cake-cutting protocol to be implemented and verified.

\paragraph*{Evaluation.}
Table~\ref{tab:experiment1} presents some statistics from verifying
envy-freeness for each of our benchmark protocols. Our experiments were
conducted on an M1 MacBook Pro with 16 GB of RAM.  We measured the time both to
compile protocols to constraints, and the actual time Z3 took to solve. We also
record here the number of paths in each protocol, as well as the number of lines
for the protocol implementation and the constraint formula. Each path
corresponds to a distinct disjunct in the constraint. We measure this against
the solving time for the original Slice constraints (Old), which uses non linear real arithmetic formulas. Our results demonstrate a reduction in solving time compared with the old constraints. The four-agent protocol is significantly more complex than the others, though Z3 can still solve
the constraints efficiently.

\begin{table}[htbp]
    \centering
    \caption{Verifying envy-freeness (averaged over 5 runs).}
    \begin{tabular}{
      @{}lrrrrrr@{}
    }
      \toprule
      \multicolumn{1}{c}{} & \multicolumn{1}{c}{~} & \multicolumn {2}{c}{Size (lines)} & \multicolumn{3}{c}{Time (seconds)}\\
      \cmidrule(r){3-4} \cmidrule(r){5-7}
      {Protocol} & {\#Paths} & {Program} & {Constraints} & {Compile} & {Z3} & { Z3 (Old)}
      \\
      \midrule
      Cut-Choose                & 2         & 6     & 35        & 1.31      & 0.00     & 0.02 \\
      Surplus                   & 2         & 11    & 56        & 1.23      & 0.00    & 0.02 \\
      Waste-Makes-Haste-3       & 24        & 8    & 924       & 0.85      & 0.02     & 0.84  \\
      Selfridge-Conway-Surplus  & 216       & 19    & 7726      & 1.09      & 0.01     & 0.82 \\
      Selfridge-Conway-Full     & 1800      & 21    & 98292     & 9.20      & 0.46      & 19.38    \\
      Aziz-Mackenzie-3          & 93384     & 23   & 8086180   & {2m4}     & 6.82      & n/a \\
      Waste-Makes-Haste-4       & 1953792   & 290   & 157553237 & {37m02}   & {1m22}    & n/a   \\
      \bottomrule
    \end{tabular}
    \label{tab:experiment1}
  \end{table}

\section{Related \& Future Work}
\label{sec:rw}

\paragraph{Cake Cutting Verification.}
In recent work, \citet{Lester} proposes a system called Crumbs to verify and
disprove envy-freeness, using C bounded model checker (CBMC)
instead of SMT. While performance on correct (envy-free) protocols is similar to
the prior version of Slice, \citet{Lester} shows that Crumbs is much more
effective at finding counterexamples for incorrect protocols. By using our new
constraint reduction, our current work significantly outperforms Slice and
Crumbs for correct protocols, and we can efficiently construct counterexamples
for incorrect protocols (details in \append). In terms of expressivity, Crumbs
supports a more restrictive, higher-level query model, enabling constraint
solving over bounded integer arithmetic. In contrast, our work supports all
protocols written in the standard Robertson-Webb model. Our system
also establishes disjointness, which Crumbs does not consider.

\paragraph{Substructural Type Systems.}
Our type system is an example of a \emph{substructural type system}, which
originate from substructural logics.  In brief, substructural logics restrict
the application of assumptions in proofs. Likewise, substructural type systems
restrict the usage of variables, enabling computational resource usage to be
restricted. A classic example is \emph{linear logic}, due
to~\citet{GIRARD19871}, which led to \emph{linear type systems}
\citep{DBLP:journals/tcs/Lafont88,ABRAMSKY,DBLP:conf/ifip2/Wadler90,10.1145/115865.115894}.
\citet{10.5555/1076265} provides a resource for learning about substructural
type systems. Our type system is designed to ensures \emph{physical
disjointness} of parts of the cake; we are not aware of prior work that uses
substructural types for a single divisible good, though there are similar ideas
in separation logic (e.g.,~\citep{DBLP:series/lncs/Boyland13}).

% JH: There seemed to be space, so I put this following paragraph back in.
\paragraph{Formal Methods and Social Choice.}
Cake-cutting protocols belong to a broader literature on social choice theory,
which has had many fruitful interactions with formal methods. In one direction,
formal methods researchers have used interactive theorem provers to verify
classical protocols and impossibility theorems in social choice theory (e.g.,
\citep{DBLP:journals/jar/Nipkow09}). In the other direction, social choice
researchers have used computer-aided solvers to prove novel theorems in social
choice theory (e.g., ~\citep{DBLP:phd/dnb/Geist16}).

\paragraph{Conclusions and future work.}
Our work makes progress in cake-cutting protocol verification, through an affine type system for disjointness and a formula reduction that enables much more efficient envy-freeness checking.
However, there are envy-free cake-cutting protocols even more complex than what we can verify here. The complexity of these protocols makes it difficult to even write them down in Slice, let alone the constraint compiling and solving time involved. New Slice language features may be needed to address transcription effort, while improvements to the Slice implementation and early pruning of unreachable paths could significantly descrease both compile and solving time.
The most notable of these protocols is the four agent version of Aziz-Mackenzie-\emph{3}~\citep{aziz4agents}, which does not discard any cake, unlike Waste-Makes-Haste-4. By making our affine type system instead linear, it could be used to verify no cake is discarded for that protocol, and others already implemented.
\begin{credits}
  \subsubsection{\ackname}
  We thank Cornell's PL discussion group (PLDG) and Martin Lester for discussions about this work.
  We also thank the anonymous reviewers for their close reading and detailed
  feedback. This work is partially supported by NSF grant CCF-2319186.
\end{credits}

\bibliographystyle{splncs04nat}
\bibliography{header,main}

\newgeometry{left=2cm,right=2cm}
\iflong
\section*{Appendix}

\subsection{Values}
\label{app:values}
Recall the definition of \hyperlink{values}{values}.
Here we define the read and unread operations carefully.

\iffalse
\[
    \begin{array}{c c c}
        \mathbb{V}_{\mathsf{Real}} = \mathbb{R} &
        \mathbb{V}_{\mathsf{Bool}} = \mathbb{B} &
        \mathbb{V}_{\tau_1 \times \cdots \times \tau_{n}} = \mathbb{V}_{\tau_1} \times \cdots \times \mathbb{V}_{\tau_{n}}
    \end{array}
\]
\[
    \begin{array}{c c c c c c}
        \mathbb{V}_{\mathsf{Point}} = \mathbb{P}\mathbb{T} &
        \mathbb{V}_{\mathsf{Vltn}} = \mathbb{VAL} &
        \mathbb{V}_{\mathsf{Intvl}} =  \mathbb{I} &
        \mathbb{V}_{\rd{\mathsf{Intvl}}} = \rd{\mathbb{I}} &
        \mathbb{V}_{\mathsf{Piece}} = \mathbb{P}\mathbb{C} &
        \mathbb{V}_{\rd{\mathsf{Piece}}} = \rd{\mathbb{P}\mathbb{C}}
    \end{array}
\]
\fi
\begin{definition}
    Define the ``read'' operation and ``unread'' operation as follows:
    \begin{align*}
        \rd{\v} =
        \begin{cases}
            \rd{[r,r']} &\ \v = [r,r'] \\
            \rd{P_{i}[r_{i},r'_{i}]} &\ \v = P_{i}[r_{i},r'_{i}] \\
            (\rd{\v_1},\ldots ,\rd{\v_{n}}) &\ \v = (\v_1,\ldots ,\v_{n})\\
            \v &\text{otherwise}.
        \end{cases}
                &
        \quad \unrd{\v} =
\begin{cases}
    [r,r'] &\ \v = \rd{[r,r']} \\
    P_{i}[r_{i},r'_{i}] &\ \v = \rd{P_{i}[r_{i},r'_{i}]} \\
    (\unrd{\v_1},\ldots ,\unrd{\v_{n}}) &\ \v = (\v_1,\ldots ,\v_{n})\\
            \v &\text{otherwise}.
        \end{cases}
    \end{align*}
    These extend to sets:
    $\rd{\mathbb{V}} = \{\rd{\v} \mid \v \in \mathbb{V}\}$, $\unrd{\mathbb{V}} = \{\unrd{\v} \mid \v \in \mathbb{V}\}$.
\end{definition}

\subsection{Types}
\label{app:types}
\begin{figure}
\begin{bnf*}
    \bnftd{$\ut$} \bnfpo
    \bnfmore{%
        \mathsf{Bool}
        \bnfor \mathsf{Point}
        \bnfor \mathsf{Vltn}
        \bnfor \rd{\mathsf{Intvl}}
        \bnfor \rd{\mathsf{Piece}}
    }\\
    \bnftd{$\lint$} \bnfpo
    \bnfmore{%
        \mathsf{Intvl}
        \bnfor \mathsf{Piece}
        \bnfor \ut_1 \times \cdots \times \ut_{n} \times \lint_1 \times \cdots \times \lint_{n}
    }
\end{bnf*}
\caption{Types for Slice expressions.}
\label{fig:types}
\end{figure}
Recall our \hyperlink{types}{types} here, or in \Cref{fig:types}. Here we give more details on our type setup.
Given type contexts $\Gamma$ and $\Gamma'$, the concatenation of them, denoted $\Gamma,\Gamma'$ is the type context given as follows:
\begin{align*}
    \Gamma,\Gamma'(x) =
    \begin{cases}
        \Gamma'(x)\ &\text{if}\ x\in \dom(\Gamma') \\
        \Gamma(x)\ &\text{if}\ x\in \dom(\Gamma)\setminus\dom(\Gamma').
    \end{cases}
\end{align*}

When concatenating linear type contexts $\Delta$ and $\Delta'$, we assert that $\dom(\Delta)\cap \dom(\Delta') = \emptyset$.
For each $o \in \mathcal{O}$, there is a unique signature $o : \ut_1,\ldots ,\ut_{n} \to \ut$ (note that this means operations do not apply to the affine types).
Our complete typing rules are shown in \Cref{rules:types}.
\begin{figure}
  \begin{mathpar}
    \inferrule*[right=T-Bool]
    { }
    {\Gamma;\Delta \vdash b : \mathsf{Bool}}
    \and
\inferrule*[right=T-Point]
    { }
    {\Gamma;\Delta \vdash r\pnt  : \mathsf{Point}}
    \and
    \inferrule*[right=T-Intvl]
    { }
    {\Gamma; \Delta \vdash [r, r'] : \mathsf{Intvl}}
    \and
    \inferrule*[right=T-RdIntvl]
    { }
    {\Gamma; \Delta \vdash \rd{[r, r']} : \rd{\mathsf{Intvl}}}
    \and
    \inferrule*[right=T-Piece]
    { }
    {\Gamma; \Delta \vdash P[r_1, r'_{1}],\ldots ,[r_n, r'_{n}] : \mathsf{Piece}}
    \and
    \inferrule*[right=T-RdPiece]
    { }
    {\Gamma; \Delta \vdash \rd{P[r_1, r'_{1}],\ldots ,[r_n, r'_{n}]} : \rd{\mathsf{Piece}}}
    \and
    \inferrule*[right=T-Vltn]
    { }
    {\Gamma; \Delta \vdash r_{1}\cdot V_{a_1}(P_1) + \cdots + r_{n}\cdot V_{a_{n}}(P_{n}) : \mathsf{Vltn}}
   \and
    \inferrule*[right=T-Cake]
    { }
    {\Gamma;\Delta \vdash \kwcake : \mathsf{Intvl}}
    \and
    \inferrule*[right=T-Ops]
    {%
    \Gamma; \Delta_1 \vdash e_{1} : \ut_1 \\
    \cdots \\
    \Gamma;\Delta_n \vdash e_n : \ut_n \\
    o : \ut_{1},\ldots ,\ut_{n} \to \ut
    }
    {\Gamma; \Delta_1 ,\ldots ,\Delta_{n} \vdash o(e_1, \ldots ,e_{n}) : \ut}
    \and
    \inferrule*[right=T-Tup]
    {%
    \Gamma;\Delta_1 \vdash e_1 : \tau_1 \\
    \cdots \\
    \Gamma;\Delta_{n} \vdash e_n : \tau_n
    }
    {\Gamma;\Delta_1,\ldots ,\Delta_{n} \vdash (e_1, \ldots ,e_{n}) : \tau_1 \times \ldots \times \tau_{n}}
    \and
    \inferrule*[right=T-Cond]
    {%
      \Gamma;\Delta_1\vdash e_1 : \mathsf{Bool} \\
      \Gamma;\Delta \vdash e_2 : \tau \\ \Gamma;\Delta \vdash e_3 : \tau
    }
    { \Gamma; \Delta_1,\Delta \vdash \kwif\ e_1\ \kwthen\ e_2\ \kwelse\ e_3 : \tau}
    \and
    \inferrule*[right=T-Var]
    { }
    {
        \Gamma, x : \ut; \Delta \vdash x : \ut
    }
    \and
\inferrule*[right=T-AffVar]
    { }
    {
        \Gamma; w : \lint \vdash w : \lint
    }
    \and
\inferrule*[right=T-Split]
    {%
      \Gamma; \Delta_1  \vdash e_1 : \ut_1\times \cdots \times \ut_{n} \times \lint_{1} \times \cdots \times \lint_{n'} \\
      \Gamma, x_1 : \ut_1,\ldots ,x_{n} : \ut_{n}, \rd{w_1} : \rd{\lint_1} ,\ldots ,\rd{w_{n}} : \rd{\lint_{n'}};\Delta_2,w_1 : \lint_1 ,\ldots ,w_{n} : \lint_{n'} \vdash e_2 : \tau
    }
    {\Gamma;\Delta_1 , \Delta_2 \vdash \kwlet\ x_1,\ldots ,x_{n}, w_1,\ldots ,w_{n} = \kwsplit\ e_1\ \kwin\ e_2 : \tau}
\and
    \inferrule*[right=T-Div]
    {%
      \Gamma;\Delta_1 \vdash e_1 : \mathsf{Intvl} \\
      \Gamma;\Delta_2 \vdash e_2 : \mathsf{Point} \\
    }
    {\Gamma;\Delta_1,\Delta_2 \vdash \ediv{e_1}{e_2} : \mathsf{Intvl} \times \mathsf{Intvl}}
    \and
    \inferrule*[right=T-Mark]
    {%
        \Gamma; \Delta_1 \vdash e_1 : \rd{\mathsf{Intvl}} \\
        \Gamma; \Delta_2 \vdash e_2 : \mathsf{Vltn} \\
    }
    {%
    \Gamma; \Delta_1, \Delta_2 \vdash \emark{a}{e_1}{e_2} : \mathsf{Point}
    }
    \and
    \inferrule*[right=T-EvalIntvl]
    {%
        \Gamma; \Delta \vdash e : \rd{\mathsf{Intvl}}
    }
    {\Gamma; \Delta \vdash \eeval{a}{e} : \mathsf{Vltn}}
    \and
    \inferrule*[right=T-EvalPc]
    {%
        \Gamma; \Delta \vdash e : \rd{\mathsf{Piece}}
    }
    {\Gamma; \Delta \vdash \eeval{a}{e} : \mathsf{Vltn}}
    \and
    \inferrule*[right=T-Piece]
    {%
        \Gamma;\Delta_1 \vdash e_1 : \mathsf{Intvl} \\
    \cdots \\
    \Gamma;\Delta_{n} \vdash e_n : \mathsf{Intvl}
    }
    {
        \Gamma ; \Delta_1,\ldots ,\Delta_{n}\vdash \kwpiece\ e_1, \ldots ,e_n : \mathsf{Piece}
    }
    \end{mathpar}
    \caption{Typing rules}
  \label{rules:types}
\end{figure}
\begin{proposition}
   For a value $\v$, there is a unique type $\tau$ such that $\cdot \vdash v : \tau$.
   \label{prop:typeunicity}
\end{proposition}
\begin{proof}
    We proceed by induction on the structure of values. For all values but those of the form $(\v_1,\ldots ,\v_{n})$, examining the typing rules show that each form of the value has only one typing rule that can be applied. Now assume that $\v = (\v_1,\ldots ,\v_{n})$. By induction, there are unique types $\tau_1,\ldots ,\tau_{n}$ such that $\cdot \vdash \tau_{i} : \v_{i}$ for $1 \leq i \leq n$. The only rule that can be applied here is \rname{T-Tup}, for which we obtain $\cdot \vdash (\v_1 ,\ldots ,\v_{n}) : \tau_1 \times \cdots \times \tau_{n}$.
\end{proof}
Let $\sem{\tau} \triangleq \{\v \in \vals \mid \cdot \vdash \v : \tau\}$, that is, the set of values that have type $\tau$. Clearly, if $\tau \neq \tau'$, then $\sem{\tau}\cap \sem{\tau'} = \emptyset$. Because of \Cref{prop:typeunicity}, we can concretely describe $\sem{\tau}$ for each $\tau$:
\[
\begin{array}{c c}
    \sem{\mathsf{Bool}} = \{\kwtrue, \kwfalse\} \quad & \quad
    \sem{\mathsf{Point}} = \{r\pnt \mid r \in \mathbb{R}\}
    \end{array}
\]
\[
        \sem{\mathsf{Vltn}} = \{\sum_{i}r_{i} \cdot V_{a_{i}}(P_{i})\mid r_{i} \in \mathbb{R}, a_{i}\in \agents, P_{i}\in \sem{\mathsf{Piece}}\}
    \]
\[
    \begin{array}{c c}
    \sem{\rd{\mathsf{Intvl}}} = \{\rd{\v} \mid \v \in \sem{\mathsf{Intvl}}\} \quad & \quad
    \sem{\rd{\mathsf{Piece}}} = \{\rd{\v} \mid \v \in \sem{\mathsf{Piece}}\}
    \end{array}
\]
\[
    \begin{array}{c c}
    \sem{\mathsf{Intvl}} = \{[r,r']\mid r, r' \in \mathbb{R}, r \leq r'\} \quad & \quad
    \sem{\mathsf{Piece}} = \{P_{i} \v_i \mid \v_{i} \in \sem{\mathsf{Intvl}}\}
    \end{array}
\]
\[
    \sem{\ut_1\times \cdots \times \ut_{n} \times \tau_1 \times \cdots \times \tau_{n'}} = \sem{\ut_1} \times \cdots \times \sem{\ut_{n}} \times \sem{\tau_1} \times \cdots \times \sem{\tau_{n'}}
\]

\subsection{Semantics}
\label{app:sem}
Before we can state the complete rules, we require a few definitions.
For each $o: \ut_1 ,\ldots ,\ut_{n}\to \ut \in O$, we assume there is a function $\sem{o} : \bigtimes_{i = 1}^{n} \unrd{\sem{\ut_{i}}}\to \unrd{\sem{\ut}}$. In this way, operations are defined not on the read only versions, but on the underlying values, which assists in the compatibility with the logic.
Recall that our semantics is given by a relation $\eval_{\vset}\subseteq \exps \times \vals$, indexed by a set of valuations $\vset$ (which we omit when the set in question is clear).
The big-step rules are given in \Cref{rules:e}.

\begin{figure}
  \begin{mathpar}
    \inferrule*[right=E-Val]
    { }
    {\v \eval \v}
    \and
    \inferrule*[right=E-Cake]
    { }
    {\kwinit \eval \left[ 0,1 \right]}
    \\
    \inferrule*[right=E-Ops]
    {%
    e_1 \eval \v_1 \\
    \cdots \\
    e_n \eval \v_n \\
    \sem{o}(\unrd{\v_1},\ldots ,\unrd{\v_{n}}) = \v
    }
    {o(e_1, \ldots ,e_{n}) \eval \rd{\v}}
    \\
    \inferrule*[right=E-Tup]
    {%
    e_1 \eval \v_1 \\
    \cdots \\
    e_n \eval \v_n
    }
    {(e_1, \ldots ,e_{n}) \eval (\v_1,\ldots ,\v_{n})}
    \and
    \inferrule*[right=E-True]
    {%
      e_1 \eval \kwtrue \\
      e_2 \eval \v_2
    }
    { \kwif\ e_1\ \kwthen\ e_2\ \kwelse\ e_3 \eval \v_2}
    \and
    \inferrule*[right=E-False]
    {%
      e_1 \eval \kwfalse \\
      e_3 \eval \v_3
    }
    { \kwif\ e_1\ \kwthen\ e_2\ \kwelse\ e_3 \eval \v_3}
    \and
    \inferrule*[right=E-Split]
    {%
        e_1 \eval (\v_1,\ldots ,\v_{n + n'}) \\
        e_2\{x_{i} \mapsto \v_i \mid 1\leq i \leq n\}\{\rd{w_{i}} \mapsto \rd{\v_{i + n}}\mid 1\leq i \leq n'\}\{w_{i} \mapsto \v_{i + n}\mid 1\leq i \leq n'\} \eval \v
    }
    {\kwlet\ x_1,\ldots ,x_{n}, w_{1},\ldots ,w_{n'} = \kwsplit\ e_1\ \kwin\ e_2 \eval \v}
\and
    \inferrule*[right=E-Div]
    {%
      e_1 \eval [r_1, r_1'] \\
      e_2 \eval r_2 \\
      r_1 \leq r_2 \leq r_1'\\
    }
    {\ediv{e_1}{e_2} \eval ([r_1, r_2], [r_2, r_1'])}
    \\
    \inferrule*[right=E-Mark]
    {%
        e_1 \eval [r_1,r'_1] \\
        e_2 \eval {\textstyle\sum_{i}r_{i}}V_{a_{i}}(P_{i}) \\
        \vset_a^{}([r_1,r]) = {\textstyle \sum_{i}r_{i}V_{a_{i}}(P_{i})}
    }
    {%
        \emark{a}{e_1}{e_2} \eval r
    }
    \and
    \inferrule*[right=E-EvalIntvl]
    {%
      e \eval [r, r']
    }
    {\eeval{a}{e} \eval \vset_a[r, r']}
    \and
    \inferrule*[right=E-EvalPc]
    {%
      e \eval P\ [r_1, r_1'] , \ldots, [r_n, r_n']
    }
    {\eeval{a}{e} \eval \vset_a(P [r_1, r_{1}'],\ldots ,[r_n, r_{n}'])}
    \and
    \inferrule*[right=E-Piece]
    {%
        e_1 \eval [r_1, r_1'] \\
    \cdots \\
    e_n \eval [r_{n}, r_{n}']
    }
    {
        \kwpiece\ e_1, \ldots ,e_n \eval P\ [r_1, r_1'] , \ldots, [r_n, r_n']
    }
\end{mathpar}
        \caption{Big-step semantics}
  \label{rules:e}
\end{figure}

\hypertarget{dnotation}{}
Here we also introduce the following derivation notation, which is helpful later on. Let $D$ be a derivation of an evaluation judgement. If the conclusion of $D$ is $e \eval_{\vset} \v$, we write $D : e \eval_{\vset} \v$. Given a derivation $D: e \eval_{\vset} \v$, if $e$ has subexpressions $e_{i}$, we let $D_{i}$ be the derivation of the premise of $e \eval_{\vset}\v$ that contains $e_{i}$. Similarly, if $e$ is an expression with one subexpression $e'$, we let $D'$ denote the derivation of the premise containing $e'$. For example, if $D: (e_1 ,\ldots ,e_{n}) \eval_{\vset} (\v_1 ,\ldots, \v_{n})$, then $D_{i} : e_{i} \eval_{\vset} \v_{i}$, and if $D : \eeval{a}{e'}\eval V_{a}([r,r'])$, then $D' : e' \eval_{\vset} [r,r']$.

 \subsection{Type soundness}
 \label{app:typesoundness}
 Here we prove type soundness. As a reminder, the statement is the following:
\begin{restatable}[Type soundness]{proposition}{typesoundness}
\label{prop:typesoundness}
If $\cdot \vdash e : \tau$ then $e \eval \v$ implies $\cdot \vdash \v : \tau$.
\end{restatable}
 We first provide some definitions.
 For each type $\tau$, we let $R_{\tau}^{\vdash} \triangleq \{e \in \exps \mid \cdot \vdash e : \tau, e \eval \v \Rightarrow \v \in \sem{\tau}\}$, that is, the set of closed expressions that have type $\tau$ and evaluate only to values within $\sem{\tau}$. We now build up some lemmas in order to prove type soundness.
\begin{lemma}
   $\v \in \sem{\tau} \iff \cdot \vdash \v : \tau \iff \v \in R_{\tau}^{\vdash}$.
   \label{prop:valtype}
\end{lemma}
\begin{proof}
    By inspection of definition of $\sem{\tau}$ and value typing rules.
\end{proof}
\begin{lemma}
        Suppose $\cdot \vdash e : \tau$. If $e \eval \v \Rightarrow \v \in R^{\vdash}_{\tau}$, then $e \in R^{\vdash}_{\tau}$.
    \label{prop:backreduction}
    \end{lemma}
    \begin{proof}
        Follows immediately from \Cref{prop:valtype}.
    \end{proof}

    We will aim to prove type soundness by induction on the typing derivation. For the proof to go through, we need to generalize the statement to open expressions.
Let $\gamma$ be an expression substitution replacing variables only with values. We say that $\gamma \psat \Gamma$ if $\dom(\gamma)\subseteq \dom(\Gamma)$ and $\gamma(x) \in R^{\vdash}_{\Gamma(x)}$ for all $x\in \dom(\gamma)$. We say that $\gamma \vDash \Gamma$ and refer to $\gamma$ as a \emph{closing} substitution if $\gamma \psat \Gamma$ and $\dom(\gamma) = \dom(\Gamma)$. For linear type contexts, instead of using $\gamma$ to denote a substitution, we use $\delta$. So given a closing substitution $\gamma$ and a linear substitution $\delta$, we let $\gamma;\delta$ be the substitution that combines them together.
\begin{proposition}
\label{prop:typesoundness}
    Suppose $\Gamma;\Delta \vdash e : \tau$.
    Let $\gamma,\delta$ be such that $\gamma \vDash \Gamma$ and $\delta \vDash \Delta$. Then $\gamma; \delta(e) \in R_{\tau}^{\vdash}$.
\end{proposition}
\begin{proof}
    This proceeds by induction on the typing derivation.
    Suppose that $\gamma \vDash \Gamma$,  $\delta \vDash \Delta$, and $\Gamma;\Delta \vdash e : \tau$.
    We break up our argument in cases based on the last typing rule used.
\begin{description}
    \item[\rname{T-Ops}] This is immediate since we assume $\sem{o} : \unrd{\sem{\ut_1}} \times \cdots \times \unrd{\sem{\ut_{n}}} \to \unrd{\sem{\ut}}$, meaning if $\gamma;\delta(o(e_1,\ldots ,e_{n})) \eval \v$, then $\v = \rd{\sem{o}(\unrd{\v_1},\ldots ,\unrd{\v_{n}})}\in \sem{\ut} = \sem{\tau}$, where $\gamma;\delta(e_{i}) \eval \v_{i}$.
    \item[\rname{T-Tup}] Then $\tau =  \tau_1 \times \cdots \times \tau_{n}$, $\Delta = \Delta_1,\ldots ,\Delta_{n}$, and $\Gamma ;\Delta_{i} \vdash e_{i} : \tau_{i}$ for $1 \leq i \leq n$, where we would have $e = (e_1 ,\ldots ,e_{n})$.
        Suppose that $\gamma;\delta (e) \eval \v$. We need to show that  $\v \in \sem{\tau_1 \times \cdots \times \tau_{n}}$. By our typing rule, we know that $\delta$ can be partitioned as $\delta_1,\ldots ,\delta_{n}$ where
 \begin{equation}
        \gamma;\delta(e) = (\gamma;\delta_{1}(e_1),\ldots ,\gamma;\delta_{n}(e_{n}))
            \label{eq:subpart}
        \end{equation}
        and $\gamma;\delta(e_{i}) \eval \v_{i}$, $\v = (\v_1 ,\ldots , \v_{n})$.
        By induction $\gamma;\delta_{i}(e_{i}) \in R^{\vdash}_{\tau_{i}}$. By defintiion of $R^{\vdash}_{\tau_{i}}$, this means that $\v_{i} \in \sem{\tau_{i}}$. This is enough to conclude that $\v \in \sem{\tau}$ and therefore $\gamma;\delta(e) \in R^{\vdash}_{\tau}$.
    \item[\rname{T-Piece}] The argument for this case is nearly identical to that of $\rname{T-Tup}$.

    \item[\rname{T-Split}] Then $e = \kwlet\ x_1 ,\ldots ,x_{n},w_1,\ldots ,w_{n'} = e_1\ \kwin\ e_2$.
        Set
        \begin{align*}
            \tau_1  &= \ut_1 \times \cdots \times \ut_{n} \times \lint_1 \times \cdots \times \lint_{n'}   \\
            \Gamma' &= \Gamma, x_1: \ut_1,\ldots ,x_{n} : \ut_{n},\rd{w_1} : \rd{\lint_1},\ldots , \rd{w_{n'}} : \rd{\lint_{n'}}\\
            \Delta' &= \Delta_2, w_1 : \lint_1 ,\ldots ,w_{n}:\lint_{n'}.
        \end{align*}
        This means we have $\Gamma;\Delta_1 \vdash e_1 : \tau_1$ and $\Gamma';\Delta'\vdash e_2 : \tau$. If $\gamma;\delta(e) \eval \v$, then we would have
        \begin{align*}
          \gamma;\delta (e_1)\eval (\v_1 ,\ldots , \v_{n}, \ldots ,\v_{n + n'})
      \end{align*}
      and
      \begin{align*}
          \gamma;\delta(e_2)\{x_{i} \mapsto \v_{i}\mid 1 \leq i \leq n\} \{w_{i}\mapsto \v_{i + n},\rd{w_{i}}\mapsto \rd{\v_{i + n}} \mid 1 \leq i \leq n'\}\eval \v
        \end{align*}
        We can partition $\delta$ up into $\delta_1$ and $\delta_2$ such that $\delta_1 \vDash \Delta_1$ and $\delta_2 \vDash \Delta_2$ and such that $\gamma;\delta_1(e_1) = \gamma;\delta(e_1)$ and $\gamma;\delta_2(e_2) = \gamma;\delta(e_2)$.
        By induction, $\gamma;\delta_1 (e_1)\in R_{\tau_1}^{\vdash}$. Then $(\v_1,\ldots ,\v_{n + n'})\in \sem{\tau_1}$, which means that $\v_{i} \in \sem{\tau_{i}}$ if $1 \leq i \leq n$ and $\v_{i + n} \in \sem{\lint_{i}}$, $\rd{\v_{i + n}} \in \sem{\rd{\lint_{i}}}$ if $1 \leq i \leq n'$. So if we set
        \[
            \gamma' = \gamma_2\{x_{i} \mapsto \v_{i}\mid 1 \leq i \leq n\}\{\rd{w_{i}}\mapsto \rd{\v_{i + n}}\mid 1 \leq i \leq n\}
        \]
        and
        \[
            \delta' = \delta_2\{w_{i}\mapsto \v_{n+ i}\mid 1 \leq i \leq n'\},
    \]
        then $\gamma'\vDash \Gamma'$ and $\delta' \vDash \Delta'$. By induction, $\gamma';\delta'(e_2) \in R_{\tau}^{\vdash}$. Moreover,
    \[
        \gamma';\delta'(e_2) = \gamma;\delta(e_2)\{x_{i} \mapsto \v_{i}\mid 1 \leq i \leq n\}\{w_{i}\mapsto \v_{n+ i}, \rd{w_{i}}\mapsto \rd{\v_{n + i}} \mid 1 \leq i \leq n'\}.
    \]
    This gives us $\gamma';\delta'(e_2)\eval \v$.
        Thus we have $\v\in \sem{\tau}$, so by \Cref{prop:valtype}, $\v \in R_{\tau}^{\vdash}$.

        The above argument establishes for any $\v$ such that $\gamma;\delta(e)\eval \v$, we have $\v \in R_{\tau}^{\vdash}$. By \Cref{prop:backreduction}, $\gamma;\delta (e) \in R_{\tau}^{\vdash}$.

    \item[\rname{T-Var}] This follows easily as $\gamma \vDash \Gamma$.
    \item[\rname{T-AffVar}] This also follows easily as $\delta \vDash w : \lint$.

    \item[\rname{T-Div}] Then $e = \ediv{e_1}{e_2}$ and $\Delta = \Delta_1,\Delta_2$, with $\Gamma; \Delta_1 \vdash e_1 : \mathsf{Intvl} $ and $\Gamma; \Delta_2 \vdash e_2 : \mathsf{Point}$.

        Suppose that $\gamma;\delta(e) \eval \v$.

        Then $\gamma;\delta(e_1) \eval [r_1,r'_{1}]$, $\gamma;\delta(e_2) \eval r_2\pnt$ with $r_1 \leq r_2\leq r'_{1}$ and $\v = ([r_1, r_2], [r_2, r'_{1}])$.
        It is immediate that $\v \in R^{\vdash}_{\mathsf{Intvl}\times \mathsf{Intvl}}$.
    \item[\rname{T-Mark}]
        Then $e = \emark{a}{e_1}{e_2}$, $\tau = \mathsf{Point}$, $\Delta = \Delta_1,\Delta_2$, and $\Gamma; \Delta_1 \vdash e_1 : \rd{\mathsf{Intvl}}$, $\Gamma; \Delta_2 \vdash e_2 : \mathsf{Real}$.
        If $\gamma;\delta (e) \eval \v$, then $\v = r'\pnt$ for some real $r'$. We are already done as $\v \in \mathsf{PT}$.
    \item[\rname{T-EvalPc}] Then $e = \eeval{a}{e'}$, $\tau = \mathsf{Vltn}$, and $\Gamma;\Delta \vdash e' : \mathsf{Piece}$. By induction, $\gamma;\delta(e') \in R^{\vdash}_{\mathsf{Piece}}$. This means that if $\gamma;\delta(e) \eval \v$, then  $\gamma;\delta(e') \eval P$ for some piece $P$, so $\v = V_{a}(P)$. Thus $\v \in \mathbb{V}$. This allows us to conclude $\gamma;\delta(e) \in R^{\vdash}_{\mathsf{Vltn}}$.
    \item[\rname{T-EvalIntvl}] Similar to \rname{T-EvalPc}.
\end{description}
\end{proof}
As a corollary, we recieve type soundness.
%\typesoundness*

 \subsection{Disjointness}
 \label{app:disjointness}
 We begin by defining disjointness for values and expressions. To do so, we introduce \emph{interval lists}, which are just finite lists of intervals. For a value $\v$, we let $I(\v)$ denote the \emph{interval list of $\v$}, which is defined by induction on the structure of $\v$:
\begin{mathpar}
    I([r, r']) \triangleq [r, r']
    \and
    I(P_{i} [r_{i}, r_{i}']) \triangleq  [r_1, r'_1] ,\ldots , [r_n, r'_n]
    \and
    I((\v_1 ,\ldots , \v_{n})) \triangleq I(\v_1), \ldots , I(\v_{n})
    \and
    I(\rd{[r, r']}) \triangleq \emptyset
    \and
    I(\rd{P_{i} [r_{i}, r_{i}']}) \triangleq  \emptyset
    \and
    I(\rd{r\pnt}) \triangleq \emptyset
    \and
    I(r\pnt) \triangleq \emptyset
    \and
    I(\kwtrue) \triangleq \emptyset
    \and
    I(\kwfalse) \triangleq \emptyset
\end{mathpar}

Let $e$ be an expression. The \emph{interval list of $e$}, denoted $I(e)$, is given as follows by induction on the structure of $e$:  \begin{mathpar}
        I(x) \triangleq \emptyset
        \and
        I(w) \triangleq \emptyset
        \and
        I(\rd{w}) \triangleq \emptyset
        \and
        I( (e_1 ,\ldots ,e_{n})) \triangleq I(e_1),\ldots ,I(e_{n})
        \and
        I( o(e_1 ,\ldots ,e_{n})) \triangleq I(e_1),\ldots ,I(e_{n})
        \and
        I(\kwif\ e_1 \kwthen\ e_2\ \kwelse\ e_3) \triangleq I(e_1),I(e_2),I(e_3)
        \and
        I(\kwlet\ x_1,\ldots ,x_{n}, w_1,\ldots ,w_{n'} = \kwsplit\ e_1\ \kwin\ e_2) \triangleq I(e_1), I(e_2)
        \and
        I(\kwcake) \triangleq [0,1]
        \and
        I(\ediv{e_1}{e_2}) \triangleq I(e_1),I(e_2)
        \and
        I(\kwpiece(e_1 ,\ldots ,e_{n})) \triangleq I(e_1) ,\ldots ,I(e_{n})
        \and
        I(\eeval{a}{e}) \triangleq I(e)
        \and
        I(\emark{a}{e_1}{e_2}) \triangleq I(e_1), I(e_2)
 \end{mathpar}
We also define interval lists for variable substitutions. Let $S$ be a variable substitution. Then $I(S) \triangleq I(e_1) ,\ldots ,I(e_{n})$ if $S = \{w_{i} \mapsto e_{i} \mid 1 \leq i \leq n\}$. That is, $I(S)$ is the list of intervals in the range of $S$.
We may also write $\abs{L}$ to denote the set of intervals in $L$.

An interval list $L$ is \emph{disjoint} if $L_{i}$ and $L_{j}$ are disjoint for all $i \neq j$. Finally, a value $\v$ is \emph{disjoint} if $I(\v)$ is disjoint and an expression $e$ is \emph{disjoint} if $I(e)$ is disjoint.

We now argue for the following proposition:
\disjointness*

Our technique will be similar to what was used for type soundness, though additional work is involved.
    First we extend $\delta \vDash \Delta$ to also reflect disjointness.
\begin{definition}
    Now let $P$ be a piece and $\Delta$ a linear type context. We write $\delta \psatl{P} \Delta$ if
\begin{enumerate}
    \item $\dom(\delta) \subseteq  \dom(\Delta)$
    \item $\delta(w) \in R_{{\Delta(w)}}$ for all $w\in \dom(\delta)$
    \item $I(\delta)$ is disjoint
    \item $\cup I(\delta)$ is disjoint with $P$ (besides possibly at finitely many points)
\end{enumerate}
If $\dom (\delta)= \dom(\Delta)$, we write $\delta \vDash^{P} \Delta$.
\end{definition}

We introduce some helpful set notation. Given a collection of sets $C$, we let $\cup C$ denote $\cup_{c \in C}c$. In particular, if $I$ is a set of intervals, then $\cup I$ is the piece formed by all intervals within $I$.
\begin{lemma}
    If $e \eval \v$, then $\cup I(\v) \subseteq \cup I(e)$.
    \label{prop:intcont}
\end{lemma}
\begin{proof}
    This goes by induction on the derivation of $e\eval \v$. We only show the interesting cases.
    \begin{description}
        \item[\rname{E-Cake}] This is clear.
        \item[\rname{E-Ops}] Then $\v = \rd{\sem{o}(\unrd{\v_1},\ldots ,\unrd{\v_{n}})}$, $e = o(e_1,\ldots ,e_{n})$, $e_{i} \eval \v_{i}$. We have $o : \ut_1 ,\ldots ,\ut_{n} \to \ut$ for some non-linear types $\ut_1,\ldots ,\ut_{n}, \ut$. We must have $\v \in \sem{\ut}$. All values contained in $\sem{\ut}$ for any non-linear type $\ut$ have their interval lists empty. Therefore, $I(\v) = \emptyset$.
        \item[\rname{E-Div}] Then $\v = ([r_1, r_2], [r_2, r'_{1}])$ where $e_1 \eval [r_1,r'_{1}]$ and $e_2 \eval r_2$. Then $\cup I(\v)\subseteq [r_1,r'_{1}]\subseteq \cup I(e_1)\subseteq \cup I(e)$, where the second containment follows by induction.

        \item[\rname{E-Piece}] Then $\v = P [r_1, r'_{1}],\ldots ,[r_{n},r'_{n}]$ and $b = \kwpiece\ e_1 ,\ldots ,e_{n}$ where $e_i \eval [r_{i},r'_{i}]$. By induction, $[r_{i},r'_{i}]\subseteq \cup I (e_{i})$. Since $\cup I(e) = \cup I(e_1) \cup \cdots \cup I (e_{n})$ and $\cup I (\v) = \bigcup_{i}[r_{i},r'_{i}]$, we conclude $\cup I (\v) \subseteq \cup I(e)$.

        \item[\rname{E-Split}] Then $e = \kwlet\ x_1,\ldots ,x_{n},w_{1},\ldots ,w_{n'} = \kwsplit\ e_1\ \kwin\ e_2$ and $e_1 \eval (\v_1 ,\ldots ,\v_{n + n'})$, and
            \[
            e_2 \{x_{i} \mapsto \v_{i} \mid 1 \leq i \leq n\}\{w_{i} \mapsto \v_{n + i} \mid 1 \leq i \leq n'\}  \eval \v.
        \]
            Let $e'_2$ denote the expression in the above judgement. By induction, $\cup I(\v_1 ,\ldots ,\v_{n + n'})\subseteq \cup I (e_1)$ and $\cup I(\v) \subseteq \cup I(e'_{2})$. Note that $\cup I(e'_{2}) = \cup (I(\v_1 ,\ldots ,\v_{n+ n'})\cup I(e_2))\subseteq \cup (I(e_1)\cup I (e_2)) = \cup I(e)$. This concludes the current case.
    \end{description}
\end{proof}
\begin{lemma}
    \label{prop:utilistempty}
    Let $\v \in \sem{\ut}$. Then $I(\v) = \emptyset$.
\end{lemma}
\begin{proof}
    By inspection of definition of $I(\v)$ for values in $\sem{\ut}$ for some non-linear type $\ut$. \qed
\end{proof}

\begin{lemma}
    Suppose that $\gamma \vDash \Gamma$. Then $I(\gamma) = \emptyset$.
    \label{prop:unrestrictedctxnoint}
\end{lemma}
\begin{proof}
    Take $x \in \dom(\gamma)$. As $\gamma\vDash \Gamma$, $\gamma(x) \in R^{\vdash}_{\Gamma(x)}$. By \Cref{prop:valtype}, $\gamma(x)\in \sem{\Gamma(x)}$. By definition, $\Gamma(x) = \ut$ for some $\ut$. By \Cref{prop:utilistempty}, $I(\gamma(x)) = \emptyset$. Thus we can conclude $I(\gamma) = \emptyset$. \qed
\end{proof}
\begin{lemma}
    For any expression $e$ and substitution $S$, $\cup I(S(e)) = (\cup I(e)) \cup (\cup I(S))$.
    \label{prop:subpiecebreakdown}
\end{lemma}
\begin{proof}
    Straightforward induction on the structure of $e$. \qed
\end{proof}

\begin{lemma}
    Let $\Delta = \Delta_1 ,\ldots ,\Delta_{n}$ be a linear type context.  Set $\delta_{i} = \delta|_{\dom(\Delta_{i})}$. The following three statements hold:
    \begin{enumerate}[(a)]
        \item If $\delta \vDash \Delta$ then  $\delta_{i} \vDash \Delta_{i}$
        \item If $\delta \psatl{P} \Delta$, then $\delta_{i} \psatl{P} \Delta_{i}$ for all $1 \leq i \leq n$.
        \item If $\delta \vDash^{P} \Delta$, then $\delta_{i} \vDash^{P} \Delta_{i}$ for all $1 \leq i \leq n$.
    \end{enumerate}
    \label{prop:partitiondelta}
\end{lemma}
\begin{proof}
    We first argue for (b), and then show how to use the arguments for (a) and (c).
    Suppose that $\delta \psatl{P} \Delta$. We must show 1-4 for each $\delta_{i}$.
    \begin{enumerate}
        \item As $\dom(\delta) \subseteq \dom(\Delta)$, we must have that $\dom(\delta_{i}) = \dom(\delta|_{\dom(\Delta_{i})}) \subseteq \dom(\Delta_{i})$ since $\dom(\Delta_{i}) \subseteq \dom(\Delta)$.
        \item For each $w \in \dom(\delta_{i})$, we have
    \[
        \delta_{i}(w) = \delta(w) \in R_{{\Delta(w)}} = R_{{\Delta_{i}(w)}}
    \]
    where the last equality uses that $\dom(\delta_{i}) \subseteq  \dom(\Delta_{i})$.

\item Since $\delta_{i}$ is a restriction of $\delta$, we also have that $I(\delta_{i})$ is disjoint as $I(\delta)$ is disjoint.

\item Since $\delta_{i}$ is a restriction of $\delta$, $\cup I(\delta_{i})\subseteq \cup I (\delta)$. As $I(\delta)$ is disjoint with $P$, it must be that $\cup I (\delta)$ is also disjoint with $P$.
    \end{enumerate}
    These four parts establish that $\delta_{i} \psatl{P} \Delta_{i}$.
    If $\dom(\delta) = \dom(\Delta)$, then $\dom(\delta_{i}) = \dom(\delta|_{\dom(\Delta_{i})}) = \dom(\Delta|_{\dom(\Delta_{i}}) = \dom(\Delta_{i})$. This, along with 2 above establishes (a). The same fact, along with 2,3, and 4 establishes (c).
    \qed
\end{proof}

\begin{lemma}
    Suppose that $I(e)$ is disjoint, $\Gamma;\Delta \vdash e : \tau$,  $\gamma\psat \Gamma$, $\delta \psatl{P}\, \Delta$ and $\cup I(e)\subseteq P$. Then $\gamma;\delta(e)$ is disjoint.
    \label{prop:subdisjoint}
\end{lemma}
\begin{proof}
    This proceeds by induction on the derivation of $\Gamma;\Delta\vdash e : \tau$.
    \begin{description}
        \item[\rname{T-Ops}] Then $e = o(e_1 ,\ldots ,e_{n})$, $\Gamma,\Delta_{i} \vdash e_{i} : \ut_{i}$, $\tau = \ut$ where $o : \ut_1 ,\ldots ,\ut_{n} \to \ut$, $\Delta = \Delta_1 ,\ldots ,\Delta_{n}$.
         According to \Cref{prop:partitiondelta}, we can partition $\delta$ into $\delta_1 ,\ldots ,\delta_{n}$ such that $\delta_{i} \psatl{P} \Delta_{i}$. Morover, we have $\gamma;\delta(e_{i}) = \gamma;\delta_{i}(e_{i})$. Then by induction, $\gamma;\delta(e_{i})$ is disjoint. By \Cref{prop:unrestrictedctxnoint} and \Cref{prop:subpiecebreakdown}, we have that
            \[
                \cup I(\gamma;\delta_{i}(e_{i})) \subseteq (\cup I(\delta_{i})) \cup (\cup I(e_{i})).
            \]
            As $I(\delta)$ is disjoint, we have that $I(\delta_{i})$ is disjoint from $I(\delta_{j})$ if $i\neq j$. As $I(e)$ is disjoint, we have that $I(e_{i})$ is disjoint from $I(e_{j})$ if $i \neq j$. Therefore  $\gamma;\delta_{i}(e_{i})$ is disjoint from $\gamma;\delta_{j}(e_{j})$ if $i \neq j$. These facts allow us to conclude that $\gamma;\delta(e)$ is disjoint.

        \item[\rname{T-Split}]
            Then $e = \kwlet\ x_1,\ldots ,x_{n} , w_1 ,\ldots ,w_{n'} = \kwsplit\ e_1\ \kwin\ e_2$, $\Delta = \Delta_1, \Delta_2$ where
            \begin{align*}
                \Gamma;\Delta_1 \vdash e_1 : \ut_1 \times \cdots \times \ut_{n} \times \lint_1 \times \cdots \times \lint_{n'}
            \end{align*}
            and
            \begin{align*}
                \Gamma, \Gamma'; \Delta_2, \Delta' \vdash e_2 : \tau.
            \end{align*}
            where
            \begin{align*}
                \Gamma ' &= x_1 : \ut_1, \ldots,  x_{n} : \ut_{n}, \rd{w_1} : \rd{\lint_1}, \ldots, \rd{w_{n'}} : \rd{\lint_{n'}}  \\
                \Delta' &= w_1 : \lint_1 ,\ldots ,w_{n} : \lint_{n'}.
            \end{align*}
            Also let $\gamma_2$ be $\gamma$ but with $x_1,\ldots ,x_{n}, \rd{w_{1}},\ldots ,\rd{w_{n'}}$ removed.
            According to \Cref{prop:partitiondelta}, we can partition $\delta$ into $\delta_1 $ and $\delta_2$ such that $\delta_1\psatl{P} \Delta_1$, $\delta_2 \psatl{P} \Delta_2$ and $\gamma;\delta(e_1) = \gamma;\delta_1(e_1)$, $\gamma_2;\delta(e_2) = \gamma_2;\delta_2(e_2)$. As $I(e_1) \subseteq I(e)$, we have $\cup I(e_1) \subseteq P$, so we can apply induction to obtain that $\gamma;\delta_1(e_1)$ is disjoint.

            As $I(e_2) \subseteq I(e)$, we have $\cup I(e_2) \subseteq P$. Therefore we can similarly apply induction to obtain $\gamma_2;\delta_2(e_2)$ is disjoint.

            By \Cref{prop:unrestrictedctxnoint} and \Cref{prop:subpiecebreakdown}, we have
            \begin{align*}
                \cup I(\gamma;\delta_1(e_1)) \subseteq (\cup I(\delta_1))\cup (\cup I(e_1))\quad \text{and} \quad \cup I(\gamma_2;\delta_2(e_2)) \subseteq (\cup I(\delta_2))\cup (\cup I(e_2)).
            \end{align*}
            As $I(\delta)$ is disjoint, $\cup I(\delta_1)$ is disjoint from $\cup I(\delta_2)$. As $I(e)$ is disjoint, $\cup I(e_1)$ is disjoint from $\cup I (e_2)$. This allows us to conclude that $I(\gamma;\delta_1(e_1))$ is disjoint from $I(\gamma_2;\delta_2(e_2))$. Note
            \[
\gamma;\delta(e) = \kwlet\ x_1,\ldots ,x_{n} , w_1 ,\ldots ,w_{n'} = \kwsplit\ \gamma;\delta_1(e_1)\ \kwin\ \gamma_2;\delta_2(e_2)
\]
and because of this, we can conclude that $\gamma;\delta(e)$ is disjoint.
        \item[\rname{T-AffVar}]
            Then $\Gamma;w : \lint \vdash w : \lint$ and $\gamma;\delta(e) = \delta(w)$. By assumption, $I(\delta)$ is disjoint, so this case is concluded.

        \item[\rname{T-Tup}] Then $e = (e_1 ,\ldots ,e_{n})$, $\Delta = \Delta_1 ,\ldots ,\Delta_{n}$, and we have $\Gamma;\Delta_{i} \vdash e_{i} : \tau_{i}$ for $1 \leq i \leq n$. According to \Cref{prop:partitiondelta}, we can partition $\delta$ into $\delta_1 ,\ldots ,\delta_{n}$ such that $\delta_{i} \psatl{P} \Delta_{i}$. Morover, we have $\gamma;\delta(e_{i}) = \gamma;\delta_{i}(e_{i})$. Then by induction, $\gamma;\delta(e_{i})$ is disjoint. By \Cref{prop:unrestrictedctxnoint} and \Cref{prop:subpiecebreakdown}, we have that
            \[
                \cup I(\gamma;\delta_{i}(e_{i})) \subseteq (\cup I(\delta_{i})) \cup (\cup I(e_{i})).
            \]
            As $I(\delta)$ is disjoint, we have that $I(\delta_{i})$ is disjoint from $I(\delta_{j})$ if $i\neq j$. As $I(e)$ is disjoint, we have that $I(e_{i})$ is disjoint from $I(e_{j})$ if $i \neq j$. Therefore  $\gamma;\delta_{i}(e_{i})$ is disjoint from $\gamma;\delta_{j}(e_{j})$ if $i \neq j$. These facts allow us to conclude that $\gamma;\delta(e)$ is disjoint.
    \end{description}
    All other cases are either trivial, or very similar to \rname{Tup}.
    \qed
\end{proof}

\begin{proposition}
    Suppose $\Gamma;\Delta \vdash e : \tau$.
    Let $\gamma,\delta$ be such that $\gamma \vDash \Gamma$ and $\delta \vDash^{P} \Delta$ and $\cup I(e)\subseteq P$. If $I(e)$ is disjoint, then $\gamma;\delta (e) \in R_{\tau}$.
    \label{prop:disjointness}
\end{proposition}
\begin{proof}

    Suppose that $\Gamma;\Delta \vdash e : \tau$, $\gamma \vDash \Gamma$, $\delta \vDash^{P} \Delta$, $\cup I(e)\subseteq P$, and $I(e)$ disjoint. We already immediately have by \Cref{prop:typesoundness} that $\gamma;\delta(e)\in R^{\vdash}_{\tau}$. This means we only have to argue for disjointness of the value evaluated here.
    So suppose that $\gamma;\delta(e) \eval \v$.
    This argument goes by induction on the derivation of $\Gamma;\Delta \vdash e : \tau$ and we break up our argument into cases based on the last typing rule applied.
\begin{description}
    \item[\rname{T-Tup}] Then $\tau =  \tau_1 \times \cdots \times \tau_{n}$, $\Delta = \Delta_1,\ldots ,\Delta_{n}$, and $\Gamma ;\Delta_{i} \vdash e_{i} : \tau_{i}$ for $1 \leq i \leq n$, where $e = (e_1 ,\ldots ,e_{n})$. Also, $\v = (\v_1 ,\ldots ,\v_{n})$ and $\gamma;\delta(e_{i})\eval \v_{i}$ for $1 \leq i \leq n$.

        Since we write $\Delta = \Delta_1,\ldots ,\Delta_{n}$,  by \Cref{prop:partitiondelta}, we can partition $\delta$ into $\delta_{1},\ldots ,\delta_{n}$ such that $\delta_{i} \vDash^{P} \Delta_{i}$. As $\Gamma;\Delta_{i} \vDash^{P} e_{i} : \tau_{i}$, it must be that $\gamma;\delta_{i}(e_{i}) = \gamma;\delta (e_{i})$.

         As $\cup I(e)\subseteq P$, $\cup I(e_{i})\subseteq P$. By \Cref{prop:subdisjoint}, $\gamma;\delta(e)$ is disjoint, so that $\gamma;\delta_{i}(e_{i})$ is disjoint for all $1 \leq i \leq n$.

As $\gamma \vDash \Gamma$, $\delta_{i }\vDash ^{P} \Delta_{i}$, $\cup I(e_{i}) \subseteq P$, and $I(e_{i})$ is disjoint, we can apply induction to obtain $\gamma;\delta_{i}(e_{i}) \in R_{\tau_{i}}$. Therefore $I(\v_{i})$ is disjoint for all $1 \leq i \leq n$. By \Cref{prop:intcont}, $\cup I(\v_{i})\subseteq \cup I(\gamma;\delta_{i}(e_{i}))$. Thus if $i \neq j$ then $\cup I(\v_{i})$ is disjoint from $\cup I(\v_{j})$. This is sufficent to conclude that $I(\v)$ is disjoint, completing this case.

\item[\rname{T-Ops}] Then $e = o(e_1 ,\ldots ,e_{n})$ and $o : \ut_1 ,\ldots ,\ut_{n} \to \ut$. As $\gamma;\delta(e) \in R_{\ut}^{\vdash}$, we have $\v \in \sem{\ut}$. By \Cref{prop:utilistempty}, $I(\v) = \emptyset$. Therefore $I(\v)$ is disjoint.
    \item[\rname{T-Piece}] The argument for this case is nearly identical to that of $\rname{T-Tup}$.

    \item[\rname{T-Split}] Then $e = \kwlet\ x_1 ,\ldots ,x_{n} = e_1\ \kwin\ e_2$. Set
        \begin{align*}
            \tau_1  &= \ut_1 \times \cdots \times \ut_{n} \times \lint_1 \times \cdots \times \lint_{n'}   \\
            \Gamma' &= \Gamma, x_1: \ut_1,\ldots ,x_{n} : \ut_{n},\rd{w_1} : \rd{\lint_1} ,\ldots ,\rd{w_{n}}:\rd{\lint_{n'}} \\
            \Delta_1' &= w_1 : \lint_1 ,\ldots ,w_{n}:\lint_{n'} \\
            \Delta' &= \Delta_2,\Delta_1'
        \end{align*}
        This means $\Gamma;\Delta_1 \vdash e_1 : \tau_1$ and $\Gamma';\Delta'\vdash e_2 : \tau$. Since $\delta\vDash^{P} \Delta$, according to \Cref{prop:partitiondelta} we can partition $\delta$ up into $\delta_1$ and $\delta_2$ such that $\delta_1 \vDash^{P} \Delta_1$ and $\delta_2 \vDash^{P} \Delta_2$. Moreover, we have that $\gamma;\delta_1(e_1) = \gamma;\delta(e_1)$ and $\gamma;\delta_2(e_2) = \gamma;\delta(e_2)$.
        We have $I(e_1) \cup I(e)$ so $\cup I(e_1) \subseteq \cup I(e) \subseteq P$, and $I(e_1)$ is disjoint as $I(e)$ is disjoint.

        As $\gamma;\delta(e) \eval \v$, it must be that $\gamma;\delta_1(e_1) \eval (\v_1 ,\ldots , \v_{n}, \ldots ,\v_{n + n'})$ for some values $\v_{i}$ for $1 \leq i \leq n + n'$, and that
        \[
        \gamma;\delta(e_2)\{x_{i} \mapsto \v_{i} \mid 1 \leq i \leq n\}\{\rd{w_{i}} \mapsto \rd{\v_{i}}, w_{i} \mapsto \v_{i} \mid 1 \leq i \leq n' \} \eval \v.
    \]
        Since $\gamma \vDash \Gamma$, $\delta_{1} \vDash^{P} \Delta_{1}$, $\cup I(e_1) \subseteq P$, and $I(e_1)$ is disjoint, we can apply induction to obtain that
$I((\v_1,\ldots ,\v_{n+ n'}))$ is disjoint.

        Now set
        \[
        \gamma' = \gamma\{x_{i} \mapsto \v_{i}\mid 1 \leq i \leq n\}\{\rd{w_{i}} \mapsto \rd{\v_{n  + i}}\mid 1 \leq i \leq n'\},
        \]
        \[
            \delta_1' = \{w_{i} \mapsto \v_{n + i} \mid 1 \leq i \leq n'\},
        \]
        and
        \[
            \delta' = \delta_{2}\delta_1'.
        \]
    Clearly
\[
\gamma';\delta'(e_2) = \gamma;\delta(e_2)\{x_{i} \mapsto \v_{i} \mid 1 \leq i \leq n\}\{\rd{w_{i}} \mapsto \rd{\v_{i}}, w_{i} \mapsto \v_{i} \mid 1 \leq i \leq n' \}
\]
so we also have $\gamma';\delta'(e_2) \eval \v$.
             By \Cref{prop:intcont},
             \[
                 \cup I( (\v_1 ,\ldots ,\v_{n + n'})) \subseteq \cup I(e_1).
             \]
             This means
             \[
                 \cup I (\delta_1') \subseteq \cup I(e_1) \subseteq P.
             \]
            As $\delta_2 \vDash^{P} \Delta_2$, $I(\delta_2)$ is disjoint from $P$ and this means $I(\delta_2)$ is disjoint from $I(\delta_1')$. So setting $P' = P\setminus \cup I(e_1)$, we have $\delta_1' \vDash^{P'}  \Delta_1'$. This means $\delta' \vDash^{P'} \Delta'$ as $\delta' = \delta_2 \delta_1'$, and $\Delta' = \Delta_2,\Delta_1'$.

As $I(e)$ is disjoint, $I(e_2)$ is disjoint and $\cup I(e_1)$ and $\cup I (e_2)$ are disjoint from each other. Because $\cup I(e) \subseteq P$ and $\cup I(e_2)$ is disjoint from $\cup I(e_1)$, it must be that $\cup I (e_2)\subseteq P'$.

Together we have $\gamma' \vDash \Gamma'$, $\delta'\vDash^{P'} \Delta'$, $\cup I(e) \subseteq P'$, and $I(e_2)$ is disjoint, so by induction, $\gamma';\delta'(e_2) \in R_{\tau}$. In particular, this means that $I(\v)$ is disjoint. This establishes that $\gamma;\delta(e)\in R_{\tau}$.

    \item[\rname{T-Var}] This follows easily as $\gamma \vDash \Gamma$.
    \item[\rname{T-AffVar}] This also follows easily as $\delta \vDash^{P} w : \lint$.
    \item[\rname{T-Div}] It must be the case that $I(\v) = [r_1, r_2], [r_2,r'_1]$ where $r_1 \leq r_2 \leq r'_{1}$. This is already enough to conclude that $I(\v)$ is disjoint.

\end{description}
The remaining cases are straightforward. \qed
\end{proof}
As a corollary, we receive our desired result.
\subsection{Paths}
\label{app:paths}
Recall the \hyperlink{paths}{section} on paths.
Here are the path rules for the assert expression.
\begin{mathpar}\small
    \inferrule*[right=\rname{T-Assert}]
    {\Gamma;\Delta_1 \vdash b_{1} : \mathsf{Bool} \\ \Gamma; \Delta_2 \vdash b_2 : \tau}
    {\Gamma;\Delta_1,\Delta_2 \vdash \kwassert\ b_1\ \kwin\ b_2 :\tau}
    \and
    \inferrule*[right=\rname{E-Assert}]
    {b_1\eval \mathsf{true}\\ b_2\eval \v}
    {\kwassert\ b_1\ \kwin\ b_2 \eval \v}
\end{mathpar}
We define the set of paths $B(e)$ by induction on the structure of $e$:
\[
    B(\v) \triangleq \{\v\}\quad B(x) \triangleq \{x\} \quad B(w) \triangleq \{w\} \quad B(\rd{w}) \triangleq \rd{w}
\]
\begin{align*}
    B(\kwif\ e_1\ \kwthen\ e_2\ \kwelse\ e_3) &\triangleq \{ \pt{b_1}{b_2} \mid b_1 \in B(e_1), b_2 \in B(e_2) \}\\
    &\cup \{ \pf{b_1}{b_3} \mid b_1 \in B(e_1), b_3 \in B(e_3) \} \\
    B(e(e_1 ,\ldots ,e_n)) &\triangleq \{ e(b_1 ,\ldots ,b_n) \mid b_1 \in B(e_1) ,\ldots , b_n \in B(e_n)\}\ \text{otherwise.}
\end{align*}
for all other expressions $e(e_1 ,\ldots ,e_{n})$.
Here we prove the theorems from \Cref{sec:paths}

\begin{restatable}{proposition}{pathtype}
    \label{prop:pathtype}
     $\Gamma; \Delta \vdash e : \tau$ if and only if $\Gamma; \Delta \vdash b : \tau$ for all $b \in B(e)$.
 \end{restatable}

\begin{proof}
    \begin{flushleft}
        $(\Rightarrow)$
    \end{flushleft}
    We proceed by induction on the typing derivation. We break up our cases based on the last rule applied.
    \begin{description}
        \item[\rname{T-Cond}] Then $e = \kwif\ e_1\ \kwthen\ e_2\ \kwelse\ e_3$, $\Delta = \Delta_1, \Delta'$, $\Gamma;\Delta' \vdash e_1 :  \mathsf{Bool}$, $\Gamma ; \Delta' \vdash e_2:  \tau$, and $\Gamma; \Delta' \vdash e_3 : \tau$. By induction, for all $b_1 \in B(e_1)$, $\Gamma; \Delta_1 \vdash b_1 : \mathsf{Bool}$ and for all $b \in B(e_2)\cup B(e_3)$, $\Gamma; \Delta' \vdash b : \tau$. Then for any $b_1\in B(e_1)$ and $b_2\in B(e_2)$, \rname{T-Assert} nets $\Gamma;\Delta_1,\Delta' \vdash \kwassert\ b_1\ \kwin\ b_2 :\tau$. Since $\Gamma; \Delta_1 \vdash \neg b_1 : \mathsf{Bool}$ when $\Gamma;\Delta_1 \vdash b_1:  \mathsf{Bool}$, \rname{T-Assert} also nets $\Gamma;\Delta_1,\Delta' \vdash\kwassert\ \neg b_1\ \kwin\ b_3 : \tau$. This means that for all $b \in B(e)$, $\Gamma;\Delta_1, \Delta'\vdash  b : \tau$, concluding this case.
    \end{description}
    The remaining cases follow quickly from applying the inductive hypothesis.

    \begin{flushleft}
        $(\Leftarrow)$
    \end{flushleft}
    Suppose $\Gamma; \Delta \vdash b : \tau$ for all $b \in B(e)$.  We proceed by induction on the structure of $e$.
    \begin{description}
        \item[\rname{$e = \kwif\ e_1\ \kwthen\ e_2\ \kwelse\ e_3$}]
            Then $b = \kwassert\ b_1\ \kwin\ b_2$ where $b_1 \in B(e_1)$, $b_2 \in B(e_2)$ or $b = \kwassert\ \neg b_1\ \kwin\ b_3$ where $b_1 \in B(e_1)$, $b_3 \in B(e_3)$. It also must be that $\Delta = \Delta_1,\Delta'$, and $\Gamma;\Delta_1 \vdash b_1 : \mathsf{Bool}$, $\Gamma;\Delta' \vdash b_2 : \tau$, and $\Gamma;\Delta' \vdash b_3 : \tau$ for any $b_1\in B(e_1)$,$b_2\in B(e_2)$, and $b_3\in B(e_3)$.
             By induction, $\Gamma;\Delta_1 \vdash e_1 : \mathsf{Bool}$, $\Gamma;\Delta' \vdash e_2 : \tau$, and $\Gamma;\Delta' \vdash e_3 : \tau$. This means we can apply \rname{T-Cond} to obtain $\Gamma;\Delta_1,\Delta' \vdash e: \tau$.
    \end{description}
    All other cases are routine applications of induction.
    \qed
\end{proof}
    \begin{restatable}[]{proposition}{patheval}
    \label{prop:path}
    $e\eval \v$ if and only if $b\eval \v$ for some $b \in B(e)$.
 \end{restatable}

\begin{proof}
    \begin{flushleft}
    ($\Rightarrow$)
    \end{flushleft}
    We proceed by induction on the derivation of $e \eval \v$. Cases are broken up based on the last rule applied.
    \begin{description}
        \item[\rname{E-True}] Then $e = \kwif\ e_1\ \kwthen\ e_2\ \kwelse\ e_3$, $e_1\eval \kwtrue$ and $e_2\eval \v$. By induction, there exists $b_1 \in B(e_1)$, $b_2 \in B(e_2)$ such that $b_1\eval \kwtrue$ and $b_2\eval \v$. Then \rname{E-Assert} nets $\pt{b_1}{b_2}\eval \v$. Since $\pt{b_1}{b_2}\in B(e)$, which means we are done with this case.
        \item[\rname{E-False}] This case is similar to that of \rname{E-True}.
    \end{description}
    The remaining cases follow quickly from applying the inductive hypothesis.
    \begin{flushleft}
    ($\Leftarrow$)
    \end{flushleft}
    We proceed by induction on the derivation of $b \eval \v$. We break up cases by the last rule applied.
    \begin{description}
        \item[\rname{E-Assert}]
            Then $b = \pt{b_1}{b_2}$, $b_1\eval \kwtrue$, $b_2\eval \v$ as well as $e = \kwif\ e_1\ \kwthen\ e_2\ \kwelse\ e_3$. Since $b\in B(e)$, $b_1 \in B(e_1)$ and  and $b_2 \in B(e_2)$ by definition of $B(e)$. By induction, $e_1\eval \kwtrue$ and $e_2\eval \v$. Applying \rname{E-True} nets $e\eval \v$, which is the desired result.
        \item[\rname{E-False}]
            This is similar to \rname{E-True}.
    \end{description}
    The remaining cases follow quickly from applying the inductive hypothesis.
    \qed
\end{proof}

\subsection{Logical syntax}
\label{app:logicalsyntax}
We use a multi-sorted first order logic, given by the signature $(\sort,
\mathtt{F}, \mathtt{C}, \mathtt{Q})$, with:
\begin{itemize}
    \item Sorts $\sort$:
$\mathtt{Real},\mathtt{Bool},\mathtt{Vltn},\mathtt{Point}, \mathtt{Intvl}, \mathtt{s}_{1}\times \cdots \times \mathtt{s}_{n}$
\item Predicate symbols $\mathtt{Q}$: $\{\neg, \geq, =\}$
\item Constants $\mathtt{C}$: $\{r\pnt \mid r \in \mathbb{R}\} \cup \{\mathtt{cake}\}$
\item Function symbols $\mathtt{F}$:
$\mathtt{O}\cup \{(\_ ,\ldots ,\_)_{k}, \pi_{k} \mid k \in \mathbb{N}\}\cup \{\tleft, \tright, [\_, \_],\cup \}\cup\{\val_{a}\mid a \in \mathcal{A}\}\cup \{+, \cdot, *\}$
\end{itemize}
For each type $\tau$, there is an associated sort $\mathtt{s}_{\tau}$ given in \Cref{fig:sorts}.
Above, the set $\mathtt{O} = \{\mathtt{o} : \mathtt{s}_{\ut_1},\ldots ,\mathtt{s}_{\ut_{n}} \to \mathtt{s}_{\ut}\mid o: \ut_1,\ldots ,\ut_{n}\to \ut \in \mathcal{O} \}$, consists of logical function symbols corresponding to program operations.
 The function symbols $\tleft$ and $\tright$ give the left and right endpoints of intervals respectively, and $[ \_, \_]$ is an interval constructor. The unary operator $*$ takes a point to its corresponding real, and the binary symbols $+$ and $\cdot$ give arithmetic on real numbers.
Using the function symbols and constants, we can encode any program value $\v$ as a
logical term $\lv$.

\begin{figure}
\begin{mathpar}
    \mathtt{s}_{\mathsf{Real}} = \mathtt{Real}
    \and
    \mathtt{s}_{\mathsf{Bool}} = \mathtt{Bool}
    \and
    \mathtt{s}_{\mathsf{Point}} = \mathtt{Point}
    \and
    \mathtt{s}_{\mathsf{Intvl}} = \mathtt{Intvl}
    \and
    \mathtt{s}_{\mathsf{Piece}} = \mathtt{Piece}
    \and
    \mathtt{s}_{\mathsf{Vltn}} = \mathtt{Vltn}
    \and
    \mathtt{s}_{\rd{\mathsf{Intvl}}} = \mathtt{Intvl}
    \and
    \mathtt{s}_{\rd{\mathsf{Piece}}} = \mathtt{Piece}
    \and
    s_{\tau_1 \times \cdots \times \tau_{n}} = s_{\tau_1} \times \cdots \times s_{\tau_{n}}
\end{mathpar}
\caption{Sort associated with each type}
\label{fig:sorts}
\end{figure}

All logical variables are drawn from $\var$, $\lvar$, $\rd{\lvar}$, and
$\mathcal{Y}$. The first three sets are the same sets used for program
variables, and $\mathcal{Y}$ is a new countably infinite set of logical variables.

We use a standard sorting judgment for well-formed terms ($\pctx;\yctx \vdash t
: \mathtt{s}$) and well-formed
formulas ($\pctx;\yctx \vdash \psi : \text{Formula}$). In the sorting judgments,
$\pctx$ is a partial map from $\var\cup \lvar\cup \rd{\lvar}$ to $\sort$, and $\yctx$ is a finite subset of $\mathcal{Y}$. $\pctx$ handles variables that arise in intermediate stages of translating expressions to their constraints, and $\yctx$ handles variables created in the constraint generation procedure. Every variable in $\mathcal{Y}$ is assumed to be a point, so $\yctx$ is a set and not a partial map. The sorting rules for terms and predicate formulas are given in \Cref{fig:sortingrules}.
\begin{figure}
    \label{fig:sortingrules}
\begin{mathpar}
    \inferrule*{ }{\pctx, x : \mathtt{s};\yctx \vdash x : \mathtt{s}}
    \and
    \inferrule*{ }{\pctx;\yctx, y  \vdash y : \mathtt{Point}}
    \and
    \inferrule*{ }{\pctx;\yctx \vdash r : \mathtt{Real}}
    \and
    \inferrule*{ }{\pctx;\yctx \vdash r\pnt : \mathtt{Point}}
    \and
    \inferrule*{ }{\pctx;\yctx \vdash \mathtt{cake} : \mathtt{Piece}}
    \and
    \inferrule*{\pctx; \yctx \vdash t_1 : \mathtt{s}_{1}\  \cdots\  \pctx; \yctx \vdash t_n : \mathtt{s}_{n} \\ \mathtt{o} : \mathtt{s}_{1},\ldots ,\mathtt{s}_{n} \to \mathtt{s}}
    {\pctx; \yctx \vdash \mathtt{o}(t_1 ,\ldots ,t_{n}) : \mathtt{s}}
    \and
    \inferrule*{\pctx;\yctx \vdash t : \mathtt{Intvl}}{\pctx;\yctx \vdash \tleft(t) : \mathtt{Point}}
    \and
    \inferrule*{\pctx;\yctx \vdash t : \mathtt{Intvl}}{\pctx;\yctx \vdash \tright(t) : \mathtt{Point}}
    \\
    %
    % Val
    %
    \inferrule*{\pctx;\yctx \vdash t : \mathtt{Intvl}}{\pctx;\yctx \vdash \val_{a}(t) : \mathtt{Vltn}}
    \and
    \inferrule*{\pctx;\yctx \vdash t : \mathtt{Piece}}{\pctx;\yctx \vdash \val_{a}(t) : \mathtt{Vltn}}
    %
    % Real
    %
    \and
    \inferrule*{\pctx;\yctx \vdash t_1 : \mathtt{Point}\\ \pctx;\yctx \vdash t_2 : \mathtt{Point}}{\pctx;\yctx \vdash [t_1,t_2] : \mathtt{Intvl}}
    \and
    \inferrule*{\pctx;\yctx \vdash t_1 : \mathtt{Intvl}\\ \cdots \\ \pctx;\yctx \vdash t_{n} : \mathtt{Intvl}}{\pctx;\yctx \vdash \cup t_1,\ldots ,t_{n} : \mathtt{Piece}}
    \\
    \inferrule*{\pctx;\yctx \vdash t_1 : \mathtt{s}_1 \\ \cdots \\ \pctx;\yctx \vdash t_{k} : \mathtt{s}_{k}}{\pctx;\yctx \vdash (t_1,\ldots ,t_{k}) : \mathtt{s}_1 \times \cdots \times \mathtt{s}_{k}}
    \and
    \inferrule*[right={($1 \leq n \leq k$)}]{\pctx;\yctx \vdash t : \mathtt{s}_1 \times \cdots \times \mathtt{s}_{k}}{\pctx;\yctx \vdash \pi_{n}t : \mathtt{s}_{n}}
\end{mathpar}
\\
\begin{mathpar}
    \inferrule*{\pctx;\yctx \vdash t_{1}  : \mathtt{Bool}}{\pctx;\yctx \vdash t_1 = \mathtt{true} : \text{Formula}}
    \and
    \inferrule*{\pctx;\yctx \vdash t_{1}  : \mathtt{Real}\\ \pctx;\yctx \vdash t_{2}  : \mathtt{Real}}{\pctx;\yctx \vdash t_1 \geqq t_2 : \text{Formula}}
    \and
    \inferrule*{\pctx;\yctx \vdash t_{1}  : \mathtt{Point}\\ \pctx;\yctx \vdash t_{2}  : \mathtt{Point}}{\pctx;\yctx \vdash t_1 \geqq t_2 : \text{Formula}}
    \and
    \inferrule*{\pctx;\yctx \vdash t_{1}  : \mathtt{Vltn}\\ \pctx;\yctx \vdash t_{2}  : \mathtt{Vltn}}{\pctx;\yctx \vdash t_1 \geqq t_2 : \text{Formula}}
\end{mathpar}
\caption{Rules for sorting terms and well-formed formulas.}
\end{figure}

\Cref{prop:simpl} is proven under the assumption of the following operations.
\begin{figure}
    \label{fig:ops}
\begin{align*}
    \begin{array}{l l l l}
        \land :\mathsf{Bool}, \mathsf{Bool} \to \mathsf{Bool}  & \lor : \mathsf{Bool}, \mathsf{Bool} \to \mathsf{Bool} & \neg : \mathsf{Bool} \to \mathsf{Bool} & +:  \mathsf{Vltn}, \mathsf{Vltn} \to \mathsf{Vltn}\\ \\
        \geq : \mathsf{Real}, \mathsf{Real}\to \mathsf{Bool} & \geq : \mathsf{Vltn}, \mathsf{Vltn} \to \mathsf{Bool} &  = : \mathsf{Vltn}, \mathsf{Vltn} \to \mathsf{Bool} & \cdot_{r} :\mathsf{Vltn} \to \mathsf{Vltn}
\end{array}
\end{align*}
\caption{The operations contained in $\mathcal{O}$.}
\end{figure}

\subsubsection{Basic simplification}
Here we discuss basic formula simplification.
Given a term $t$, we define the simplification $\simpl(t)$ inductively as follows:
    \begin{align*}
        \simpl(\pi_{k}(t)) &\triangleq
        \begin{cases}
            t_{k} &\text{if}\ \simpl(t) = (t_1,\ldots ,t_{k},\ldots ,t_{n}) \\
            \pi_{k}(\simpl(t)) &\text{otherwise}
        \end{cases}\\
        \simpl(\tleft (t)) &\triangleq
        \begin{cases}
            t_{1} &\text{if}\ \simpl(t) = [t_1,t_{2}] \\
            \tleft(\simpl(t)) &\text{otherwise}.
        \end{cases}\\
        \simpl(\tright (t)) &\triangleq
        \begin{cases}
            t_{2} &\text{if}\ \simpl(t) = [t_1,t_{2}] \\
            \tright(\simpl(t)) &\text{otherwise}.
        \end{cases}\\
        \\
        \simpl(f(t_1 ,\ldots ,t_{n})) &\triangleq f(\simpl(t_1),\ldots ,\simpl(t_{n}))\ \text{for $t = f(t_1 ,\ldots ,t_{n})$ otherwise}.
    \end{align*}
        The goal of this simplification is to reduce formulas to real linear inequalities and equalities in $\val_{a}$ evaluated on pieces or intervals.
    To argue that this is the case, we define the following logical predicate on terms indexed by sort:
        \begin{align*}
        R_{\mathtt{Bool}} &= \{t \mid \cdot; \yctx\vdash t : \mathtt{Bool}, t = t_1 \wedge t_2\ \text{and}\ t_1, t_2\in \mathtt{Bool}\ \text{or}\ t = t_1 \vee t_2\ \text{and}\ t_1, t_2\in \mathtt{Bool}\ \text{or}\ \\
                          &t = t_1 \geqq t_2\ \text{and}\ t_1,t_2\in R_{\tau}, \tau = \mathtt{Real}, \mathtt{Vltn} \} \\
        R_{\mathtt{Real}} &= \{t \mid \cdot; \yctx \vdash t : \mathtt{Real}, t = t_1 + t_2\ \text{and}\ t_1,t_2\in R_{\mathtt{Real}}\ \text{or}\\
        & t = r \cdot t'\ \text{and}\ t'\in R_{\mathtt{Real}}\ \text{or}\ t = t'^{*}, t'\in R_{\mathtt{Point}}\ \text{or}\ t = r\}\\
 R_{\mathtt{Point}} &= \{t \mid \cdot ; \yctx \vdash t : \mathtt{Point}, t = r\pnt\ \text{or}\ t \in \mathcal{Y}\} \\
       R_{\mathtt{Vltn}} &= \{t \mid \cdot; \yctx \vdash t : \mathtt{Vltn}, t = \sum_{i}t_{i}\ \text{and}\ t_{i}\in R_{\mathtt{Vltn}},\ \text{or}\ t = r\cdot t'\  \text{and}\ t' \in R_{\mathtt{Vltn}}, \\
                        &\text{or}\ t = \val_{a}(t')\ \text{and}\ t' \in R_{\mathtt{Intvl}}\cup R_{\mathtt{Piece}}\}\\
       R_{\mathtt{Intvl}} &= \{t \mid \cdot; \yctx \vdash t : \mathtt{Intvl}, t = [t_1,t_2], t_1, t_2\in R_{\mathtt{Real}}\} \\
       R_{\mathtt{Piece}} &= \{t \mid  \cdot ; \yctx \vdash t : \mathtt{Piece}, t = \cup t_1\cdots  t_{n}, t_{i}\in R_{\mathtt{Intvl}}\} \\
        R_{\mathtt{s}_1 \times \cdots \times \mathtt{s}_{k}} &= \{t \mid \cdot; \yctx \vdash t : \mathtt{s}_{1} \times \cdots \times \mathtt{s}_{k}, t = (t_1,\ldots ,t_{k}), t_{i}\in R_{\mathtt{s}_{i}} \}
        \end{align*}

        \begin{proposition}
        If $ \cdot ;\yctx \vdash t: \mathtt{s} $ then $\simpl(t) \in R_{\mathtt{s}}$.
        \label{prop:relationsimpl}
        \end{proposition}
        \begin{proof}
        This proceeds by induction on the sorting derivation. We break up cases based on the last rule used to sort $t$. We expand out the ops rule now that we've fixed the possible operations. Here we write $\geqq$ as short-hand for both $\geq$ and $=$ symbols.
\begin{mathpar}
    \inferrule*{ }{\cdot; \yctx \vdash r : \mathtt{Real}}
    \and
    \inferrule*{ }{\cdot; \yctx \vdash b : \mathtt{Bool}}
    \and
    \inferrule*{ }{\cdot; \yctx \vdash r \pnt : \mathtt{Point}}
    \and
    \inferrule*{ }{\cdot ; \yctx, y : \mathtt{Point} \vdash y : \mathtt{Point}}
\end{mathpar}
These are obvious.
\begin{mathpar}
    \inferrule*{\cdot ; \yctx \vdash t_1 : \mathtt{Real}\\ \cdot ; \yctx \vdash t_2 : \mathtt{Real} }{t_1 \geqq t_2: \mathtt{Bool}}
    \and
    \inferrule*{\cdot ; \yctx \vdash t_1 : \mathtt{Vltn}\\ \cdot ; \yctx \vdash t_2 : \mathtt{Vltn} }{\cdot ; \yctx \vdash t_1 \geqq t_2 : \mathtt{Bool}}
    \and
    \inferrule*{\cdot ; \yctx \vdash t_1 : \mathtt{Bool}\\ \cdot ; \yctx \vdash t_2 : \mathtt{Bool} }{\cdot ; \yctx \vdash t_1 \wedge t_2: \mathtt{Bool}}
    \and
    \inferrule*{\cdot ; \yctx \vdash t_1 : \mathtt{Bool}\\ \cdot ; \yctx \vdash t_2 : \mathtt{Bool} }{\cdot ; \yctx \vdash t_1 \vee t_2: \mathtt{Bool}}
    \and
    \inferrule*{\cdot ; \yctx \vdash t : \mathtt{Bool}}{\cdot ; \yctx \vdash \neg t: \mathtt{Bool}}
\end{mathpar}
These all readily follow by induction.
\begin{mathpar}
    \inferrule*{\cdot ; \yctx \vdash t_1 : \mathtt{Point}\\ \cdot ; \yctx \vdash t_2 : \mathtt{Point}}{\cdot ; \yctx \vdash [t_1,t_2] : \mathtt{Intvl}}
    \and
    \inferrule*{\cdot ; \yctx \vdash t_1 : \mathtt{Intvl}\\ \cdots \\ \cdot ; \yctx \vdash t_{n} : \mathtt{Intvl}}{\cdot ; \yctx \vdash \cup t_1,\ldots ,t_{n} : \mathtt{Piece}}
    \and
    \inferrule*{\cdot ; \yctx \vdash t_1 : \mathtt{Real}\\ \cdot ; \yctx \vdash t_2 : \mathtt{Real}}{\cdot ; \yctx \vdash t_1 + t_2 : \mathtt{Real}}
    \and
    \inferrule*{\cdot ; \yctx \vdash t_1 : \mathtt{Vltn}\\ \cdot ; \yctx \vdash t_2 : \mathtt{Vltn}}{\cdot ; \yctx \vdash t_1 + t_2 : \mathtt{Vltn}}
    \and
    \inferrule*{\cdot ; \yctx \vdash t : \mathtt{Real}}{\cdot ; \yctx \vdash r \cdot t : \mathtt{Real}}
    \and
    \inferrule*{\cdot ; \yctx \vdash t : \mathtt{Vltn}}{\cdot ; \yctx \vdash r\cdot t : \mathtt{Vltn}}
    \and
    \inferrule*{\cdot ; \yctx \vdash t : \mathtt{Point}}{\cdot ; \yctx \vdash t^{*} : \mathtt{Real}}
\end{mathpar}
These arguments are straightforward applications of the inductive hypothesis.
\begin{mathpar}
    \inferrule*{\cdot ; \yctx \vdash t' : \mathtt{Intvl}}{\cdot ; \yctx \vdash \tleft(t') : \mathtt{Point}}
    \and
    \inferrule*{\cdot ; \yctx \vdash t' : \mathtt{Intvl}}{\cdot ; \yctx \vdash \tright(t') : \mathtt{Point}}
    \end{mathpar}
    By induction $\mathtt{s}(t) \in R_{\mathtt{Intvl}}$ so $\mathtt{s}(t) = [t_1, t_2]$ where $t_1, t_2 \in R_{\mathtt{Point}}$. This means $\mathtt{s}(\ell(t)) = t_1$. Similarly, $\mathtt{s}(r(t)) = t_2$. As $t_1, t_2 \in R_{\mathtt{Point}}$, we are done with these cases.
    \begin{mathpar}
    \inferrule*{\cdot ; \yctx \vdash t' : \mathtt{Intvl}}{\cdot ; \yctx \vdash \val_{a}(t') : \mathtt{Vltn}}
    \end{mathpar}
    Then $\simpl(t) = \val_{a}(\simpl(t'))$. By induction, $\simpl(t') \in R_{\mathtt{Intvl}}$, completing this case.
    \begin{mathpar}
    \inferrule*{\cdot ; \yctx \vdash t' : \mathtt{Piece}}{\cdot ; \yctx \vdash \val_{a}(t') : \mathtt{Vltn}}
    \end{mathpar}
    Almost identical to the previous case.

\begin{mathpar}
    \inferrule*{\cdot ; \yctx \vdash t_1 : \mathtt{s}_1 \\ \cdots \\ \cdot ; \yctx \vdash t_{n} : \mathtt{s}_{n}}{\cdot ; \yctx \vdash (t_1,\ldots ,t_{n}) : \mathtt{s}_1 \times \cdots \times \mathtt{s}_{n}}
\end{mathpar}
    By induction, $\simpl(t_{i})\in R_{\mathtt{s}_{i}}$. Since $\simpl(t_1,\ldots ,t_{n}) = (\simpl(t_1),\ldots ,\simpl(t_{n}))$, we are done with this case.
\begin{mathpar}
    \inferrule*[right={($1 \leq k \leq n$)}]{\cdot ; \yctx \vdash t' : \mathtt{s}_1 \times \cdots \times \mathtt{s}_{n}}{\cdot ; \yctx \vdash \pi_{k}t' : \mathtt{s}_{k}}
\end{mathpar}
By induction, $\simpl(t')\in R_{\mathtt{s}_{1}\times \cdots \times \mathtt{s}_{k}}$ which means $\simpl(t') = (t_1 ,\ldots ,t_{k})$ where $t_{i} \in R_{\mathtt{s}_{i}}$. Therefore $\simpl(t) = t_{k} \in R_{\mathtt{s}_{k}}$. \qed
\end{proof}

If $t \in R_{\mathsf{Bool}}$, there is a clear formula it can be converted to, which we denote $F(t)$. From this, we can extend $\simpl$ to formulas.
Given a well-sorted formula $f$, we inductively define $\simpl(f)$ as follows
    \begin{align*}
        \simpl(t = \mathsf{true}) &\triangleq F(\simpl(t)) \\
        \simpl(t_1 \geq t_2) &\triangleq \simpl(t_1) \geq \simpl(t_2)\\
        \simpl(t_1 = t_2) &\triangleq \simpl(t_1) = \simpl(t_2)\ \text{if}\ t_2 \neq \mathsf{true}\\
        \simpl(\neg f) &\triangleq \neg (\simpl (f)) \\
        \simpl(f_1 \land f_2) & \triangleq \simpl(f_1)\land \simpl(f_2) \\
        \simpl(f_1 \lor f_2) & \triangleq \simpl(f_1)\lor \simpl(f_2) \\
        \simpl(f_1 \Rightarrow f_2) & \triangleq \simpl(f_1)\Rightarrow \simpl(f_2)
    \end{align*}
    We say that a formula is \emph{simplified} if for each term contained in it, it is contained within $R_{\mathtt{s}}$ for some sort $\mathtt{s}$. By \Cref{prop:relationsimpl}, $\simpl(f)$ is simplified for any well-sorted formula $f$.
    Based on our assumptions on predicate symbols, we have the following corollary.
    \begin{corollary}
        Suppose that $\cdot ; \yctx \vdash f : \text{Formula}$. Then $\simpl(f)$ consists entirely of conjunctions, disjunctions, and negations of the forms:
        \[
            \sum_{i}r_{i}\val_{a}(P_{i}) \geqq \sum_{j}r_{j}\val_{a'}(P_{j})
    \]
    where $P_{i},P_{j}\in R_{\mathtt{Intvl}}\cup R_{\mathtt{Point}}$
    and

    \[
            \sum_{i}r_{i}t_{i}\geqq \sum_{j}r_{j}t_{j}
    \]
    where $t_{i},t_{j} \in \{t \mid t = t'^{*}, t' \in R_{\mathtt{Point}}\}\cup \mathbb{R}$ and $\geqq$ stands for $\geq$ and $=$.
    \label{prop:lra}
    \end{corollary}
    In particular, $\cn(b)$ for any branch $\cdot \vdash b : \tau$ satisfies the corollary's conditions.
\subsection{Logical semantics}
\label{app:logicalsem}
We give semantics to formulas through means of an interpretation, $\interp$, and
variable assignment $\assn$. An interpretation associates with each sort $\mathtt{s}$, a set $\sem{\mathtt{s}}_{\interp}$, and with each function symbol $f$ a function $\sem{f}_{\interp}$. All sorts are interpreted as shown in \Cref{fig:sortinterp} regardless of $\interp$ and all symbols are interpreted as shown in \Cref{fig:symbolinterp} \emph{besides} $\val_{a}$ for each $a$--these are interpreted as valuations determined by the specific interpretation. As such, interpretations are completely determined by the valuations they use to interpret. Thus for a valuation set $\vset$, we let $\interp_{\vset}$ be the interpretation such that $\sem{\val_{a}}_{\interp_{\vset}} = \vset_{a}$ for all $a \in \agents$.
Before moving on we describe the function symbol interpretations in words. The symbol $\caketerm$ is the whole cake, $\tleft$ and $\tright$ provide the left and right endpoints of an interval respectively, $\cup$ forms a piece from intervals. Recall that for each $o: \ut_1,\ldots ,\ut_{n} \to \ut \in O$, there is $\sem{o} : \unrd{\sem{\ut_{1}}} \times \cdots \times \unrd{\sem{\ut_{n}}} \to \unrd{\sem{\ut}}$.

Now a variable assignment is a map from variables to elements of the interpreted sorts.
Given a term $t$, we let $\sem{t}_{\interp}^{\assn}$ denote the interpretation
of $t$ according to $\interp$,  with variable values determined by $\assn$,
defined in the usual way. Likewise, given a formula $f$, we write $\interp,\assn
\vDash f$ if $f$ is true when interpreted through $\interp$ with variable values
determined by $\assn$, also defined in the usual way. We write $\interp \vDash
f$ if for all assignments $\assn$ we have $\interp, \assn \vDash f$. If $t$ is a
term containing no $\val_{a}$ symbols, then for a fixed assignment $\assn$, the
term $t$ is always interpreted the same way and we write just $\sem{t}^{\assn}$.

\begin{figure}
\begin{mathpar}
    \sem{\mathtt{Bool}}_{\interp} \triangleq \{\mathtt{true}, \mathtt{false}\}
    \and
    \sem{\mathtt{Num}}_{\interp} \triangleq \mathbb{N}
    \and
    \sem{\mathtt{Real}}_{\interp} \triangleq \mathbb{R}
    \and
    \sem{\mathtt{Vltn}}_{\interp} \triangleq \mathbb{V}
    \and
    \sem{\mathtt{Point}}_{\interp} \triangleq \mathtt{P}\mathtt{T}
    \and
    \sem{\mathtt{Intvl}}_{\interp} \triangleq \mathbb{I}
    \and
    \sem{\mathtt{Piece}}_{\interp} \triangleq \mathbb{P}\mathbb{C}
\end{mathpar}
\caption{Interpretations of each sort.}
\label{fig:sortinterp}
\end{figure}
\begin{figure}
\begin{mathpar}
    \sem{\mathtt{o}}_{\interp} = \sem{o}
    \and
    \sem{\pi_{i}}_{\interp}(\v_1 ,\ldots, \v_{i},\ldots ,\v_{n}) = \v_{i}
    \and
    \sem{(\_ ,\ldots ,\_)_{k}}_{\interp}(\v_1,\ldots ,\v_{k}) = (\v_1,\ldots ,\v_{k})
    \and
\sem{\tleft}_{\interp}([r, r']) = r
    \and
\sem{\tright}_{\interp}([r, r']) = r'
    \and
\sem{\caketerm}_{\interp} = [0,1]
    \and
\sem{\cup}_{\interp}([r_{1}, r'_{1}],\ldots ,[r_{n},r'_{n}]) = P_{i= 1}^{n}[r_{i},r'_{i}]
\end{mathpar}
\caption{Interpretations of each function symbol.}
\label{fig:symbolinterp}
\end{figure}

We close off this section with the soundess of the previous section's formula simplification with respect to the semantics given here.
\begin{proposition}
    \label{prop:simpl}
    Let $\interp$ be an interpretation and $\assn$ be a variable assignment. Let $\varphi$ be any formula. Then
    $\interp,\assn \vDash \varphi \iff \interp,\assn \vDash \simpl(\varphi)$.
\end{proposition}
\subsection{Constraints}
\label{app:constraints}
This section argues for soundness and completeness of the \hyperlink{constraints}{constraints}.

We require each mark to have a unique identifier, $\# s$, to prevent variable clashes when translating expressions to their constraints:
\begin{equation*}
    \emark{a}{e_1}{e_2}\# s.
\end{equation*}
    For each identifier in the program, $\# s$, we assume there is a unique member of $\mathcal{Y}$, denoted $y_{\# s}$, corresponding to it. We also let $\id(b)$ denote the set of indentifiers in $b$. The complete definition for $\rho(b)$ is given here:
      \begin{mathpar} \small
          \rho(\rd{\v}) = \lv
          \and
          \rho_{}(\v)_{\ilist{i}}  ={\lv}
          \and
          \rho_{}(x)_{\ilist{i}}  ={x}
          \and
          \rho_{}(w)_{\ilist{i}} ={w}
          \and
          \rho_{}(\rd{w})_{\ilist{i}}  = \rd{w}
          \and
          \rho(\kwcake) \triangleq [0,1]
    \and
\rho_{ }(\ediv{b_1}{b_{2}})_{\ilist{i}} =
    \left(\left[\tleft(\rho_{ }(b_1)_{\ilist{1:i}}), \rho_{ }(b_2)_{\ilist{2:i}}  \right]_{},
    \left[\rho_{ }(b_2)_{\ilist{2:i}}, \tright(\rho_{ }(b_1)_{\ilist{1:i}})  \right]_{}\right)
    \and
    {\rho_{ }}(\emark{a}{b_1}{b_2}\# s)_{\ilist{i}}  = y_{\# s}
    \and
    {\rho_{ }}(\eeval{a}{b})_{\ilist{i}} = \val_{a}(\rho_{ }(b)_{\ilist{1:i}})
    \and
      \rho_{ } (\pt{b_1}{b_2})_{\ilist{i}}  =  \rho(b_2)_{\ilist{2:i}}
      \and
      \rho_{ }(o(b_1 , \ldots , b_n))_{\ilist{i}}  = o(\rho_{ }(b_1)_{\ilist{1:i}}, \ldots , \rho_{ }(b_n)_{\ilist{n:i}})
      \and
      \rho_{ }(\kwlet\ x_1 ,\ldots ,x_{n}, w_{1},\ldots ,w_{n'} = \kwsplit\ b_1\ \kwin\ b_2)_{\ilist{i}}            \triangleq \\ \rho_{ }(b_2)_{\ilist{2:i}}\{\pi_{i}\rho_{ }(b_1)_{\ilist{1:i}}\mapsto x_{i}\mid 1 \leq i \leq n\}
                                                                                                                    \{\pi_{i}\rho_{ }(b_1)_{\ilist{1:i}}\mapsto w_{i}, \pi_{i}\rho_{ }(b_1)_{\ilist{1:i}}\mapsto \rd{w_{i}}\mid 1 \leq i \leq n'\}
                                                                                                                    \and
    \rho (\epiece{b_1 ,\ldots ,b_{n}})_{\ilist{i}}= \cup\rho(b_{1})_{\ilist{1:i}}, \ldots , \rho(b_{n})_{\ilist{n :i}}\\
    \and
    \rho((b_1 ,\ldots ,b_{n})) \triangleq (\rho(b_1),\ldots ,\rho(b_{n}))
      \end{mathpar}

Here is the complete definition for $\cn(b)$:
\begin{mathpar} \small
\cn(\ediv{b_1}{b_2})_{\ilist{i}}
    \triangleq \cn(b_1)_{\ilist{1:i}}\wedge \cn(b_2)_{\ilist{2:i}} \wedge \tleft(\rho(b_1)_{\ilist{1:i}})\leq \rho(b_2)_{\ilist{ 2 : i}}\leq \tright(\rho(b_1)_{\ilist{ 1 : i}})
    \and
    \cn(\emark{a}{b_1}{b_2} \# s)_{\ilist{i}} = \cn(b_1)_{\ilist{1:i}}\wedge \cn(b_2)_{\ilist{2 : i}}\wedge (\val_{a}([\tleft(\rho(b_1)), \rho(\emark{a}{b_1}{b_2}\# s)_{\ilist{i}}])= \rho(b_2)_{\ilist{2:i}})
    \and
    \cn(\eeval{a}{b})_{\ilist{i}} = \cn(b)
    \and
    \cn(\pt{b_1}{b_2}) = (\rho(b_1) = \mathsf{true})\wedge \cn(b_1)\wedge \cn(b_2)
    \and
    \cn(\kwlet\ x_1,\ldots ,x_{n}, w_1,\ldots ,w_{n'}= \kwsplit\ b_1\ \kwin\ b_2) \triangleq \cn(b_1)\land\\
    \cn(b_2)\{x_{i} \mapsto \pi_{i}\rho(b_1)_{\ilist{1:i}})\mid 1 \leq i \leq n\}
    \{w_{i}\mapsto \pi_{i + n}\rho(b_1), \rd{w_{i}}\mapsto \pi_{i + n}\rho(b_1)\mid 1 \leq i \leq n'\}
    \and
    \cn( (b_1,\ldots ,b_{n})) \triangleq \cn(b_1) \land \cdots \land \cn(b_{n})
    \and
    \cn(\kwpiece(b_1,\ldots ,b_{n})) \triangleq \cn(b_1) \land \cdots \land \cn(b_{n})
\end{mathpar}

    Given a substitution from program variables variables to values $S$, let $\unrd{S}$ be the logical substitution $\{x \mapsto \rho(\v) \mid  S(x) = \v\}$. Recall that $\rho$ converts values to their logical counterparts, and makes a read only value a regular value.

\begin{lemma}
    Let $b$ be a branch and $S$ a substitution of values for program variables. Then $\rho(S(b)) = \unrd{S}(\rho(b))$, and $\cn(S(b)) = \unrd{S}(\cn(b))$.
\label{prop:subswapf}
\end{lemma}
\begin{proof}
    Once the first statement is established the second follows trivially. The first is argued by induction on the structure of $b$. We exhibit the interesting cases, and the rest follow easily.
     \begin{description}
         \item[\rname{$b = x$}]
             If $x$ is not a variable in $S$ then this is trivial, so suppose it is.
             Then we have
             \begin{align*}
                 \rho(S(b)) = \rho(S(x)) = \rho(\v) = \unrd{S}(x) = \unrd{S}(\rho(x)) = \unrd{S}(\rho(b)).
             \end{align*}
             $\cn(x) = \cn(\v)$ for any variable and value so clearly $\unrd{S}(\cn(b)) = \cn(S(b))$.
         \item[\rname{$b = w$}]
            This is identical to the above case.
        \item[\rname{$b = \rd{w}$}]
            If $w$ is not a variable in $S$, then this is trivial, so suppose it is.
            Then we have
            \begin{align*}
                \rho(S(b)) = \rho(S(\rd{w})) = \rho(\rd{\v}) = \v = \unrd{S}(\rd{w}) = \unrd{S}(\rho(\rd{w})) =\unrd{S} (\rho(b)).
             \end{align*}
             $\cn(\rd{w}) = \cn(\rd{\v})$ for any variable and value so clearly $S(\cn(b)) = \cn(S(b))$.
         \item[\rname{$b = \kwlet\ x_1,\ldots ,x_{n}, w_1,\ldots ,w_{n'} = \kwsplit\ b_1\ \kwin\ b_2$}]
         Let
         \[
             S' = S\setminus \{x_1,\ldots ,x_{n},w_{1},\ldots ,w_{n'},\rd{w_1} ,\ldots ,\rd{w_{n}}\}.
     \]
         We have that

         \[
             S(b) = \kwlet\ x_1 ,\ldots ,x_{n}, w_1,\ldots ,w_{n'} = \kwsplit\ S(b_1)\ \kwin\ S'(b_2),
         \]
         and also have
         \[
             \rho(b) = \rho(b_2)\{x_{i } \mapsto \pi_{i}\rho(b_1)  \mid 1 \leq i \leq n\}\{w_{i} \mapsto \pi_{i + n}\rho(b_1),\rd{w_{i}} \mapsto \pi_{i + n}\rho(b_1)  \mid 1 \leq i \leq n'\}.
         \]
         By induction,
         \[
             \rho(S(b_1)) = \unrd{S}(\rho(b_1)) \quad \text{and} \quad \rho(S'(b_2)) = \unrd{S'}(\rho(b_2)).
         \]
         We can then calculate:
         \begin{align*}
             \rho(S(b)) &= \rho(S'(b_2))\{x_{i}\mapsto \pi_{i}\rho(S(b_1))\mid 1 \leq i \leq n\}\{w_{i}\mapsto \pi_{i + n}\rho(S(b_1)), \rd{w_{i}}\mapsto \pi_{i + n}\rho(S(b_1))\mid 1 \leq i \leq n'\} \\
                        &= \unrd{S'}(\rho(b_2)) \{x_{i}\mapsto \unrd{S}(\pi_{i}\rho(b_1))\mid 1 \leq i \leq n\}\\
                        &\quad \{w_{i}\mapsto \unrd{S}(\pi_{i + n}\rho(b_1)), \rd{w_{i}}\mapsto \unrd{S}(\pi_{i + n}\rho(b_1))\mid 1 \leq i \leq n'\}
         \end{align*}
         and
         \begin{align*}
             \unrd{S}(\rho(b)) &= \unrd{S}(\rho(b_2)\{x_{i}\mapsto \pi_{i}\rho(b_1)\mid 1 \leq i \leq n\}\{w_{i}\mapsto \pi_{i + n}\rho(b_1), \rd{w_{i}}\mapsto \pi_{i + n}\rho(b_1)\mid 1 \leq i \leq n'\})  \\
                               &= \unrd{S'}(\rho(b_2))\{x_{i}\mapsto \unrd{S}(\pi_{i}\rho(b_1))\mid 1 \leq i \leq n\}\\
             &\quad \{w_{i}\mapsto \unrd{S}(\pi_{i + n}\rho(b_1)), \rd{w_{i}}\mapsto \unrd{S}(\pi_{i + n}\rho(b_1))\mid 1 \leq i \leq n'\}.
         \end{align*}
         Therefore $\rho(S(b)) = \unrd{S}(\rho(b))$.
         The argument for $\cn(S(b)) = \unrd{S}(\cn(b))$ in this case is very similar.
    \qed
    \end{description}
\end{proof}

Given a term $t$, let $\fv_{\mathcal{Y}}(t)$ be the set of free variables contained in $t$ from $\mathcal{Y}$. Analogously for a formula $f$, let $\fv_{\mathcal{Y}}(f)$ be the set of free variables contained in $f$ from  $\mathcal{Y}$. Also let $\mathcal{Y}_b = \fv_{\mathcal{Y}}(\cn(b))\cup \fv_{\mathcal{Y}}(\rho(b))$. We write $\pctx \vdash \Gamma;\Delta$ if for all $x \in \dom(\Gamma)$,
$\pctx(x) = \mathtt{s}_{\Gamma(x)}$, and for all $w \in \dom(\Delta)$, $\pctx(w) = \mathtt{s}_{\Delta(w)}$.
\begin{proposition}
   Suppose $\Gamma;\Delta \vdash b : \tau$. Then for any $\pctx$ such that $\pctx \vdash \Gamma;\Delta$ and any $\yctx$ such that $\fv_{\mathcal{Y}}(b) \subseteq \yctx$,
\[
   \pctx; \yctx \vdash \rho(b) : \tau\quad \text{and}\quad \pctx; \yctx \vdash \cn(b) : \text{Formula}.
\]
  \label{prop:wellsorted}
\end{proposition}
\begin{proof}
    This proceeds by induction on the typing derivation.
    Let $\pctx$ be such that $\pctx \vdash \Gamma;\Delta$ and let $\yctx$ be such that $\fv_{\mathcal{Y}}(b)\subseteq \yctx$.
\begin{description}
      \item[\rname{T-Vltn}] This is immediate.
      \item[\rname{T-Split}]
          Then $b = (\kwlet\ x_1,\ldots ,x_{n},w_1 ,\ldots ,w_{n'} = \kwsplit\ b_1\ \kwin\ b_2)$, and necessarily that $\Gamma;\Delta \vdash b_1 : \ut_{1}\times \cdots \times \ut_{n}\times \lint_{1}\times \cdots \times \lint_{n'}$, and
          \[
              \Gamma, x_1 : \tau_1,\ldots ,x_{n}, \rd{w_{1}} : \rd{\lint_1} ,\ldots ,\rd{w_{n'}} : \rd{\lint_{n'}};\Delta, w_{1} : \lint_1 ,\ldots ,w_{n} : \lint_{n'} \vdash b_2 : \tau.
          \]
          We also have
          \[
              \rho(b) = \rho(b_2)\{x_{i}\mapsto \pi_{i} \rho(b_1)\mid 1 \leq i \leq n\} \{w_{i} \mapsto \pi_{n + i}\rho(b_1), \rd{w_{i}} \mapsto \pi_{n + i}\rho(b_1)\mid 1 \leq i \leq n'\}.
      \]
      Set
      \begin{align*}
          S_{\var} &= \{x_{i }\mapsto  \pi_{i} \rho(b_1)\mid 1 \leq i \leq n\}\\
          S_{\lvar} &= \{w_{i}\mapsto \pi_{n + i}\rho(b_1)\mid 1 \leq i \leq n'\}\\
          S_{\rd{\lvar}} &= \{\rd{w_{i}}\mapsto \pi_{n + i}\rho(b_1)\mid 1 \leq i \leq n'\}.
      \end{align*}
      In this case, $\rho(b) = S_{\var}S_{\lvar}S_{\rd{\lvar}}(\rho(b_2))$.
          By induction, $\pctx;\yctx \vdash \rho(b_1) : \ut_1 \times \cdots \times \ut_{n}\times \lint_{1} \times \cdots \times \lint_{n'}$. By the projection sorting rule, $\pctx;\yctx \vdash \pi_{i} \rho(b_1) : \ut_{i} $ for $1 \leq i \leq n$ and $\pctx;\yctx \vdash \pi_{n + i} \rho(b_1) : \lint_{i}$ for $1 \leq i \leq n'$.

          Let
          \[
              \pctx' = \pctx, x_1 : \ut_1 ,\ldots , x_{n} : \ut_{n}, \rd{w_1} : \lint_1 ,\ldots ,\rd{w_{n'}} : \lint_{n'}, w_{1} : \lint_{1} ,\ldots ,w_{n} : \lint_{n'}.
          \]
          Then
          \[
              \pctx' \vdash \Gamma, x_1 : \ut_1,\ldots ,x_{n} : \ut_{n},\rd{w_{1}} : \rd{\lint_1} ,\ldots ,\rd{w_{n'}} : \rd{\lint_{n'}} ;\Delta, w_1 : \lint_1 ,\ldots ,w_{n} : \lint_{n'}.
          \]
          As $\dom(S_{x})\cup \dom(S_{w})$ does not contain any variables from $\mathcal{Y}$, we also have $\fv_{\mathcal{Y}}(b_2)$. By induction, $\pctx' ; \yctx \vdash \rho(b_2) : \tau$.
          Because
          \begin{align*}
              \dom(S_{\var}) &= \{x_1 ,\ldots x_{n}\} \\
              \dom(S_{\lvar}) &= \{w_1 ,\ldots ,w_{n'}\}\\
              \dom(S_{\rd{\lvar}}) &= \{\rd{w_{1}},\ldots ,\rd{w_{n'}}\},
          \end{align*}
          we have $\pctx; \yctx \vdash S_{\var}S_{\lvar}S_{\rd{\lvar}}(\rho(b_2)) : \tau$.
          This means $\pctx; \yctx \vdash \rho(b) : \tau$. Using this fact and in a similar manner, we can argue $\pctx;\yctx \vdash \cn(b) : \text{Formula}$.
      \item[\rname{T-Ops}]
          Then $b = o(b_1,\ldots ,b_{n})$, and necessarily that $\Gamma ; \Delta\vdash  b_{i} : \ut_{i}$ where $o : \ut_{1},\ldots ,\ut_{n}\to \ut$. By induction, $\pctx; \yctx \vdash  \rho(b_{i}) : \ut_{i}$, and $\pctx; \yctx \vdash \cn(b_{i}) : \text{Formula}$. Then by the ops sorting rule, $\pctx; \yctx  \vdash \rho(b) : \tau$ and by the conjunction sorting rule, $\pctx; \yctx \vdash \cn(b) : \text{Formula}$.
      \item[\rname{T-Assert}]
        Then $b = \pt{b_1}{b_2}$, and necessarily $\Gamma;\Delta \vdash b_{1} : \mathsf{Bool}$, $\Gamma ; \Delta \vdash b_{2} : \tau$.
        By induction, $\pctx;\yctx \vdash \rho(b_1) : \mathsf{Bool}$, $\pctx;\yctx \vdash \rho(b_2) : \tau$, and $\pctx;\yctx \vdash \cn(b_1)$, $\pctx; \yctx \vdash \cn(b_{2}) : \text{Formula}$. From the first statement, we have $\pctx; \yctx \vdash \rho(b_1) = \mathsf{true} : \text{Formula}$. As $\rho(b) = \rho(b_2)$, we have $\pctx; \yctx \vdash \rho(b) : \tau$. From the conjuction sorting rule, $\pctx ;\yctx \vdash \cn(b) : \text{Formula}$.
      \item[\rname{T-Mark}]
          Then $b = \emark{a}{b_1}{b_2}\# s$, and necessarily that $\Gamma; \Delta \vdash b_1 : \mathsf{Intvl}$, $\Gamma; \Delta \vdash b_2 : \mathsf{Vltn}$.
          As $\rho(b) \in \fv_{\mathcal{Y}}(b)$, and we assume that $\fv_{\mathcal{Y}}(b)\subseteq \yctx$, then we have $\pctx;\yctx \vdash \rho(b) : \mathsf{Point}$.
          By induction, $\pctx; \yctx \vdash \rho(b_1) : \mathsf{Intvl}$, $\pctx;\yctx \vdash \rho(b_2) : \mathsf{Real}$, and $\pctx;\yctx \vdash \cn(b_1) : \text{Formula}$, $\pctx; \yctx \vdash \cn(b_2): \text{Formula}$. By the former two statements we can obtain through a straightforward series of sorting rules
          \[
              \pctx; \yctx \vdash \val_{a}([\tleft(\rho(b_1)), \rho(b)]) = \rho(b_2) : \text{Formula}.
          \]
          We can apply the sorting conjuction rule with this and the latter two statements to net
          $\pctx;\yctx \vdash \cn(b) : \text{Formula}$.

      \item[\rname{T-EvalPc}]
          Then $b = \eeval{a}{b'}$, and necessarily that $\Gamma; \Delta \vdash b' : \mathsf{Piece}$. By induction, $\pctx; \yctx \vdash \rho(b') : \mathsf{Piece}$ and $\pctx; \yctx \vdash \cn(b') : \text{Formula}$, which means by the Val sorting rule $\pctx; \yctx \vdash \rho(b) : \mathsf{Vltn}$ and since $\cn(b') = \cn(b)$, $\pctx; \yctx \vdash \cn(b) : \text{Formula}$.

      \item[\rname{T-EvalIntvl}] Almost identical to the argument for \rname{T-EvalPc}.

      \item[\rname{T-Div}]
          Then $b = \ediv{b_1}{b_2}$, and necessarily that $\Gamma;\Delta \vdash b_1 : \mathsf{Intvl}$, $\Gamma;\Delta \vdash b_2 : \mathsf{Point}$. By induction, $\pctx;\yctx \vdash \rho(b_1): \mathsf{Intvl}$, $\pctx;\yctx \vdash \rho(b_2): \mathsf{Point}$, and $\pctx; \yctx \vdash \vdash \cn(b_1)$, $\pctx;\yctx \vdash \cn(b_2) : \text{Formula}$. The former statements, through a straightforward set of sorting rules can be used to show both
          \[
              \pctx; \yctx \vdash ([\tleft(\rho(b_1)), \rho(b_2)], [\rho(b_2), \tright(\rho(b_1))]) : \mathsf{Intvl}\times \mathsf{Intvl}
          \]
          and
          \[
              \pctx; \yctx \vdash\tleft(\rho(b_1)) \leq \rho(b_2) \leq \tright(\rho(b_1)) : \text{Formula}.
          \]
We can apply the sorting conjuction rule with this and the latter two statements from the induction hypothesis to net
\[
    \pctx; \yctx \vdash \cn(b) : \text{Formula}.
\]
As $\rho(b) = ([\tleft(\rho(b_1)), \rho(b_2)], [\rho(b_2), \tright(\rho(b_1))])$, we have
$\pctx;\yctx \vdash \rho(b) : \mathsf{Intvl}\times \mathsf{Intvl}$.

      \item[\rname{T-Piece}] This is very similar to that of \rname{T-Ops}. \qed
    \end{description}
\end{proof}

We now move toward semantic results about constraints.
Fix an interpretation and variable assignment. We say that two substitutions $S_1$ and $S_2$ are interpreted the same $\dom(S_1) =\dom(S_2)$ and for all $x\in \dom(S_1)$, $\interp,\assn \vDash S_1(x) = S_2(x)$. We have the following related lemma.
\begin{lemma}
    Let $S_1$ and and $S_2$ be two logical substitutions. If $S_1$ and $S_2$ are interpreted the same, then $\interp, \assn \vDash S_1(t) = S_2(t)$ for any term $t$, and  $\interp, \assn\vDash S_1(f) \iff S_2(f)$.
    \label{prop:semsub}
\end{lemma}

Since our constraints contain variables intended to represent points marked within the protocol, it will be important to ensure that our variable assignment maps the variables to the proper value for a given derivation.
\begin{definition}
    Let $\assn$ be a variable assignment and $D$ a derivation. We say that $\assn$ \emph{agrees with} $D$ if for any $\#s$, occuring such that $\emark{a}{e_1}{e_2}\#s \eval r\pnt$ occurs in $D$, we have that $\assn(y_{\#s}) = r\pnt$.
\end{definition}
If we assume that each mark within a protocol has a unique identifier, any derivation involving that protocol will always have an assignment that agrees with it.

    \begin{proposition}[Soundness]
  Let $\vset$ be a valuation set and suppose $b$ is disjoint and  $\cdot \vdash b : \tau$.
  If $D : b\eval_{\vset} \v$, then $\interp_{\vset}, \assn\vDash \cn(b)$ and $\sem{\rho(b)}_{\interp}^{\assn} = \unrd{\v}$ for any variable assignment $\assn$ that agrees with $D$.
  \label{prop:soundness}
\end{proposition}
\begin{proof}
    This proceeds by induction on $D$.
    Suppose that $D : b \eval_{\vset} \v$ and suppose $\assn$ agrees with $D$.
    We omit $\vset$ as a subscript from here on. We break up our cases based on the last rule applied.

\begin{description}
      \item[\rname{Split}]
          Then $b = (\kwlet\ x_1,\ldots ,x_{n}, w_1,\ldots ,w_{n'} = \kwsplit\ b_1\ \kwin\ b_2)$. It must be that $b_1 \eval (\v_1 ,\ldots ,\v_{n+ n'})$ and $S(b_2) \eval \v$ where we set
          \[
              S = \{x_{i} \mapsto \v_{i} \mid 1 \leq i \leq n\}\{w_{i} \mapsto \v_{n  +i}, \rd{w_{i}} \mapsto \rd{\v_{n  +i}} \mid 1 \leq i \leq n'\}.
      \]
      Immediately by induction, $\interp, \assn \vDash \cn(b_1) $ and $\sem{\rho(b_1)}_{\interp}^{\assn} = (\unrd{\v_1} ,\ldots ,\unrd{\v_{n+ n'}})$. Also immediately by induction, $\interp,\assn \vDash \cn(S(b_2))$ and $\sem{\rho(S(b))}_{\interp}^{\assn} = \unrd{\v}$. By \Cref{prop:subswapf}, $\cn(S(b_2)) = \unrd{S}(\cn(b_2))$ and $\rho(S(b_2)) = \unrd{S}(\rho(b_2))$. Therefore $\interp,\assn \vDash\unrd{S}(\cn(b_2))$ and $\sem{\unrd{S}(\rho(b_2))}_{\interp}^{\assn} = \unrd{\v}$.

          Set
          \[
              S_{\rho} = \{x_{i} \mapsto \pi_{i}\rho(b_1) \mid 1 \leq i \leq n\}\{w_{i} \mapsto \pi_{n + i}\rho(b_1), \rd{w_{i}} \mapsto \pi_{n + i}\rho(b_1) \mid 1 \leq i \leq n'\}.
          \]
          As $\sem{\rho(b_1)}_{\interp}^{\assn} = (\unrd{\v_1} ,\ldots, \unrd{\v_{n + n'}})$, we have $\sem{\pi_{i}\rho(b_1)}_{\interp}^{\assn} = \unrd{\v_{i}}$. This means $S_{\rho}$ and $\unrd{S}$ are interpreted the same. Then by
          \Cref{prop:semsub}, $\interp,\assn \vDash_{\vset} S_{\rho}(\cn(b_2))$, and $\sem{S_{\rho}(\rho(b_2))}_{\interp}^{\assn} = \unrd{\v}$.
          Because $\cn(b) = \cn(b_1)\wedge S_{\rho}(\cn(b_2))$ and $\rho(b) = S_{\rho}(\rho(b_2))$, we have established this case.
      \item[\rname{Ops}]
          Then $b = o(b_1,\ldots ,b_{n})$, $b_{i}\eval \v_{i}$, $\rd{\sem{o}(\unrd{\v_1} ,\ldots ,\unrd{\v_{n}})} = \v$, $\rho(b) = \mathtt{o}(\rho(b_1),\ldots ,\rho(b_n))$, and $\cn(b) = \cn(b_1)\land \cdots \land \cn(b_{n})$. By induction, $\sem{\rho(b_{i})}_{\interp}^{\assn} = \unrd{\v_{i}}$ and $\interp,\assn \vDash \cn(b_{i})$.  Therefore $\interp,\assn\vDash \cn(b)$. Also,
          \[
              \sem{\rho(b)}_{\interp}^{\assn} = \sem{o}(\sem{\rho(b_1)}_{\interp}^{\assn} ,\ldots , \sem{\rho(b_n)}_{\interp}^{\assn} ) = \sem{o} (\unrd{\v_1},\ldots ,\unrd{\v_{n}}) = \unrd{\v}.
      \]
      \item[\rname{Assert}]
        Then $b = \pt{b_1}{b_2}$, and $b_1 \eval \mathsf{true}$, $b_2\eval \v$.
        By induction, $\sem{\rho(b_{1})}_{\interp}^{\assn} = \mathsf{true}$, $\interp,\assn \vDash \cn(b_1)$, $\sem{\rho(b_2)}_{\interp}^{\assn} = \unrd{\v}$, and $\interp,\assn \vDash \cn(b_2)$. This immediately gives us $\sem{\rho(b)}_{\interp}^{\assn} = \unrd{\v}$ and $\interp,\assn \vDash \cn(b)$.
      \item[\rname{Mark}]
          Then $b = \emark{a}{b_1}{b_2}$, and $b_1 \eval [r_1, r'_1]$ for some $r_1,r'_1\in [0,1]$, $b_2\eval \v_{a'}(P)$, with $V_{a}[r_1,r] = V_{a'}(P)$. Let $y =\rho(b)$.
          Since $\assn$ agrees with $D$, $\assn(y) = r\pnt$. This already establishes $\sem{\rho(b)}_{\interp}^{\assn} = \unrd{\v}$.

          By induction, $\sem{\rho(b_1)}_{\interp}^{\assn} = [r_1,r'_{1}]$, $\sem{\rho(b_2)}_{\interp}^{\assn} = V_{a'}(P)$, and both $\interp,\assn \vDash\cn(b_1)$ and $\interp,\assn\vDash\cn(b_2)$.
      As $V_{a}[r_1, r] = V_{a'}(P)$ and $\sem{\rho(b)}_{\interp}^{\assn} = r\pnt$, we have $\interp,\assn\vDash (\val_{a}[\tleft(\rho(b_1)), \rho(b)] = \rho(b_2))$.
          In culmination, $\interp,\assn\vDash\cn(b)$.

      \item[\rname{EvalPc}]
          Then $b = \eeval{a}{b'}$, and $b' \eval P_{i}[r_{i}, r'_{i}]$ and $\v = V_{a}(P_{i}[r_{i},r'_{i}])$. By induction, $\sem{\rho(b')}_{\interp}^{\assn} = P_{i}[r_{i},r'_{i}]$ and $\interp,\assn \vDash \cn(b')$. The former gives us $\sem{\rho(b)}_{\interp}^{\assn} = \unrd{\v}$. The latter gives us $\interp,\assn \vDash \cn(b)$ since $\cn(b) = \cn(b')$.

      \item[\rname{EvalIntvl}] The argument here is just a more simple verion of the argument for \rname{EvalPc}.

      \item[\rname{Div}]
          Then $b = \ediv{b_1}{b_2}$, and $b_1 \eval [r_1,r'_{1}]$, $b_2 \eval r_2$ with $r_1 \leq r_2 \leq r'_{1}$. By induction, $\sem{\rho(b_1)}_{\interp}^{\assn} = [r_1, r'_{1}]$, $\sem{\rho(b_2)}_{\interp}^{\assn} = r_2\pnt$, and both $\interp,\assn \vDash\cn(b_1)$ and $\interp,\assn\vDash \cn(b_2)$.  We can conclude with the former inductive statements that $\sem{\rho(b)}_{\interp}^{\assn} = ([r_1, r_{2}], [r_{2}, r'_{1}]) = \unrd{\v}$. Since $r_1 \leq r_2 \leq r'_{1}$, we also have $\interp,\assn \vDash \tleft(\rho(b_1))\leq \rho(b_2) \leq \tright(\rho(b_1))$ which we can use to conclude $\interp,\assn \vDash \cn(b)$.

      \item[\rname{Tup}] Then $b = (b_1,\ldots ,b_{n})$, $b_{i} \eval \v_{i}$, $\v = (\v_1 ,\ldots ,\v_{n})$, $\rho(b) = (\rho(b_1) ,\ldots ,\rho(b_{n}))$, and $\cn(b) = \cn(b_1) \land \cdots \land \cn(b_{n})$. By induction, $\sem{\rho(b_{i})}_{\interp}^{\assn} = \unrd{\v_{i}}$, and $\interp,\assn \vDash \cn(b_{i})$. Clearly $\interp,\assn \vDash \cn(b)$. Also,
          \[
              \sem{\rho(b)}_{\interp}^{\assn} = (\sem{\rho(b_{1})}_{\interp}^{\assn}, \sem{\rho(b_{n})}_{\interp}^{\assn}) = (\unrd{\v_1} ,\ldots ,\unrd{\v_{n}}) = \unrd{\v}.
          \]

      \item[\rname{Piece}] This is very similar to the case for \rname{Tup}.

    \end{description}
    \qed
\end{proof}

\begin{proposition}[Completeness]
    \label{prop:completeness}
    Let $\vset$ be a valuation set, and $\Gamma; \Delta \vdash b : \tau$. Let $\gamma\vDash \Gamma$ and $\delta \vDash \Delta$. Suppose $\interp_{\vset}, \assn \vDash \cn(\gamma;\delta (b))$. Then there is a derivation $D$ such that $\assn$ agrees with $D$ and $D : b \eval_{\vset} \v$ and $\unrd{\v} = \sem{\rho(\gamma;\delta (b))}_{\interp_{\vset}}^{\assn}$.
\end{proposition}
\begin{proof}
    We prove for any $\assn$ and any $\gamma \vDash \Gamma$, $\delta \vDash \Delta$ such that $\interp_{\vset},\assn \vDash \cn(\gamma;\delta (b))$, there is $D : b \eval \v$, $\unrd{\v} = \sem{\rho(\gamma;\delta b))}_{\interp}^{\assn}$ and $\assn$ agrees with $D$.
    This proceeds by induction on $\Gamma;\Delta \vdash b : \tau$. Suppose that $\interp, \assn\vDash_{\vset} \cn(\gamma;\delta(b))$. We break up our cases based on the structure of $b$.
    The cases for $\rname{Bool}$, $\rname{Real}$, $\rname{Point}$, $\rname{Piece}$, and $\rname{Cake}$ are obvious.
     \begin{description}
         \item[\rname{Var}] This is also obvious.
             Then $b = x$ and $x\in \var$, or $b = \rd{w}$ and $w \in \lvar$. The former is obvious so suppose the latter.
         Then $\Gamma(\rd{w}) = \rd{\v}$ and $\gamma;\delta(\rd{w}) \eval \rd{\v}$. Also, $\sem{\gamma;\delta(\rd{w})}_{\interp}^{\assn} = \sem{\rd{\v}}_{\interp}^{\assn} = \v = \unrd{\rd{\v}}$, which completes this case.
         \item[\rname{AffVar}]  This is obvious.
         \item[\rname{Split}] Then
             \begin{align*}
             \cn(\gamma;\delta (b)) = \cn(\gamma; \delta (b_1))\land S_{b}(\cn(\gamma ;\delta (b_2))) \quad \text{and}\quad \rho(\gamma;\delta (b)) = S_{b}(\rho(\gamma;\delta (b_2)))
             \end{align*}
             where we set
          \[
              S_{\rho} = \{x_{i} \mapsto \pi_{i}\rho(b_1) \mid 1 \leq i \leq n\}\{w_{i} \mapsto \pi_{n + i}\rho(b_1), \rd{w_{i}} \mapsto \pi_{n + i}\rho(b_1) \mid 1 \leq i \leq n'\}.
          \]
          We also have
          \begin{align*}
              &\Gamma;\Delta_1 \vdash b_1 : \ut_1 \times \cdots \times \ut_{n} \times \lint_{1} \times \cdots \times \lint_{n'} \\
              &\Gamma; x_1 : \ut_1 ,\ldots ,x_{n} : \ut_{n}, \rd{w_1} : \rd{\lint_{1}} ,\ldots ,\rd{w_{n'}} ; \Delta_2, w_1 : \lint_1 ,\ldots ,w_{n'} : \lint_{n'}\vdash b_2 : \tau.
          \end{align*}
          According to \Cref{prop:partitiondelta}, we can partition $\delta$ into $\delta_1$ and $\delta_2$ such that $\gamma;\delta_1 (b_1) = \gamma;\delta(b_1)$, $\gamma;\delta_2 (b_2) = \gamma;\delta(b_2)$ and $\delta_1 \vDash \Delta_1$, $\delta_2 \vDash \Delta_2$.
             By induction, there is $D_1$ agreeing with $\assn$ for which
             \[
                 D_1 : \gamma;\delta_1 (b_1) \eval \v',
             \]
             and $\unrd{\v'} = \sem{\rho(\gamma;\delta_1 (b_1))}_{\interp}^{\assn}$.
             As $\Gamma;\Delta_1 \vdash b_1 : \ut_1 \times \cdots \times \ut_{n}\times \lint_{1} \times \cdots \times \lint_{n}$,
             we have $\cdot \vdash \gamma;\delta_1 (b_1) : \ut_1 \times \cdots \times \ut_{n}\times \lint_{1} \times \cdots \times \lint_{n}$.
             by \Cref{prop:typesoundness} we have $\v' = (\v_1 ,\ldots ,\v_{n+ n'})$, where $\v_{i} \in \sem{\ut_{i}}$ for $1 \leq i \leq n$ and $\v_{n + i}\in \sem{\lint_{n'}}$ for $1 \leq i \leq n'$. This means $\sem{\pi_{i} \rho(\gamma;\delta_1 (b_1))}_{\assn}^{\interp} = \unrd{\v_{i}}$ for all $i$. By \Cref{prop:semsub},
             \begin{equation}
                 \interp,\assn \vDash S_{b'}(\cn(\gamma;\delta_2 (b_2))) \quad \text{and} \quad \sem{S_{b}(\rho(\gamma;\delta_2 (b_2)))}_{\interp}^{\assn} = \sem{S_{b'}(\rho(\gamma;\delta_2 (b_2)))}_{\interp}^{\assn}
                 \label{eq:complsplit1}
             \end{equation}
             where $S_{b'} = \{x_{i} \mapsto \v_{i} \mid 1 \leq i \leq n + n'\}$. If we set
             \[
                 S = \{x_{i} \mapsto \v_{i} \mid 1 \leq i \leq n\}\{w_{i} \mapsto \v_{n + i}\mid 1 \leq i \leq n'\} \{\rd{w_{i}} \mapsto \rd{\v_{n + i}} \mid 1 \leq i \leq n'\}
\]
             by \Cref{prop:subswapf},
             \begin{equation}
                 \interp,\assn \vDash \cn(S(\gamma;\delta_2 (b_2)) ) \quad \text{and}\quad \sem{\rho(S(\gamma;\delta_2 (b_2)))}_{\interp}^{\assn} =\sem{S_{b'}(\rho(\gamma;\delta_2 (b_2)))}_{\interp}^{\assn}.
                 \label{eq:complsplit2}
             \end{equation}
             If we let $\gamma' = \gamma \{x_{i} \mapsto \v_{i} \mid 1 \leq i\leq n\}\{\rd{w_{i}} \mapsto \rd{\v_{n + i}} \mid 1 \leq i\leq n'\}$ and $\delta' = \delta_2  \{w_{i} \mapsto \v_{n + i}\mid 1 \leq i \leq n'\}$, then
             \begin{equation}
               S(\gamma;\delta_2 (b_2)) = \gamma';\delta' (b_2).
               \label{eq:complsplit3}
             \end{equation}
             Evidently $\gamma' \vDash \Gamma, x_1 : \ut_{1},\ldots , x_{n} : \ut_{n}, \rd{w_{1}} : \rd{\lint_1},\ldots ,\rd{w_{n'}} : \rd{\lint_{n'}}$ and $\delta' \vDash \Delta_2, w_1 : \lint_1 ,\ldots ,w_{n'} : \lint_{n'}$.

             By induction, there is $D_2$ agreeing with $\assn$ such that $D_{2} : \gamma';\delta'(b_2)\eval \v $ and $\unrd{\v} = \sem{\rho(\gamma';\delta'( b_2))}_{\interp}^{\assn}$. By \Cref{eq:complsplit3}, we can conclude $D : \gamma;\delta (b) \eval \v$, with $D_1$ and $D_2$ being the premises of $D$. As $\assn$ agrees with both $D_1$ and $D_2$, $\assn$ also agrees with $D$.

     By combining \Cref{eq:complsplit1}, \Cref{eq:complsplit2}, and \Cref{eq:complsplit3}, we obtain that $\sem{\rho(\gamma;\delta (b))}_{\interp}^{\assn} = \sem{\rho(\gamma'; \delta' (b_2))}_{\interp}^{\assn}$. This means $\unrd{\v} = \sem{\rho(\gamma;\delta (b))}_{\interp}^{\assn}$, concluding this case.

 \item[\rname{Ops}] Then $\rho(b) = o(\rho(b_1),\ldots ,\rho(b_{n}))$, $\cn(b) = \cn(b_1) \land \cdots \land \cn(b_{n})$ and $o : \ut_1 ,\ldots ,\ut_{n} \to \tau$ so $\Gamma;\Delta_{i} \vdash b_{i} : \ut_{i}$, where $\Delta = \Delta_1,\ldots ,\Delta_{n}$. According to $\Cref{prop:partitiondelta}$, $\gamma;\delta_{i} (b_{i}) = \gamma;\delta(b_{i})$ and $\delta_{i} \vDash \Delta$. Then by induction, $D_{i} : b_{i} \eval \v_{i}$ and $\unrd{\v_{i}} = \sem{\gamma;\delta_{i}(b_{i})}_{\interp}^{\assn}$ and $\assn$ agrees with $D_{i}$. By \Cref{prop:typesoundness}, $\v_{i} \in \sem{\ut_{i}}$. This enables us to apply \rname{Ops} so $\gamma;\delta(b) \eval \sem{o}(\v_1 ,\ldots ,\v_{n})$.
     We also have $\sem{\gamma;\delta(\rho(b))}_{\interp}^{\assn} = \sem{o}(\sem{\gamma;\delta_{1}(b_1)}_{\interp}^{\assn},\ldots ,\sem{\gamma;\delta_{n}(b_n)}_{\interp}^{\assn})$.

 \item[\rname{Mark}] Then $\rho(b) = y$ for some $y\in \mathcal{Y}$, $\cn(b) = \cn(b_1) \land \cn(b_2) \land (\val_{a}([\tleft(\rho(b_1)), \rho(b)]) = \rho(b_2))$, and $\tau = \mathsf{Point}$, $\Gamma;\Delta_1 \vdash b_1 : \rd{\mathsf{Intvl}}$, $\Gamma;\Delta_2 \vdash b_2 : \mathsf{Vltn}$, where $\Delta = \Delta_1, \Delta_2$.
     According to \Cref{prop:partitiondelta}, $\gamma;\delta_1 (b_1) = \gamma;\delta(b_1)$, $\gamma;\delta_2(b_2) = \gamma;\delta(b_2)$, and $\delta_1 \vDash \Delta_1$, $\delta_2 \vDash \Delta_2$. By induction, $D_1 : \gamma;\delta_1 (b_1 ) \eval \v_1$, where $\unrd{\v_1} = \sem{\rho(\gamma;\delta_1(b_1))}_{\interp}^{\assn}$ and $\assn$ agrees with $D_1$. By \Cref{prop:typesoundness} $\v_1 = \rd{[r, r']}$ for some $r, r' \in \mathbb{R}$.
     Also by induction, $D_2 : \gamma;\delta_2(b_2) \eval \v_2$, where $\unrd{\v_2} = \sem{\gamma;\delta_2(\rho(b_2))}_{\interp}^{\assn}$ and $\assn$ agrees with $D_2$. By \Cref{prop:typesoundness}, $\v_2 = \val_{a'}(P)$ for some piece or interval $P$ and some agent $a'$.

     Now $\sem{\tleft(\rho(\gamma;\delta_1(b)))}_{\interp}^{\assn} = r\pnt$.
     As $\interp,\assn \vDash \cn(b)$, we have $V_{a}[r, \assn(y)] = V_{a'}(P)$. So using $D_1$, $D_2$, and $V_{a}[r, \assn(y)] = V_{a'}(P)$ as premises for \rname{Mark}, we can apply it to obtain $D: \gamma;\delta(b) \eval \assn(y)$ for which $\assn$ evidently agrees with $D$. Now $\unrd{\assn(y)} = \assn(y) = \sem{y}_{\interp}^{\assn}$, completing this case.

 \item[\rname{EvalPc}] Then $\Gamma;\Delta \vdash b : \mathsf{Vltn}$, and $\Gamma;\Delta \vdash b' : \mathsf{Piece}$. By induction, there is $D' : b' \eval \v'$ and $\unrd{\v'} = \sem{\rho(\gamma;\delta(b'))}_{\interp}^{\assn}$, and $\assn$ agrees with $D'$. By \Cref{prop:typesoundness}, $\v = \rd{P}$ for some piece $P$. We can apply \rname{EvalPc} to obtain $\gamma;\delta(b)\eval V_{a}(P)$.
     Now
     \begin{align*}
         \sem{\rho(\gamma;\delta(b))}_{\interp}^{\assn} = \sem{\val_{a}(\rho(\gamma;\delta(b')))}_{\interp}^{\assn} = V_{a}(\sem{\rho(\gamma;\delta(b'))}_{\interp}^{\assn}) = V_{a}(P) = \unrd{V_{a}(P)}.
     \end{align*}
     The argument for when $\Gamma;\Delta \vdash b' : \mathsf{Intvl}$ is nearly identical.

 \item[\rname{EvalIntvl}] The argument here is nearly identical to the argument for \rname{EvalPc}.
 \item[\rname{Div}] Then $\rho(b) = ([\tleft(\rho(b_1)), \rho(b_2)], [\rho(b_2), \tright(\rho(b_1))])$, $\cn(b) = \cn(b_1)\land \cn(b_2)\land \tleft(\rho(b_1))\leq  \rho(b_2) \leq \tright(\rho(b_1))$, $\Gamma;\Delta_1 \vdash b_1 : \mathsf{Intvl}$, $\Gamma;\Delta_2 \vdash b_2 : \mathsf{Point}$, and $\Delta = \Delta_1, \Delta_2$. By \Cref{prop:partitiondelta}, $\gamma;\delta_1(b_1) = \gamma;\delta(b_1)$ and $\gamma;\delta_2(b_2) = \gamma;\delta(b_2)$, and $\delta_1 \vDash \Delta_1$, $\delta_2 \vDash \Delta_2$. So by induction, $D_1 : \gamma;\delta_1 (b_1)\eval [r,r']$, $[r,r'] = \sem{\rho(\gamma;\delta_1(b_1))}_{\interp}^{\assn}$ and $\assn$ agrees with $D_1$. Also by induction, $D_2 :\gamma;\delta_2 (b_2) \eval r_2\pnt$, $r_2\pnt = \sem{\rho(\gamma;\delta_2(b_2))}_{\interp}^{\assn}$, and $\assn$ agrees with $D_2$.
     Thus, $\sem{\tleft(\rho(\gamma;\delta_1(b_1)))}_{\interp}^{\assn} = r_1$ and $\sem{\tright(\rho(\gamma;\delta_1(b_1)))}_{\interp}^{\assn} = r'_1$.
     Since $\interp,\assn \vDash \cn(b)$, we have $r_1 \leq r_2 \leq r'_{1}$. Therefore we can apply \rname{Div} with $D_1$ and $D_2$ as premises to obtain
     $D : \gamma;\delta(b) \eval ([r_1, r_2], [r_2, r'_{1}])$, and $\assn$ agrees with $D$ as it agreed with $D_1$ and $D_2$. We are done as $\sem{([\tleft(\rho(b_1)), \rho(b_2)], [\rho(b_2), \tright(\rho(b_1))])}_{\interp}^{\assn} = ([r_1, r_2], [r_2, r'_{1}]) =\unrd{([r_1, r_2], [r_2, r'_{1}])}$.

\item[\rname{Assert}] Then $\rho(b) = \rho(b_2)$, $\cn(b) = (\rho(b_1) = \kwtrue)\land \cn(b_1)\land \cn(b_2)$, and $\Gamma;\Delta_1 \vdash b_1 : \mathsf{Bool}$, $\Gamma;\Delta_2 \vdash b_2 : \tau$, where $\Delta = \Delta_1, \Delta_2$. According to \Cref{prop:partitiondelta}, $\gamma;\delta_1 (b_1) = \gamma;\delta(b_1)$, $\gamma;\delta_2 (b_2)=\gamma;\delta(b_2) $, and $\delta_1 \vDash \Delta_1$, $\delta_2 \vDash \Delta_2$. By induction, $D_1 : \gamma;\delta_1(b_1)\eval \kwtrue$, where $\assn$ agrees with $D_1$. Also by induction, $D_2 : \gamma;\delta_2(b_2) \eval \v$ where $\unrd{\v} = \sem{\rho(\gamma;\delta_2(b_2))}_{\interp}^{\assn}$, and $\assn$ agrees with $D_2$. This allows us to use \rname{Assert} to obtain $D : \gamma;\delta(b)\eval \v$, with premises $D_1$ and $D_2$. As $D$ is composed in this way, $\assn$ agrees with $D$. Since $\rho(b) = \rho(b_2)$, we are done with this case.

\item[\rname{Tup}] Then $b = (b_1 ,\ldots ,b_{n + n'})$, $\rho(b) = (\rho(b_1),\ldots ,\rho(b_{n+ n'}))$, $\cn(b) = \cn(b_1) \land \cdots \land \cn(b_{n + n'})$, and  $\Gamma ;\Delta_{i} \vdash b :  \ut_i$ for $1 \leq i \leq n$ and $\Gamma; \Delta_{n + i} \vdash b_{n + i} : \lint_{i}$ for $1 \leq i \leq n'$, and $\Delta = \Delta_1 ,\ldots ,\Delta_{n + n'}$. According to \Cref{prop:partitiondelta}, $\gamma;\delta_{i}(b_{i}) = \gamma;\delta(b_{i})$, and $\delta_{i} \vDash \Delta_{i}$ for $1 \leq i \leq n + n'$.
     By induction, $D_{i} : \gamma;\delta_{i}(b_{i})\eval \v_{i}$, where $\unrd{\v_{i}} = \sem{\rho(\gamma;\delta_{i}(b_{i}))}_{\interp}^{\assn}$, and $\assn$ agrees with $D_{i}$ for $1 \leq i \leq n + n'$. Using $D_1 ,\ldots ,D_{n + n'}$ as premises, we can apply \rname{Tup} to obtain $D : \gamma;\delta(b) \eval (\v_1 ,\ldots ,\v_{n + n'})$ where $\assn$ agrees with $D$. We have
     \begin{align*}
         \unrd{(\v_1 ,\ldots ,\v_{n + n'})} = (\unrd{\v_1} ,\ldots ,\unrd{\v_{n + n'}}) &= (\sem{\rho(\gamma;\delta_1(b_1))}_{\interp}^{\assn} ,\ldots ,\sem{\rho(\gamma;\delta_{n+ n'}(b_{n + n'}))}_{\interp}^{\assn}) \\
 &= \sem{(\rho(\gamma;\delta_1(b_1)) ,\ldots ,\rho(\gamma;\delta_{n+ n'}(b_{n + n'})))}_{\interp}^{\assn}\\
 &= \sem{\rho(\gamma;\delta((b_1 ,\ldots ,b_{n + n'})))}_{\interp}^{\assn},
     \end{align*}
     concluding this case.

 \item[\rname{Piece}] This is nearly identical to the tuple case.
\end{description}
\end{proof}

Soundness and completeness results lead to the following corollary.
\ccharacterization*

We now argue for the following theorem:
\solvingthm*
We first generalize to the envy-freeness formula to formulae that have mild restrictions.
\begin{definition}
    We say that a formula $F$ is a \emph{well-formed property} if $x : \mathtt{S}_{\tau} \vdash F : \text{Formula}$, and $F$ contains no point values.
\end{definition}
Here, $x$ stands in for what a protocol might possibly output. The point value restriction is not overly restrive, as most interesting properties do not reference constant points, however, our theory could be modified slightly to relax this restriction. \Cref{eq:envyfree} is easily seen to be a well-formed property.
We also require some mild assumptions on our expressions.
\begin{definition}
    We say that an expression $e$ is \emph{well-formed} if $e$ is closed, well-typed, disjoint, $e$ does not contain any point values (besides 0 and 1), and each mark subexpression contains a unique identifier.
\end{definition}
Similar to protocol properties, very few interesting protocols consider constant points (besides 0 and 1) within the cake.

Here we also generalise $e \vDash E(x)$ to arbitrary formulae.
Given an expression $\cdot \vdash e : \tau$ and a formula $F$ such that $x : \mathtt{s}_{\tau} \vdash F : \text{Formula}$, we say that $e \vDash F$ if for all valuation sets $\vset$, $e\eval_{\vset} \v$ implies $\interp_{\vset} \vDash F\fsub{x}{\lv}$. Now we state and prove our generalized theorem.
\begin{proposition}
    \label{prop:solvingthm}
    Suppose that $e$ is a well-formed expression and $\cdot \vdash e : \tau$. Let $F$ be a formula such that $x : \mathtt{S}_{\tau} \vdash F : \text{Formula}$.
    Then
    \begin{equation}
        \interp_{\vset} \vDash \bigwedge_{b \in B(e)} \forall \mathcal{Y}_{b}.(\cn(b)\Rightarrow F\fsub{x}{\rho(b)})
        \label{eq:solving}
    \end{equation}
    for all $\vset$ if and only if $e \vDash F$.
\end{proposition}

\begin{proof}
    ($\Rightarrow$)
    Let $\interp = \interp_{\vset}$.
    Suppose that $D: e \eval_{\vset} \v$. By \Cref{prop:pathtype} and \Cref{prop:path}, there is $b \in B(e)$ such that $\vdash b : \tau$ and $D: b\eval_{\vset} \v$. By \Cref{prop:soundness}, for any $\assn$ that agrees with $D$, $\interp,\assn \vDash c(b)$ and $\sem{\rho(b)}_{\interp}^{\assn} = \unrd{\v}$. By \Cref{eq:solving},  $\interp_{\vset},\assn \vDash F\fsub{x}{\rho(b)}$.
    By \Cref{prop:semsub}, $\interp_{\vset},\assn \vDash F\fsub{x}{\lv}$. $F\fsub{x}{\v}$ is a closed formula so $\interp_{\vset} \vDash F\fsub{x}{\lv}$.

    ($\Leftarrow$)
    Suppose \Cref{eq:solving} does not hold. Then there is some interpretation $\interp$ and variable assignment $\assn$ such that $\interp,\assn \vDash \cn(b)$ yet $\interp,\assn \not\vDash F\fsub{x}{\rho(b)}$. Then by \Cref{prop:completeness}, $b \eval_{\vset} \v$, where $\unrd{\v} = \sem{\rho(b)}_{\interp}^{\assn}$ for some $b \in B(e)$. By \Cref{prop:path}, $e \eval_{\vset} \v$. By \Cref{prop:semsub}, we have that $\interp,\assn \not\vDash F\fsub{x}{\lv}$, and since $F\fsub{x}{\lv}$ is closed, $\interp \not\vDash F\fsub{x}{\lv}$, which completes the proof.
\end{proof}

\subsection{Protocol execution replication}
\label{app:repl}
The goal of this subsection is to prove the following theorems:
\pured*
\pufequiv*
We first become more precise about derivations and points considered within them.
First recall our \hyperlink{dnotation}{derivation notation}. Briefly, given a derivation of an evaluation judgement $D$, we write $D : e \eval_{\vset} \v$ if the conclusion of $D$ is $e \eval_{\vset}\v$.
Define \[M(D) \triangleq \{r \mid r\pnt\ \text{is a subexpression of}\ \v, \v\ \text{appears in}\ D\}.\]
\Cref{thm:red} can now be expressed as:
\begin{theorem}
    Let $\uset$ and $\vset$ be a valuation set, and suppose $D: e\eval_{\vset} \v$. If $\uset$ and $\vset$ agree on $M(D)$, then $D: e \eval_{\uset} \v$.
\end{theorem}

\begin{proof}
    This goes by induction on $D$. We break it up into cases based on the last rule applied, of which there are 4 interesting cases.
     \begin{description}
         \item[\rname{E-EvalPc}]

             Here $e = \kweval(e')$ for some $e'$, $D': e'\eval_{\vset}P_{i = 1 }^{n}[r_{i},r'_{i}]$, and $\v = \vset_{a}(P_{i=1}^{n}[r_i,r'_i])$. Now $r_{i}, r'_{i}\in M(D)$ for all $i$ so $\cup_{i = 1}^{n}[r_{i},r'_{i}]$ has all boundary points in $M(D)$. Therefore, since $\uset$ and $\vset$ agree on $M(D)$, ${\uset}_{a}(\cup_{i =1}^{n}[r_{i},r'_{i}]) = \vset_{a}(\cup_{i=1}^{n}[r_{i},r'_{i}])$.
             By induction $D' : e'\eval_{\uset}P_{i= 1}^{n}[r_{i},r'_{i}]$, so
             we can apply \rname{E-EvalPc} to obtain $D : e\eval_{U_{\vset}}\v$.
         \item[\rname{E-EvalIntvl}] This argument is nearly identical to the argument for \rname{E-EvalPc}.
         \item[\rname{E-Mark}]
             Here $e = \kwmark_{a}(e_1,e_2)$ for some $e_1$ and $e_2$, $D_1 : e_1\eval_{\vset}[r_1,r_2]$, $D_2 : e_2\eval_{\vset}\v_{2}$, $\vset_{a}[r_1,r] = \v_{2}$ with $\v = r\pnt$. Now $r_1,r \in M(D)$. Since $\uset$ and $\vset$ agree on $M(D)$, ${{\uset}}_{a} [r_1,r] = \v_2$. By induction, $D_1 : e_1\eval_{{\uset}}[r_1,r_2]$ and $D_2: e_2\eval_{{\uset}}\v_2$, so applying \rname{E-Mark}, we obtain $D : e\eval_{{\uset}}\v$. \qed
    \end{description}
\end{proof}

We restate \Cref{thm:pufequiv} in the following way:
\begin{theorem}
    Suppose that $\yctx \vdash \varphi : \text{Formula}$. Let $\vset$ and $\uset$ be valuation sets and $\assn$ be a variable assignment. Suppose $\uset$ and $\vset$ agree on $\{ \sem{t}_{\interp_{\vset}}^{\assn}\mid \yctx \vdash t : \mathtt{Point}, t\ \text{occurs in}\ \varphi\}$. Then
    $\interp_{\vset},\assn \vDash \varphi \iff \interp_{\uset},\assn\vDash \varphi$.
    \label{thm:pufequivfull}
\end{theorem}
\begin{proof}
    By \Cref{prop:simpl}, it suffices to check this on simplified formulas. So assume that $\varphi$ is simplified. This argument proceeds by induction on the structure of $\varphi$.
    Set $M = \{\sem{t}_{\interp_{\vset}}^{\assn}\mid \yctx \vdash t : \mathtt{Point}, t\ \text{occurs in}\ \varphi\}$.
    The whole argument boils down to showing $\sem{\val_{a}(t)}_{\interp_{\vset}}^{\assn} = \sem{\val_{a}(t)}_{\interp_{U_{\vset}}}^{\assn}$.

    Let $\val_{a}(t)$ be a term contained in $\varphi$. Then either $t = [t_1, t'_1]$ for some $t_1,t'_1 \in R_{\mathtt{Point}}$ or $t = \cup ([t_1,t_1'] ,\ldots ,[t_{n},t'_{n}])$ for some $t_{i},t'_{i} \in R_{\mathtt{Point}}$.
    We note that $\sem{t_{i}}_{\interp_{\vset}}^{\assn} = \sem{t_{i}}^{\assn} = \sem{t_{i}}_{\interp_{\uset}}^{\assn}$ for all $i$, which also gives $\sem{t}_{\interp_{\vset}}^{\assn} = \sem{t}^{\assn} = \sem{t}_{\interp_{\uset}}^{\assn}$.
    Since $\uset$ and $\vset$ agree on $M$, we have that $\uset_{a}(\sem{t}^{\assn}) = \vset_{a}(\sem{t}^{\assn})$ as $\sem{t}^{\assn}$ has boundary points in $M$. This means $\sem{\val_{a}(t)}_{\interp_{\vset}}^{\assn} = \sem{\val_{a}(t)}_{\interp_{U_{\vset}}}^{\assn}$.
    \qed
\end{proof}

\subsection{Piecewise uniform valuations}
\label{app:pu}
Here we discuss the arguments for \hyperlink{def:pu}{piecewise uniform valuations}.
In this section we prove the following theorem.
\puexists*

We first start by introducing a form of a valuation and argue some facts about it. Following this, we consider a whole valuation set of this form.

Given an arbitrary valuation $V$ and a finite set of points $\{0,1\} \subseteq M$, we can have a piecewise uniform valuation replicate $V$ on pieces whose endpoints are in $M$.

Let $M = \{m_1,\ldots ,m_{k}\}$ be given such that $ 0= m_0 \leq m_1 < \cdots < m_{k}= 1$. Now let $V$ be any valuation on $[l,r]$, and define
\[
    \maxdens{V}{M} = \max_{0 \leq i < k} V[m_{i},m_{i+1}]/(m_{i + 1} - m_{i}).
\]
This quantity is referred to as the \emph{max density of $V$ on $M$}. It represents the largest density of $V$ on adjacent points in $M$. Recall that a piece uniquely determines a piecewise uniform valuation on that piece.
\begin{definition}
    Let $d \geq \maxdens{V}{M}$. We let $U_{V}(M,d)$ be the piecewise uniform valuation determined by the piece
    \[
        P = \bigcup _{i = 1}^{k} [m_{i} - V[m_{i - 1}, m_{i}]/d, m_{i} ] = [m_{1} - V[0, m_{1}]/d, m_{1} ] \cup \cdots \cup [m_{k} - V[m_{k - 1}, m_{k}]/d, m_{k} ].
    \]
    That is, $U_{V}(M,d) = U_{P}$, in the notation below \Cref{def:pu}.
\end{definition}
$U_{V}(M,d)$ is not well-defined for $ d< \maxdens{V}{M}$. For if this is the case, some intervals in the piece above would overlap further than their endpoints.

We first aim to show this valuation \hyperlink{def:repl}{agrees with $V$ on $M$}. To do so, we state a basic fact about $U_{V}(M,d)$.
\begin{lemma}
    $c(U_{V}(M,d)) = d$.
    \label{prop:uconst}
\end{lemma}
\begin{proof}
 Let $c$ be the constant associated with $U_{V}(M,d)$. Then
\begin{align*}
    V[l,r] &= U_{V}(M,d)[l, r] \\
                        &= U_{V}(M,d)[m_{0}, m_{k}]\\
                        &= c\sum_{k \geq i > 0}m_{i} - l_{i}\\
                        &= c\sum_{k \geq i  > 0}(m_{i} - (m_{i} - V[m_{i-1}, m_{i}]/d)\\
                        &= c\sum_{k \geq i > 0}V[m_{i-1},m_{i}] /d \\
                        &= cV[m_{0},m_{k}] /d \\
                        &= cV[l,r] /d
\end{align*}
which implies that $c = d$.
\qed
\end{proof}

\begin{restatable}[]{lemma}{pup}
     Let $P$ be any piece with boundary points entirely within $M$ then
\begin{equation}
    U_{V}(M,d)(P) = V(P).
\end{equation}
    \label{prop:puequiv}
\end{restatable}
\begin{proof}
    First, write $M = \{m_0, m_1,\ldots ,m_{k}\}$ where $0 = m_0 \leq m_1 < \cdots < m_{k} = 1$.
    We have for any $1 \leq i \leq k$,
\begin{align*}
    U_{W}(M,d)[m_{i-1}, m_{i}] &= d(m_{i} - (m_{i} - W[m_{i-1}, m_{i}]/d)) \\
                         &= d(W[m_{i-1}, m_{i}]/d) \\
                         &= W[m_{i-1}, m_{i+1}]
\end{align*}
and since $[m_{i}, m_{i'}]$ is the disjoint union of $[m_{j-1}, m_{j}]$ for $i < j \leq i'$, we have by disjointness,
\[
    U_{W}(M,d)[m_{i}, m_{i'}] = W[m_{i},m_{i'}]
\]
for any $0 \leq i \leq i' \leq k$. Since $P$ has boundary points entirely within $M$, we can write $P = [m_{i_1}, m'_{i_1}] \cup \cdots \cup [m_{i_{n}}, m'_{i_{n}}]$, for $m_{i_1}\leq m'_{i_{1}}\leq \cdots \leq m_{i_{n}} \leq m'_{i_{n}}$. This means
\begin{align*}
    U_{W}(M,d)(P) &= U_{W}(M,d)([m_{i_1}, m'_{i_1}] \cup \cdots \cup [m_{i_{n}}, m'_{i_{n}}]) \\
                  &= U_{W}(M,d)[m_{i_1}, m'_{i_1}] + \cdots + U_{W}(M,d)[m_{i_{n}}, m'_{i_{n}}] \\
                  &= W[m_{i_1}, m'_{i_1}] + \cdots + W[m_{i_{n}}, m'_{i_{n}}]\\
                  &= W([m_{i_1}, m'_{i_1}] \cup \cdots \cup [m_{i_{n}}, m'_{i_{n}}])\\
                  &= W(P).
\end{align*}
\qed
\end{proof}

We now extend both the above definition and lemma to valuation sets. We define for $\vset$ a valuation set, $\maxdens{\vset}{M} \triangleq \max_{a \in \agents}\maxdens{\vset_{a}}{M}$.
\begin{definition}
    Let $\vset$ be a valuation set, $M \supseteq \{0,1\}$ a set of points, and $d \geq \maxdens{\vset}{M}$. Then $U_{\vset}(M,d)$ is the valuation set with $U_{\vset}(M,d)_{a} \triangleq U_{\vset_{a}}(M,d)$ for all $a \in \agents$.
\end{definition}
Let $d \geq \maxdens{\vset}{M}$. In particular, this means for each $a \in \agents$, $d \geq \maxdens{\vset_{a}}{M}$. Thus,
by \Cref{prop:puequiv}, $U_{\vset}(M,d)$ and $\vset$ agree on $M$.
We now show that $U_{\vset}(M,d)$ is \hyperlink{def:er}{easily replaceable} on $M$.
If for each $m \in \noz{M}$, we set $l_{a}(m_{i}) \triangleq m_{i} - V[m_{i - 1}, m_{i}]/d$, then
\[P(U_{\vset}(M,d)_{a}) = \bigcup_{m \in \noz{M}}[l_{a}(m), m], \]
which satisfies condition (1). To satisfy condition (2), we note that by \Cref{prop:uconst}, $c(U_{\vset}(M,d)_{a}) = d$ for all $a$. This shows that $U_{\vset}(M,d)$ is easily replaceable on $M$.
\subsection{Formula simplification}
\label{app:simpl}
Before proving any theorems about simplification, we write the definition for $\pnt(t)$ and $\pnt(f)$ more carefully.
    \begin{definition}
        Let $t \in R_{\mathtt{Intvl}}$. Then $\pnt(t) = \{t_1, t_2\}$, where $t = [t_1, t_2]$ and $t_1, t_2 \in R_{\mathtt{Intvl}}$.
        Let $t \in R_{\mathtt{Piece}}$. Then
         \begin{align*}
            \pnt(t) \triangleq \pnt(t_1)\cup \cdots \cup \pnt(t_{n})
        \end{align*}
        where $t = \cup t_1 ,\ldots ,t_{n}$ for $t_{i} \in R_{\mathtt{Intvl}}$.
        Let $t\in R_{\mathtt{s}}$. We define $\pnt(t)$ as the union of all $\pnt(t')$ for $t'$ an interval or point in $t$. For a simplified formula $f$, we let $\pnt(f)$ denote the union of $\pnt(t)$ for all terms $t$ in $f$. We refer to the elements of $\pnt(f)$ as \emph{point atoms}.
    \end{definition}

We first aim to prove the following theorem:
\purepsound*
This requires the following lemma.
\begin{lemma}
    \label{prop:techsum}
    Suppose we have $\con{S}{\assn}{\uset}$.
    If $t \in R_{\mathtt{Piece}}$ and $\pnt(t)\subseteq S$, let $\underline{S}|_{t} = \{y \in S|_{t} \mid \forall y' \in S.\assn(y) = \assn(y') \Rightarrow y \leq_{S} y' \}$. Then
    \[
        \uset_{a}(\sem{t}^{\assn}) = c(\uset_{a})\sum_{y \in \underline{S}|_{t}} \assn(y) - \assn(z_{a,y}).
    \]
\end{lemma}
\begin{proof}
    Before proceeding, we give a formula for $\uset_{a}(\sem{t}^{\assn})$. Write $t = \cup [y_1, y'_{1}],\ldots ,[y_{n},y'_{n}]$ and let $\assn(S)|_{t} = \{m \in \assn(S) \mid \assn(y_{i}') \geq m > \assn(y_{i})\ \text{for some}\ 1\leq i \leq n\}$. Since $\sem{t}^{\assn} = \bigcup_{i = 1}^{n} [\assn(y_{i}),\assn(y_{i}')]$ and $\pnt(t)\subseteq S$, we have that $\partial\sem{t}^{\assn}\subseteq \assn(S)$. Therefore by (1) and by \Cref{eq:pusum}, we have
    \[
        \uset_{a} (\sem{t}^{\assn}) =  c(\uset_{a})\cdot \sum_{m \in \assn(S)|_{t}} (m - l_{a}(m)).
    \]
    To complete this argument, it suffices to show that $\assn$ gives a bijection between $\assn(S)|_{t}$ and $\underline{S}|_{t}$, for if $\assn(y) = m\in M$ and $y = \min\{y' \in S \mid \assn(y') = m\}$ then $\assn(y) - \assn(z_{a,y}) = m - l_{a}(m)$ according to (2), agreeing with the summand above.

    Let $y \in \underline{S}|_{t}$. As $y \in S|_{t}$, we have $y_{i}' \geq_{S} y >_{S} y_{i}$ for some $1 \leq i \leq n$. By (4), $\assn(y'_{i}) \geq \assn(y) \geq \assn(y_{i})$. It also must be that $\assn(y) > \assn(y_{i})$, as $y = \min\{y' \in S\mid \assn(y') = m\}$. Therefore $\assn(y) \in \assn(S)|_{t}$.
     Thus, $\assn|_{\underline{S}|_{t}} : \underline{S}|_{t} \to \assn(S)|_{t}$. It remains to show that $\assn|_{\underline{S}|_{t}}$ is injective and surjective.

    Injectivity is almost immediate.  Since $>_{S}$ is a total order, $\min \{y' \in S\mid \assn(y') = m\}$ is unique for each $m \in \assn(S)$. Therefore $\underline{S}|_{t}$ contains at most one $y$ such that $\assn(y) = m$ for each $m \in \assn(S)$.

    Let $m \in \assn(S)|_{t}$. Let $y_{m} = \min \assn|_{S}^{-1}(m)$. If we show that $y_{m}\in \underline{S}|_{t}$, surjectivity is established. As $m \in \assn(S)|_{t}$, $\assn(y_{i}') \geq m > \assn(y_{i})$ for some $1 \leq i \leq n$. Hence $\assn(y_{i}') \geq \assn(y_{m}) > \assn(y_{i})$, so by (4), $y' \geq_{S} y_{m} >_{S} y$. This means $y_{m} \in \sub{y}{y'}\subseteq S|_{t}$. As $y_{m} = \min \assn|_{S}^{-1}(m)$, $y_{m} \in \underline{S}|_{t}$. This completes the proof.
    \qed
\end{proof}

\begin{proof}[\Cref{prop:purepsound}]
    Let $\interp = \interp_{\uset}$.
    As $f$ is simplified, it only consists of predicates of the form given in \Cref{prop:lra}. $S$ does nothing to predicates of the form $\sum_{i}r_{i}d_{i} \geq \sum_{j}r_{j}d_{j}$, so we focus on those of the form
    \[
        \sum_{i}r_{i}\val_{a_{i}}(t_{i}) \geq \sum_{j} r_{j}\val_{a_{j}}(t_{j}),
    \]
    for $t_{i}$,$t_{j}\in R_{\mathtt{Intvl}}\cup R_{\mathtt{Piece}}$. Without loss of generality, we can assume that $t_{i}, t_{j} \in R_{\mathtt{Piece}}$, for if $t \in R_{\mathtt{Intvl}}$, then $\sem{\val_{a}(t_{i})}_{\interp}^{\assn} = \sem{\val_{a}(\cup t_{i})}_{\interp}^{\assn}$.
    We know that $\pnt(f) \subseteq S$ so $\pnt(t_{i})\subseteq S$ and $\pnt(t_{j})\subseteq S$ for all $i$ and $j$. Let $P_{i} = \sem{t_{i}}^{\assn}$ and $P_{j} = \sem{t_{j}}^{\assn}$.
Let $\underline{S}|_{t} = \{y \in S|_{t} \mid \forall y' \in S. \assn(y) = \assn(y') \Rightarrow y \leq_{S} y'\}$.
We have
\begin{align*}
    \sem{S(\val_{a_{i}}(t_{i}))}_{\interp}^{\assn}
                                                            &= \sum_{y\in S|_{t_{i}}} \assn(y) - \assn(z_{a_{i},y}) \\
                                                            &= \sum_{y\in \underline{S}|_{t}} \assn(y) - \assn(z_{a_{i},y}) + \sum_{y\in S|_{t_{i}}\setminus \underline{S}|_{t}} \assn(y) - \assn(z_{a_{i},y}) \\
                                                            &= \sum_{y \in \underline{S}|_{t}} \assn(y) - l_{a_{i}}(\assn(y)) + \sum_{y \in S|_{t} \setminus \underline{S}|_{t}} \assn(y) - \assn(y) \\
                                                            &= {{\uset}}_{a_{i}}(P_{i})/ d
\end{align*}
where we use $\con{S}{\assn}{\uset}$ (2) and (3) for the third equality, and \Cref{prop:techsum} for the fourth equality.
Similarly,
\begin{align*}
\sem{S(\val_{a_{j}}(t_{j}))}_{\interp}^{\assn} = {{\uset}}_{a_{j}}(P_{j})/ d.
\end{align*}

Now
\begin{align*}
        &\interp,\assn \vDash \sum_{i}r_{i}\val_{a_{i}}(t_{i}) \geq   \sum_{j}r_{j}\val_{a_{j}}(t_{j}) \\
        &\iff \sem{\sum_{i}r_{i}\val_{a_{i}}(t_{i})}_{\interp}^{\assn} \geq \sem{\sum_{j}r_{j}\val_{a_{j}}(t_{j})}_{\interp}^{\assn} \\
        &\iff \sum_{i} r_{i} {\uset}_{a_{i}}(P_{i}) \geq \sum_{j} r_{j} {{\uset}}_{a_{j}}(P_{j}) \\
        &\iff \sum_{i} r_{i} \sem{S(\val_{a_{i}}(t_{i}))}_{\interp}^{\assn}\geq \sum_{j} r_{j} \sem{S(\val_{a_{j}}(t_{j}))}_{\interp}^{\assn}\\
        &\iff \interp,\assn \vDash S\left(\sum_{i}r_{i}\val_{a_{i}}(t_{i}) \geq   \sum_{j}r_{j}\val_{a_{j}}(t_{j})\right).
    \end{align*}
    This completes our argument.
    \qed
\end{proof}

We now aim to prove the following theorem:
\simplsolvingenvy*

Like with \Cref{prop:solvingenvy}, we generalize to well-formed properties:
\begin{restatable}[]{theorem}{simplsolvingthm}
    Suppose $e$ is well-formed and $\cdot \vdash e : \tau$. Let $F$ be a well-formed property such that $x : \mathtt{s}_{\tau}; \cdot \vdash F : \text{Formula}$. Then
    \begin{equation}
        \label{eq:pusolving}
        \vDash \bigwedge_{b \in B(e)}  \bigwedge_{S \in S_{b}} \forall \mathcal{Y}_{b}.S(\simpl(\cn(b)\wedge \conj(S) \Rightarrow F\fsub{x}{\rho(b)}))
    \end{equation}
if and only if $e \vDash F$.
\label{prop:simplsolving}
\end{restatable}
Here we state and prove some results which help connect points in programs to points in their constraints.
\begin{lemma}
    For any interpretation $\interp$, if $\con{S}{\assn}{{\uset}}$, then
    \[ \interp , \assn \vDash \conj(S).\]
    \label{prop:orderformula}
\end{lemma}
\begin{proof}
    Straightforward from the construction of $U_{\vset}(M,d)$ and the definition of $\con{S}{\assn}{\uset}$.
    \qed
\end{proof}

\begin{lemma}
    \label{prop:dmarks}
    Suppose $D : b \eval_{\vset} \v$. Then $M(D) = \{r\pnt \mid \emark{a}{b_1}{b_2}\# s \eval r\pnt\ \text{occurs in}\ D\}\cup M(b)$.
\end{lemma}
\begin{proof}
    Let $M(D)' \triangleq \{r\pnt \mid \emark{a}{e_1}{e_2}\# s\eval r\pnt\ \text{occurs in}\ D\}\cup M(e)$.
    This argument proceeds by induction on $D$.
    \begin{description}
        \item[\rname{E-Val}]
            First note that $\{ r \mid r\pnt\ \text{appears in}\ \v\} = M(\v)$. And then observe that if $D : \v \eval \v$, then
            $\{ r \mid r\pnt\ \text{appears in}\ \v\} = \{ r \mid r\pnt\ \text{appears in}\ \v, \v\ \text{appears in}\ D\}$. Therefore, $M(D) = M(\v) = M(e)$ in this case.

        \item[\rname{E-Ops}]
            Then $e  = o(e_1 ,\ldots ,e_{n})$, $D_{i} : e_{i} \eval \v_{i}$, and $\v = \sem{o}(\v_1,\ldots ,\v_{n})$ for some $o \in O$.
            Inspection of the allowed operations enables us to claim that $r\pnt$ only occurs in $\v$ if it occurs also in $\v_{i}$ for some $i$.
            Therefore, $M(D) = M(D_1)\cup \cdots M(D_{n})$. Also, $M(D)' =M(D_1)'\cup \cdots \cup M(D_{n})'$. By induction, $M(D_{i}) = M(D_{i})'$, concluding this case.

        \item[\rname{E-Mark}] Then $e = \emark{a}{e_1}{e_2} \# s $,  $\v = r\pnt$, $D_1 : e_1 \eval \rd{[r_1,r'_{1}]}$, $D_2 : e_2 \eval \vset_{a'}(P)$, and $\vset_{a}[r_1,r] = \vset_{a'}(P)$. Let $\pnt(P)$ be the set of points within $P$.
            Now $M(D) = M(D_1)\cup M(D_2)\cup \{ r_1\pnt, r'_1\pnt, r\pnt\}\cup \pnt(P)$. Since $D_2 : e_2 \eval \vset_{a'}(P)$, $\pnt(P) \subseteq M(D)'$. Since $D_1 : e_1 \eval \rd{[r_1, r'_{1}]}$, we have $r_1, r'_{1}\in M(D_1)$. Therefore, $M(D)=  M(D_1)\cup M(D_2)\cup \{r\pnt\}$.
            We also have $M(D)' = M(D_1)'\cup M(D_2)'\cup \{r\pnt\}$.
            By induction, $M(D_1) = M(D_1)'$ and $M(D_2) = M(D_2)'$. This completes the case.

                \item[\rname{E-EvalPc}]
            Then $e = \eeval{a}{e'}$, $\v = \vset_{a}(P)$, and $D' :e' \eval P$. Let $\pnt(P)$ be the set of points in $P$. Then $M(D) = M(D')\cup \pnt(P)$. Since $D' : e' \eval P$, $\pnt(P) \subseteq M(D')$. Therefore $M(D) = M(D')$. Now $M(D)' = M(D')'$. By induction, $M(D')' = M(D')$, so we can conclude this case.
        \item[\rname{E-EvalIntvl}] The argument is very similar to the argument for \rname{EvalPc}.

        \item[\rname{E-Div}]
            Then $e = \ediv{e_1}{e_2}$, $\v = ([r_1, r_2], [r_2, r_1'])$, $D_1 : e_1 \eval [r_1, r'_{1}]$, and $D_2 : e_2 \eval r_2\pnt$. Then $M(D) = M(D_1)\cup M(D_2)\cup \{r_1, r_1', r_2\}$. Since $D_1 : e_1 : \eval [r_1, r_1']$, $\{r_1, r_1'\}\subseteq M(D_1)$, and as $D_2 : e_2 \eval r_2\pnt$, we have $r_2\pnt \in M(D_2)$. Therefore $M(D) = M(D_1) \cup M(D_2)$.
            Now $M(D)' = M(D_1)' \cup M(D_2)'$. By induction, $M(D_1) = M(D_1)'$ and $M(D_2) = M(D_2)'$. This case is then complete.
        \item[\rname{E-Split}]
            Then $e = \kwlet\ x_1 ,\ldots ,x_{n}, w_1,\ldots ,w_{n'} = \kwsplit\ e_1\ \kwin\ e_2$, $D_1 : e_1 \eval (\v_1 ,\ldots ,\v_{n + n'})$, $D_2 : S_{\var}S_{\rd{\lvar}}S_{\lvar}(e_2) \eval \v$, where
            \begin{align*}
                S_{\var} &\triangleq \{x_{i} \mapsto \v_{i}\mid 1 \leq i \leq n\} \\
                S_{\rd{\lvar}} &\triangleq \{\rd{w_{i}} \mapsto \rd{\v_{n +i}}\mid 1 \leq i \leq n'\} \\
                S_{\var} &\triangleq \{w_{i} \mapsto \v_{n + i}\mid 1 \leq i \leq n'\}.
            \end{align*}
            Now $M(D) = M(D_1)\cup M(D_2)$, and $M(D)' = M(D_1)'\cup M(D_2)'\cup M(e_2)$. Since $M(D_2)'\supseteq M(S_{\var}S_{\rd{\lvar}}S_{\lvar}(e_2))$ and $M(S_{\var}S_{\rd{\lvar}}S_{\lvar}(e_2))\supseteq M(e_2)$, we have $M(D)' = M(D_1)' \cup M(D_2)'$. By induction, $M(D_1) = M(D_1)'$ and $M(D_2) = M(D_2)'$, completing this case.
    \qed
    \end{description}
\end{proof}

\begin{lemma}
    Suppose that $D : b \eval \v$. Suppose that $\assn$ is a variable assignment that agrees with $D$. Then $M(D) \subseteq M(b)\cup \assn( \{y_{\#s}\in \mathcal{Y}\mid \# s \in \id(b)\})$.
    \label{prop:mpoints}
\end{lemma}
\begin{proof}
    According to \Cref{prop:dmarks}, $M(D) = M(b) \cup \{r \mid \emark{a}{b_1}{b_2}\#s \eval r\pnt\ \text{occurs in}\ D \}$. Therefore we just have to show
    \[
\{r \mid \emark{a}{b_1}{b_2}\#s \eval r\pnt\ \text{occurs in}\ D \} \subseteq  \assn( \{y_{\#s}\in \mathcal{Y}\mid \# s \in \id(b)\}).
    \]
    So let $r$ be in the left hand side. Then $\emark{a}{b_1}{b_2}\# s \eval r\pnt$ occurs in $D$. We have $\#s \in \id(b)$ and since $\assn$ agrees with $D$, we have that $\assn(y_{\#s}) = r\pnt$. Therefore $r\pnt$ is in the right hand side. This completes the argument.
    \qed
\end{proof}

\begin{lemma}
    \label{prop:yvarids}
    For any path $b$, $\fv(\rho(b))\cap \mathcal{Y} \subseteq \fv(\cn(b))\cap \mathcal{Y}$, and $\fv(\cn(b))\cap \mathcal{Y} = \{ y_{\# s} \in \mathcal{Y} \mid \# s \in \id(b)\}$.
\end{lemma}
\begin{proof}
    Straightforward. \qed
\end{proof}

\begin{lemma}
    Let $\pnt(b)$ be the set of point values contained in $\cn(b)$ and $\rho(b)$. Then $\pnt(b) \subseteq M(b)$.
\end{lemma}
\begin{proof}
    This argument goes by induction on the structure of $b$ and is straightforward.
    \qed
\end{proof}

\begin{proof}[\Cref{prop:simplsolving}]
    ($\Rightarrow$)
    Suppose that $e \eval_{\vset} \v$. By \Cref{prop:path}, there is $b \in B(e)$ such that $D : b \eval_{\vset} \v$.

    Let $\assn'$ be a variable assignment that agrees with $D$ and let $S$ be any piecewise uniform substitution on $\mathcal{Y}_{b}$ whose order is such that if $\assn'(y) < \assn'(y')$, then $y <_{S} y'$.
    By \Cref{prop:mpoints}, $\mathcal{Y}_{b} = \{y_{\# s}\mid \#s \in \id(b)\}$, so by \Cref{prop:yvarids}, $M(D) \subseteq \assn'(\mathcal{Y}_{b})$. Let $M = \assn'(\mathcal{Y}_{b})$.
    By \Cref{prop:puexists}, there is a valuation set $\uset$ that agrees with $\vset$ on $M$ and is easily replaceable on $M$.
    By \Cref{thm:red}, $D : b \eval_{\uset} \v$.
    Now let $\assn$ be any assignment that agrees with $\assn'$ on all of $\mathcal{Y}_{b}$, but assigns $z_{a,y}$ for all $a \in \agents$ and $y\in S$ such that $\con{S}{\assn}{\uset}$. This is possible since $\mathcal{Z}$ is disjoint from $\mathcal{Y}$ and for each $a$ and $y$, there is a unique $z_{a,y}\in \mathcal{Z}$.
Let $\interp = \interp_{\uset}$.
    By \Cref{prop:soundness}, $\sem{\rho(b)}_{\interp}^{\assn} = \unrd{\v}$ and $\interp, \assn \vDash \cn(b)$.
    By \Cref{prop:simpl}, $\sem{\simpl(\rho(b))}_{\interp}^{\assn} = \unrd{\v}$ and $\interp,\assn\vDash \simpl(\cn(b))$. By \Cref{prop:orderformula}, $\interp, \assn \vDash \conj(S)$. Then we have
    \begin{align*}
        &\interp,\assn \vDash S(\simpl(\cn(b)\land \conj(S))) \tag{\Cref {prop:purepsound}} \\
        &\interp, \assn \vDash S(\simpl (F\fsub{x}{\rho(b)})) \tag{\Cref{eq:pusolving}} \\
        &\interp \vDash \simpl (F\fsub{x}{\rho(b)}) \tag{\Cref {prop:purepsound}} \\
        &\interp \vDash F\fsub{x}{\rho(b)} \tag{\Cref{prop:simpl}}
    \end{align*}
    and by \Cref{prop:semsub} as $\sem{\rho(b)}_{\interp}^{\assn} = \unrd{\v}$,
    \[
        \interp \vDash F\fsub{x}{\lv}.
    \]
    Now $\cdot \vdash F\fsub{x}{\lv} : \text{Formula}$ so that $\cdot \vdash F : \text{Formula}$. Since $F$ is well-formed, $F$ contains no point values. Therefore all point values are contained within $\lv$. This allows us to apply \Cref{thm:pufequiv} to obtain
    \[
        \interp_{\vset} \vDash F\fsub{x}{\lv}.
\]

($\Leftarrow$)
We approach this by assuming \Cref{eq:pusolving} does not hold and constructing a piecewise uniform valuation set witness to $e \not\vDash F$.
So suppose that \Cref{eq:pusolving} does not hold. Observe that \Cref{eq:pusolving} does not contain any $\val$ function symbols. Then there is some $\assn$, $b \in B(e)$, and $S\in S_{b}$ such that $\assn \vDash S(\simpl(\cn(b)))\land \conj(S)$ yet $ \assn \not\vDash F\fsub{x}{\rho(b)}$.
Write out $S = \{y_1 ,\ldots ,y_{n}\}$ such that $y_1 <_{S} \cdots <_{S} y_{n}$.
Consider the piece
\[
    [\assn(z_{a,y_{1}}), \assn(y_{1})]\cup \cdots \cup [\assn(z_{a,y_{n}}), \assn(y_{n})].
\]
As $\assn\vDash \conj(S)$,
\begin{equation}
\assn(z_{a,y_{1}}) \leq \assn(y_{1}) \leq \cdots \leq \assn(z_{a,y_{n}}) \leq \assn(y_{n}).
\label{eq:witnessineq}
\end{equation}
Then let $U_{a}$ be the piecewise uniform valuation on the above piece.
Also observe by $\assn \vDash \conj(S)$, for any $a, a'$,
\[
    \sum_{y \in S} \assn(y) - \assn(z_{a,y}) = \sum_{y \in S} \assn(y) - \assn(z_{a',y}).
\]
Therefore we can set $d$ to be the reciprocal of the above sum, and obtain that $d = c(U_{a})$ for all $a$.
Then $\uset \triangleq (a \mapsto U_{a}\mid a \in \agents)$ is a piecewise uniform valuation set which is easily replaceable on $\assn(S)$. Set $\interp = \interp_{U}$.

We can verify that $\con{S}{\assn}{U}$:
(1) is satisfied directly above.
(2) is also satisfied by definition.
(3) is mildly more difficult. Let $y \in \mathcal{Y}_{b}$ be such that there is $y' <_{S}y$ and $\assn(y') = \assn(y)$. Then $\assn(y') \leq \assn(z_{a,y})\leq \assn(y)$ hence $\assn(z_{a,y}) = \assn(y)$ for all $a \in \agents$.
(4) is verified by \Cref{eq:witnessineq} recalling that $y_1 <_{S} \cdots <_{S} y_{n}$.

Then by \Cref{prop:purepsound},
\[
    \interp,\assn\vDash \simpl(\cn(b))\quad \text{and}\quad \interp,\assn \not\vDash \simpl(F\fsub{x}{\rho(b)}).
\]
By \Cref{prop:simpl},
\[
    \interp,\assn\vDash \cn(b)\quad \text{and}\quad \interp,\assn \not\vDash F\fsub{x}{\rho(b)}.
\]
Therefore by \Cref{prop:completeness}, $b \eval_{\uset} \v$ where $\unrd{\v} = \sem{\rho(b)}_{\interp}^{\assn}$. By \Cref{prop:path}, $e \eval_{\uset} \v$. By \Cref{prop:semsub}, $\interp,\assn \not\vDash F\fsub{x}{\lv}$. Since $F\fsub{x}{\lv}$ has no unbound variables, $\assn$ is redundant. Therefore $\interp \not\vDash F\fsub{x}{\lv}$. This shows that $e \not\vDash F$.
\qed
\end{proof}

We also have the following generalization of \Cref{cor:decidable}
\begin{restatable}[]{corollary}{decidable}
    Let $e$ be a well-formed protocol, only using standard operations, and let $F$ a well-formed property. Checking if $e$ satisfies $F$ is decidable.
\end{restatable}

\subsection{Evaluation of protocols with envy}
\label{app:exp}

\subsubsection{Non-Envy Free Protocols.}
To demonstrate the proficiency of our approach for finding envy, we deliberately incorporated errors into select protocols and time how long it takes to find envy, following methodology established by Lester~\cite{Lester}. The results, contained in \Cref{tab:experiment2}, demonstrate that our approach can also efficiently find envy in all protocols tested.

\begin{table}[htbp]
    \centering
    \caption{Time to find counterexamples to non-envy-free protocols.}
    \begin{tabular}{@{}lrl@{}}
      \toprule
      {Protocol} & {Time} & {Description} \\
      \midrule
      Cut-Choose & 0.020 & Agent 1 cuts and chooses. \\
      Cut-Choose & 0.021 & Slices allocated wrong way round in one branch. \\
      Surplus    & 0.035 & Unsafe trim of slice smaller than reference. \\
      Selfridge-Conway-Surplus        & 0.034 & Allocates trimmings of slice instead of trimmed slice. \\
      Selfridge-Conway-Surplus        & 0.033 & Agent 2 not forced to take trimmed slice if available. \\
      Selfridge-Conway-Full        & 0.161 & Trimmings cut by agent who took trimmed slice. \\
      Aziz-Mackenzie-3 & 2.594 & Did not check if Agent 1 and 2 has the same favorite piece. \\
      \bottomrule
    \end{tabular}
    \label{tab:experiment2}
  \end{table}

\fi
\end{document}